\documentclass[11pt,letterpaper]{article}

\usepackage{amsmath}
\usepackage{amsthm}
\usepackage{amssymb}
\usepackage{algorithm,algpseudocode}
\usepackage{color}
\usepackage[english]{babel}

\usepackage{graphicx}
\usepackage{caption}
\usepackage{subcaption}

\usepackage{enumitem}
\usepackage{float}

\usepackage[square,sort,comma,numbers]{natbib}
\usepackage{url}
\usepackage{algpseudocode}
\usepackage[pdfencoding=auto,psdextra]{hyperref}
\usepackage{xcolor}
\usepackage{hyperref}

\definecolor{ZhuQing}{HTML}{3B7D56}     
\definecolor{ZhuShen}{HTML}{326B49}     
\definecolor{ZhuAn}{HTML}{28573B}       
\definecolor{ZhuTai}{HTML}{4B573E}      
\definecolor{ZhuYin}{HTML}{BE4A3C}      
\definecolor{ZhuYinDark}{HTML}{8C3630}  
\definecolor{YanZhi}{HTML}{9B3A5C}      
\definecolor{ShiQing}{HTML}{4A7FA0}     
\definecolor{ShiQingDark}{HTML}{3A6480} 
\definecolor{TengHuang}{HTML}{7A5C15}   

\hypersetup{
  colorlinks = true,
  linkcolor  = ZhuYinDark,  
  citecolor  = ZhuShen,      
  urlcolor   = ShiQingDark, 
  filecolor  = ZhuTai,
  unicode    = true
}

\usepackage[colorinlistoftodos]{todonotes}

\usepackage{tikz}
\hypersetup{
  pdftitle={Sparsify Sums of Submodular Functions}
}
\usetikzlibrary{arrows.meta,positioning,fit,backgrounds,calc}
\usepackage[margin=1in]{geometry}
\newtheorem{theorem}{Theorem}[section]
\newtheorem{lemma}[theorem]{Lemma}

\newtheorem{definition}[theorem]{Definition}

\newtheorem{proposition}[theorem]{Proposition}

\newtheorem{claim}[theorem]{Claim}
\newtheorem{example}[theorem]{Example}

\DeclareMathOperator{\cut}{\mathrm{cut}}

\newcommand{\wh}{\widehat}

\newcommand{\supp}{\mathsf{supp}}

\DeclareMathOperator*{\argmin}{arg\,min}
\DeclareMathOperator*{\E}{\mathbb{E}}

\newcommand{\R}{\mathbb{R}}

\newenvironment{proofof}[1]{\bigskip \noindent {\it Proof of #1.}\quad }
{\qed\par\vskip 4mm\par}

\title{Nearly Optimal Submodular Sparsification,\\
with and without Cardinality Constraints}
\author{
  Zhengting Bao\thanks{
    \texttt{bztminamoto@mail.ustc.edu.cn},
    University of Science and Technology of China.
  }
  \and
  Dongrun Cai\thanks{
    \texttt{cdr@mail.ustc.edu.cn},
    University of Science and Technology of China.
  }
  \and
  Xue Chen\thanks{\tt{xuechen1989@ustc.edu.cn}, University of Science and Technology of China, Hefei 230026, China and Hefei National Laboratory, Hefei 230088, China. Supported by NSFC 62372424 and Quantum Science and Technology-National Science and Technology Major Project 2021ZD0302901.}
  \and
  Yuchen He\thanks{
    \texttt{yuchen\_he@sjtu.edu.cn},
    Shanghai Jiao Tong University.
  }
  \and
  Chihao Zhang\thanks{
    \texttt{chihao@sjtu.edu.cn},
    Shanghai Jiao Tong University. Supported by NSFC 62372289.
  }
}
\date{}

\begin{document}

\pagenumbering{roman}
\maketitle
\thispagestyle{empty}

\begin{abstract}
We study the sparsification of decomposable submodular functions: given a sum $F(S) = \sum_{i=1}^m f_i(S)$ of $m$ nonnegative submodular functions on an $n$-element ground set $[n]$ (without assuming monotonicity or symmetry), the goal is to find a sparse reweighting $\widetilde{F}(S) = \sum_{i=1}^m w_i f_i(S)$ with at most $s \ll m$ nonzero weights that preserves $F(S)$ within a $(1 \pm \varepsilon)$ factor across a target family of queries. For all subset queries $S \subseteq [n]$, we give a polynomial-time randomized algorithm that constructs a sparsifier of size $O(n^2 \log n / \varepsilon^2)$. This improves upon the previous best upper bound of $O(n^3 / \varepsilon^2)$ by Kenneth and Krauthgamer~\cite{KK24} and resolves their open conjecture up to a logarithmic factor, matching the $\Omega(n^2)$ lower bound from directed cuts. Motivated by capacity budgets in submodular applications and the Restricted Isometry Property (RIP) in compressed sensing, we further study submodular sparsification under a cardinality constraint $|S| \le k$, and obtain a polynomial-time algorithm that constructs a sparsifier of size $O(n k \log n / \varepsilon^2)$, nearly matching the $\Omega(nk)$ lower bound. This shows that a cardinality constraint unconditionally reduces the size of general submodular sparsifiers by a factor of roughly $k/n$; we complement this with $\Omega(n)$ lower bounds showing that even regular symmetric or almost-regular monotone submodular functions cannot be sparsified to $o(n)$ size for $k = 2$ and $k = 1$, respectively.

Technically, to overcome an $\Omega(n^3)$ sampling bottleneck of classical sensitivity and pairwise min-cut scores from Kenneth and Krauthgamer \cite{KK24}, we introduce \emph{directed block min-cut functions} $f_i^{\to}(A, B) := \underset{A \subseteq T \subseteq [n] \setminus B}{\min} f_i(T)$---which are monotone and submodular in each argument even when $f_i$ is not---and define importance weights by summing weights from $2n$ \emph{local Gibbs distributions} on the single-terminal families $f^{\to}(u, \cdot)$ and $f^{\to}(\cdot, v)$ with homogeneous total budget $O(n^2)$. Exploiting the monotonicity of $f_i^{\to}$, we decompose each query across scales into a telescoping sum of nonnegative increments (\emph{positive chaining}), establishing uniform concentration across all queries via Bernstein's inequality and a union bound without Rademacher symmetrization or Dudley's entropy integral. Finally, cooling these local Gibbs distributions to zero temperature collapses them to deterministic ground states, derandomizing the weight computation, connecting Gibbs sampling to classical minimum-ratio strength decompositions, and enabling us to compress the total sampling weight from $O(n^2)$ to $O(nk)$ under $|S| \le k$ via capped decompositions and low-collision partitions.

\end{abstract}

\clearpage
\tableofcontents
\clearpage
\pagenumbering{arabic}

\section{Introduction}
\label{sec:introduction}

Sparsification is a fundamental paradigm for data compression in computer science: given a large mathematical object composed of many simple components, the goal is to select and reweight a small subset of components so that a target structural property is preserved across all queries of interest. Two classic incarnations of this paradigm are \emph{graph sparsification} and \emph{matrix sparsification}. In graphs, the seminal work of Bencz\'{u}r and Karger~\cite{BK96} showed that any undirected graph on $n$ vertices can be compressed into a reweighted subgraph with $\widetilde{O}(n/\varepsilon^2)$ edges that preserves the value of every cut within a $(1 \pm \varepsilon)$ factor; this breakthrough led to linear-sized spectral sparsifiers~\cite{ST04, batson2012twice}, extensions to hypergraphs~\cite{KK15, BST19, CKN20, KapralovKTY21a} and codes~\cite{KhannaPS24, brakensiek2025redundancy, hsieh2026sparsifying}, and a host of algorithmic applications~\cite{spielman2011graph, Sherman13, KLOS14, CSWZ16}. In matrices, row sparsification compresses a tall matrix $A \in \mathbb{R}^{m \times n}$ ($m \gg n$) into a reweighted subset of rows $A'$ while preserving the norm $\|A'x\|_p \approx \|Ax\|_p$. When this guarantee is required for \emph{all} vectors $x \in \mathbb{R}^n$---known as an $\ell_p$ subspace embedding (which generalizes spectral graph sparsification at $p=2$)---classical results achieve a nearly optimal row count of $\widetilde{O}(n)$ for all $p \in [1, 2]$~\cite{BLM89, Talagrand90, CP15_Lewis}. Crucially, when the query set is restricted from all vectors in $\mathbb{R}^n$ to \emph{$k$-sparse vectors} ($\|x\|_0 \le k$)---the celebrated \emph{Restricted Isometry Property (RIP)} in compressed sensing and sparse recovery~\cite{Donoho_compressed_sensing, CT05, RV}---the required number of rows for structured matrices such as Fourier and Hadamard matrices drops dramatically from $\Theta(n)$ to $k (\log n)^{O(1)}$~\cite{CT05, RV, CheraghchiGV13, HR16}.

In discrete optimization, machine learning, and algorithmic game theory, a natural combinatorial generalization of graph and hypergraph cuts is a sum of submodular functions, known as a \emph{decomposable submodular function}~\cite{SK10, LHM20, pmlr-v139-axiotis21a, RafieyY22submodular}. Identifying Boolean vectors in $\{0,1\}^n$ with subsets of $[n] := \{1, \dots, n\}$, a set function $f : \{0,1\}^n \to \mathbb{R}_{\ge 0}$ is \emph{submodular} if it exhibits diminishing marginal returns: $f(S) + f(T) \ge f(S \cup T) + f(S \cap T)$ for all $S, T \subseteq [n]$. Given $m$ nonnegative submodular functions $f_1, \dots, f_m : \{0,1\}^n \to \mathbb{R}_{\ge 0}$ on a common ground set $[n]$, consider their sum $F(S) := \sum_{i=1}^m f_i(S)$. 
In \emph{submodular sparsification}, given an error parameter $\varepsilon \in (0, 1)$, the goal is to find a sparse nonnegative weight vector $w \in \mathbb{R}_{\ge 0}^m$ with at most $s \ll m$ nonzero entries such that the reweighted sum $\widetilde{F}(S) := \sum_{i=1}^m w_i f_i(S)$ satisfies
\begin{equation}\label{eq:def_sparsification}
    \widetilde{F}(S) = (1 \pm \varepsilon) \cdot F(S) \qquad \forall S \in \{0,1\}^n.
\end{equation}
Once constructed, a size-$s$ sparsifier $\widetilde{F}$ compresses the representation of $F$ from $m$ components to $s$ components and accelerates every subsequent query or optimization step by a factor of $m/s$.

A central theme in recent research~\cite{RafieyY22submodular, KZ23, JLLS23, KK24, KPS24, Qua24} has been to determine the optimal size $s$ of submodular sparsifiers. When the components $f_i$ are \emph{symmetric} ($f_i(S) = f_i(\overline{S})$ with $f_i(\varnothing) = 0$, as in undirected graph and hypergraph cuts), their continuous Lov\'{a}sz extensions~\cite{Lovasz83} define \emph{seminorms} on $\mathbb{R}^n$. Exploiting this bridge between submodular functions and matrix subspace embeddings, Jambulapati, Lee, Liu, and Sidford~\cite{JLLS23} proved that any sum of seminorms---and hence any sum of symmetric submodular functions---admits a sparsifier of nearly linear size $\widetilde{O}(n/\varepsilon^2)$; Khanna, Putterman, and Sudan~\cite{KPS24} subsequently extended this $\widetilde{O}(n/\varepsilon^2)$ bound to \emph{monotone} submodular functions ($f_i(S) \le f_i(T)$ for $S \subseteq T$) via a reduction to the symmetric case. However, general submodular functions are \emph{neither symmetric nor monotone}: their Lov\'{a}sz extensions are no longer seminorms, and dropping these assumptions fundamentally changes the complexity of sparsification. Even the simplest asymmetric submodular function---the directed cut function $\operatorname{cut}_{H \to T}(S)$, which equals $1$ if $S \cap H \ne \varnothing$ and $\overline{S} \cap T \ne \varnothing$, and $0$ otherwise (abbreviated $\operatorname{cut}_{a \to b}$ for a single directed edge $(a,b)$)---breaks the $\widetilde{O}(n)$ barrier: sparsifying directed cuts requires $\Omega(n^2)$ components~\cite{CohenKPPRSV17, OST23}, in stark contrast to the $\widetilde{O}(n)$ bound for undirected cuts $\operatorname{cut}_e(S) := \mathbf{1}[\varnothing \subsetneq S \cap e \subsetneq e]$~\cite{BK96, batson2012twice, BST19}. At the same time, general non-monotone, asymmetric submodular sums $F = \sum_{i=1}^m f_i$ with $m \gg n^2$ arise pervasively across applications:

\begin{example}[Directed hypergraphs and distributed/multi-layer subgraph cuts]\label{ex:sum_directed_cuts}
In a directed hypergraph on $n$ vertices, each component $f_i = \operatorname{cut}_{H_i \to T_i}$ is a directed hyperedge cut, and the number of hyperedges $m$ can be as large as $3^n$. More generally, because any nonnegative linear combination of submodular functions remains submodular, a single component $f_i(S) := \sum_{(H, T) \in E_i} u_{(H, T)} \cdot \operatorname{cut}_{H \to T}(S)$ can represent the cut function of an entire weighted directed subgraph (or hypergraph) $E_i$. For instance, consider a massive directed network distributed across $m$ servers (or evolving across $m$ temporal layers), where server $i$ holds a subgraph $E_i$. Sparsifying $F = \sum_{i=1}^m f_i$ corresponds to selecting a small fraction of $s \ll m$ servers whose reweighted subgraphs simultaneously approximate the directed cut values and connectivity of the entire network.
\end{example}

\begin{example}[Social welfare and net utility aggregation]\label{ex:welfare}
In algorithmic game theory and economics, consider $m$ agents with preferences over $n$ resources or public projects~\cite{RafieyY22submodular}. Each agent $i \in [m]$ has a submodular valuation $f_i : \{0,1\}^n \to \mathbb{R}_{\ge 0}$ capturing diminishing marginal returns—and when resources carry usage costs, congestion penalties, or diversity trade-offs, $f_i$ is naturally non-monotone as well as asymmetric. The sum $F(S) = \sum_{i=1}^m f_i(S)$ represents the aggregate social welfare of bundle $S$, and a size-$s$ sparsifier $\widetilde{F}$ selects a small weighted committee of $s \ll m$ representative agents whose collective valuation approximates the social welfare of the entire population on every subset of resources.
\end{example}

Despite this $\Omega(n^2)$ lower bound, the best known upper bound for sparsifying general submodular functions in~\eqref{eq:def_sparsification} was $O(n^3/\varepsilon^2)$, proved by Kenneth and Krauthgamer~\cite{KK24}. Closing this $O(n)$-factor gap between $\Omega(n^2)$ and $O(n^3)$---and resolving whether near-quadratic sparsifiers of size $\widetilde{O}(n^2/\varepsilon^2)$ exist for arbitrary nonnegative submodular functions, as conjectured by Kenneth and Krauthgamer~\cite[Conjecture~1.14]{KK24}---is the first central question of this work.

Our second central question is inspired by the role of sparsity constraints in both submodular applications and compressed sensing: \emph{Can we break the $\Omega(n^2)$ barrier when queries are restricted to subsets of cardinality at most $k$?} Let $\binom{[n]}{\le k} := \{S \subseteq [n] : |S| \le k\}$ denote the family of subsets of size at most $k$. Given $F = \sum_{i=1}^m f_i$, an error parameter $\varepsilon \in (0, 1)$, and a cardinality budget $k \in [n]$, a size-$s$ sparsifier of $F$ \emph{under the cardinality constraint $|S| \le k$} is a nonnegative weight vector $w \in \mathbb{R}_{\ge 0}^m$ with at most $s$ nonzero entries such that $\widetilde{F}(S) := \sum_{i=1}^m w_i f_i(S)$ satisfies
\begin{equation}\label{eq:def_sparsification_cardinality}
    \widetilde{F}(S) = (1 \pm \varepsilon) \cdot F(S) \qquad \forall S \in \binom{[n]}{\le k}.
\end{equation}
Studying~\eqref{eq:def_sparsification_cardinality} is compelling from two complementary viewpoints:
\begin{itemize}[leftmargin=*]
    \item \textbf{Cardinality budgets in applications}: In downstream tasks, one rarely evaluates $F(S)$ on arbitrary subsets of $[n]$; instead, feasible queries are almost always governed by a capacity or budget $k \ll n$. In social welfare aggregation (Example~\ref{ex:welfare}) and data summarization, a decision-maker typically has a budget to select at most $k$ items out of $n$, so the representative committee only needs to preserve welfare across feasible bundles $S \in \binom{[n]}{\le k}$ (which also subsumes sparsification under matroid constraints of rank $k$, as noted by Rafiey and Yoshida~\cite{RafieyY22submodular}). Likewise, in directed networks and hypergraphs (Example~\ref{ex:sum_directed_cuts}), restricting to $|S| \le k$ corresponds to preserving \emph{unbalanced cuts} and small-set expansion around sets of size at most $k$.
    \item \textbf{A discrete analogue of the Restricted Isometry Property}: Viewing $S \in \binom{[n]}{\le k}$ as a $k$-sparse Boolean vector $x \in \{0,1\}^n$ with $\|x\|_0 \le k$, cardinality-constrained submodular sparsification~\eqref{eq:def_sparsification_cardinality} is the direct discrete analogue of RIP in matrix sparsification~\cite{Donoho_compressed_sensing, CT05, RV}. Just as restricting from all $x \in \mathbb{R}^n$ to $k$-sparse vectors reduces the required number of rows in RIP, it is natural to ask whether restricting to $|S| \le k$ can systematically reduce the size of submodular sparsifiers below the $\Omega(n^2)$ barrier.
\end{itemize}

\subsection{Main results}

Our main results are polynomial-time algorithms that construct submodular sparsifiers of nearly optimal size for both the unrestricted setting~\eqref{eq:def_sparsification} and the cardinality-constrained setting~\eqref{eq:def_sparsification_cardinality}, complemented by matching and near-matching lower bounds.

Our first main theorem resolves the open question of Kenneth and Krauthgamer~\cite[Conjecture~1.14]{KK24} up to a logarithmic factor by improving the upper bound for general submodular sparsification from $O(n^3/\varepsilon^2)$ to $\widetilde{O}(n^2/\varepsilon^2)$.

\begin{theorem}[Informal version of Theorem~\ref{thm:near-quadratic-submodular-sparsification}]\label{inform:sparsification_all_subsets}
    Given any sum $F(S) = \sum_{i=1}^m f_i(S)$ of nonnegative submodular functions $f_1, \dots, f_m : \{0,1\}^n \to \mathbb{R}_{\ge 0}$ and any error parameter $\varepsilon \in (0, 1)$, there is an efficient randomized algorithm that constructs with high probability a sparsifier $\widetilde{F}$ of size $O\!\left(\frac{n^2 \log n}{\varepsilon^2}\right)$ satisfying~\eqref{eq:def_sparsification} for all $S \subseteq [n]$.
\end{theorem}

Several remarks are in order. First, the sparsity bound $O(n^2 \log n / \varepsilon^2)$ is optimal up to the $O(\log n)$ factor in view of the $\Omega(n^2)$ lower bound for directed cut functions~\cite{CohenKPPRSV17, OST23}, and requires neither symmetry nor monotonicity of the summands $f_i$. In particular, for natural families of asymmetric or non-monotone submodular functions---such as weighted directed subhypergraph cuts in Example~\ref{ex:sum_directed_cuts} and general social welfare functions in Example~\ref{ex:welfare}---the best previously known upper bound was $O(n^3/\varepsilon^2)$~\cite{KK24}; Theorem~\ref{inform:sparsification_all_subsets} yields the first near-quadratic sparsifiers for these families. Second, our algorithm computes deterministic importance-sampling probabilities in polynomial time using an oracle for submodular function minimization. Specifically, let $\mathcal{T}_{\min}$ denote the time complexity of minimizing a submodular function on an $n$-element ground set with a unit-cost value oracle (where $\mathcal{T}_{\min} = O(n^4 \log n)$ strongly polynomially~\cite{LSS15, JLSZ24}), and let $\mathcal{T}_{\mathrm{eval}}$ denote the time of a single value-oracle query to $f_i$. Then the deterministic preprocessing step runs in time $O(mn^3 (\mathcal{T}_{\mathrm{eval}} + n) \mathcal{T}_{\min}^2)$ (which simplifies to $O(mn^4 \mathcal{T}_{\min}^2)$ when $\mathcal{T}_{\mathrm{eval}} = O(n)$), after which each component is sampled independently once. We also give an alternative randomized construction based on sampling from local Gibbs distributions in Appendix~\ref{sec:Gibbs}.

Our second main theorem shows that imposing a cardinality constraint $|S| \le k$ always reduces the size of general submodular sparsifiers from $\widetilde{O}(n^2/\varepsilon^2)$ to $\widetilde{O}(nk/\varepsilon^2)$.

\begin{theorem}[Informal version of Theorem~\ref{thm:small-query-submodular-sparsification}]\label{inform:sparsification_cardinality_constraint}
    Given any sum $F(S) = \sum_{i=1}^m f_i(S)$ of nonnegative submodular functions $f_1, \dots, f_m : \{0,1\}^n \to \mathbb{R}_{\ge 0}$, any error parameter $\varepsilon \in (0, 1)$, and any cardinality budget $k \in [n]$, there is an efficient randomized algorithm that constructs with high probability a sparsifier $\widetilde{F}$ of size $O\!\left(\frac{n k \log n}{\varepsilon^2}\right)$ satisfying~\eqref{eq:def_sparsification_cardinality} for all $S \in \binom{[n]}{\le k}$.
\end{theorem}

The sparsity bound $O(nk \log n / \varepsilon^2)$ in Theorem~\ref{inform:sparsification_cardinality_constraint} is likewise nearly tight: in Section~\ref{sec:lower_bounds} (Proposition~\ref{clm:lower_bound_arbitrary}), we show that the directed-cut construction implies an $\Omega(nk)$ lower bound on the support size of any reweighted sparsifier satisfying~\eqref{eq:def_sparsification_cardinality}. The running time of Theorem~\ref{inform:sparsification_cardinality_constraint} incurs only an extra factor of $k$ over Theorem~\ref{inform:sparsification_all_subsets}, namely $O(mn^3 k (\mathcal{T}_{\mathrm{eval}} + n) \mathcal{T}_{\min}^2)$.

Theorem~\ref{inform:sparsification_cardinality_constraint} reveals an interesting contrast with classical graph and matrix sparsification. For general submodular functions, restricting queries to $|S| \le k$ \emph{unconditionally} reduces the required sparsifier size by a factor of roughly $k/n$ (from $\widetilde{O}(n^2)$ to $\widetilde{O}(nk)$); for instance, in Example~\ref{ex:sum_directed_cuts}, preserving only $k$-unbalanced directed cuts across distributed subgraphs requires $\widetilde{O}(nk)$ components rather than $\widetilde{O}(n^2)$. By contrast, in undirected graph sparsification and worst-case matrix subspace embeddings, a cardinality constraint alone yields no asymptotic reduction in sparsity: an undirected graph of $n/2$ disjoint edges still requires all $\Omega(n)$ edges even to preserve cuts of size $k=1$ (to avoid isolated vertices), and an $\ell_p$ subspace embedding of the identity matrix $I_n$ restricted to $1$-sparse vectors still requires all $n$ rows.

\paragraph{Lower bounds under cardinality constraints.}
While Theorem~\ref{inform:sparsification_cardinality_constraint} shows that a cardinality constraint $|S| \le k$ reduces the sparsity of \emph{general} submodular functions from $\widetilde{O}(n^2)$ to $\widetilde{O}(nk)$, several structured families---most notably symmetric~\cite{JLLS23} and monotone~\cite{KPS24} submodular functions---already admit sparsifiers of size $\widetilde{O}(n)$ without any cardinality constraint. This raises a natural question: \emph{Under a cardinality constraint $|S| \le k$ (say for constant $k$), can symmetric or monotone submodular functions be sparsified to size $o(n)$, or even $(k \log n)^{O(1)}$ as in the Restricted Isometry Property (RIP)?}

In matrix sparsification, achieving RIP of size $k (\log n)^{O(1)}$ requires ruling out spiky rows (such as the identity matrix $I_n$) by assuming an incoherence or \emph{regularity} condition---namely, that every entry of a row has roughly uniform magnitude, as in a normalized Hadamard or Fourier matrix~\cite{RV, CheraghchiGV13, HR16, CD20}. In Section~\ref{sec:lower_bounds}, we formulate the natural submodular analogue of this condition (Definition~\ref{def:regular_submodular}): a submodular function $f$ is \emph{regular} if all singleton marginals $f(\{v\})$ ($v \in [n]$) are within a constant factor of one another, and \emph{almost-regular} if $\max_{v \in [n]} f(\{v\}) = O(1) \cdot \frac{1}{n} \sum_{v \in [n]} f(\{v\})$. Surprisingly, we prove that unlike linear measurements in RIP, neither symmetry combined with regularity, nor monotonicity combined with almost-regularity, suffices to break the $\Omega(n)$ barrier under small cardinality constraints:
\begin{itemize}[leftmargin=*]
    \item \textbf{Regular and symmetric submodular functions}: Even for $k = 2$, there exist sums of regular, symmetric submodular functions (with $f_i(\{v\}) = 1$ for all $i \in [m]$ and $v \in [n]$) for which any sparsifier preserving queries of size $|S| \le 2$ must retain $\Omega(n)$ components (Proposition~\ref{clm:lower_bound_regular}).
    \item \textbf{Almost-regular and monotone submodular functions}: Even for $k = 1$, there exist sums of almost-regular, monotone submodular functions for which any sparsifier preserving singleton queries $|S| \le 1$ must retain $\Omega(n)$ components (Proposition~\ref{clm:lower_bound_monotone}).
\end{itemize}
We defer the formal statements and proofs of these lower bounds to Section~\ref{sec:lower_bounds}.

\subsection{Techniques}

We provide a detailed technical overview in Section~\ref{sec:overview}; here we briefly highlight our main technical innovations:
\begin{itemize}[leftmargin=*]
    \item \textbf{Directed block min-cuts and $2n$ local Gibbs distributions}: As we show via an explicit $\Omega(n^3)$ counterexample in Claim~\ref{clm:counter_examples}, existing importance scores---classical sensitivity $s_i = \max_S \frac{f_i(S)}{F(S)}$, pairwise directed min-cut weights $\rho_i$~\cite{KK24}, and the global Gibbs distribution of~\cite{JLLS23}---either lack homogeneity across disjoint blocks or over-amplify components by treating vertex pairs in isolation. We introduce the \emph{directed block min-cut function} $f_i^{\to}(A, B) := \underset{A \subseteq T \subseteq [n] \setminus B}{\min} f_i(T)$, which is \emph{monotone and submodular} in each argument (Lemma~\ref{lem:block-minimum-separate-submodularity}), and define importance weights $p_i$ by summing the weights from $2n$ \emph{local Gibbs distributions} on the single-terminal families $f^{\to}(u, \cdot)$ and $f^{\to}(\cdot, v)$. Each local Gibbs distribution is homogeneous with total weight at most $n - 1$ (giving $\sum_{i=1}^m p_i = O(n^2)$) while capturing submodular diminishing returns across subsets.
    \item \textbf{Positive chaining without symmetrization}: Instead of chaining directly on the non-monotone values $f_i(S_j) - f_i(S_{j-1})$ (which can take both signs and required Rademacher symmetrization and Dudley's entropy integral in~\cite{JLLS23, KPS24}), we exploit the monotonicity of $f_i^{\to}$ to decompose each query $f_i(S) = f_i^{\to}(S, \overline{S})$ across scales into a telescoping sum of \emph{nonnegative} increments $f_i^{\to}(I_j, O_j) - f_i^{\to}(I_{j-1}, O_{j-1}) \ge 0$ over high-influence coordinate subsets $(I_j, O_j) \subseteq (S, \overline{S})$. Because the increments are nonnegative, an elementary Bernstein bound and a union bound at each scale directly yield uniform concentration across all $2^n$ queries with $O(n^2 \log n / \varepsilon^2)$ samples.
    \item \textbf{Zero-temperature cooling and low-collision partitions for $|S| \le k$}: Cooling the $2n$ local Gibbs distributions to zero temperature collapses each distribution to a deterministic ground state, derandomizing the weight computation and connecting Gibbs sampling to classical minimum-ratio strength decompositions~\cite{BK96, Qua24} (see Section~\ref{sec:sparsification-n2} and Appendix~\ref{sec:Gibbs}). Under the cardinality constraint $|S| \le k$ (Theorem~\ref{inform:sparsification_cardinality_constraint}), this decomposition viewpoint lets us compress the total weight from $O(n^2)$ to $O(nk)$ by capping each inside-terminal decomposition at $O(k)$ and replacing the $n$ outside terminals with $O(k)$ blocks from a family of low-collision partitions.
\end{itemize}

\subsection{Related work}
\label{sec:previous-work}

\paragraph{Graph, hypergraph, and matrix sparsification.}
Graph sparsification has been extensively studied over the past three decades, starting with Bencz\'{u}r and Karger's cut sparsifiers~\cite{BK96} and Spielman and Teng's spectral sparsifiers~\cite{ST04, batson2012twice}, which have become standard primitives in fast graph algorithms~\cite{spielman2011graph, Vishnoi13}. Related graph compression notions include graph spanners~\cite{Spanner_survey}, expander decompositions~\cite{GRST21}, and quotient sparsification~\cite{Qua24}. Generalizing beyond graphs, a large body of recent work has established nearly linear-sized cut sparsifiers for undirected hypergraphs~\cite{KK15, BST19, CKN20, BKV21} and near-quadratic cut sparsifiers for directed hypergraphs~\cite{KPS24,Qua24}, as well as spectral sparsifiers for hypergraphs and submodular hypergraphs~\cite{SomaY19spectral, KapralovKTY21a, OST23, Lee23, JambulapatiLS23, JLLS23}. In the continuous setting, matrix sparsification and sketching encompass Johnson--Lindenstrauss projections~\cite{JohnsonLindenstrauss:84}, $\ell_p$ subspace embeddings~\cite{BLM89, Talagrand90, CP15_Lewis}, low-rank approximations~\cite{CW17}, and the Restricted Isometry Property (RIP) for $k$-sparse vectors~\cite{Donoho_compressed_sensing, CT05, RV}; recently, Jambulapati, Lee, Liu, and Sidford~\cite{JLLS23, JLLS24} extended row sparsification from norms to general seminorms and generalized linear models such as the Huber and Tukey loss functions.

\paragraph{Sparsification of decomposable submodular functions.}
Sparsifying a sum of submodular functions $F = \sum_{i=1}^m f_i$ (a decomposable submodular function~\cite{SK10, LHM20,pmlr-v139-axiotis21a}) generalizes hypergraph cut sparsification to arbitrary submodular components. For the special case of cardinality-based submodular functions of the form $f_i(S) = g_i(|S \cap T_i|)$, Benson, Kleinberg, and Veldt~\cite{BKV21, VBK21} constructed sparsifiers of size $O(n)$ (allowing auxiliary edges). Rafiey and Yoshida~\cite{RafieyY22submodular} initiated the study of reweighting-based sparsifiers for general submodular functions, proving an existential bound of $O(\frac{M n^2}{\varepsilon^2})$ and a constructive bound of $O(\frac{M n^{2.5}}{\varepsilon^2})$ for monotone submodular functions, where $M$ is the maximum number of extreme points in the base polytope of $f_i$ (which can be exponential in $n$~\cite{KK24}). Kudla and \v{Z}ivn\'{y}~\cite{KZ23} subsequently improved the constructive bound to $O(M n^2 / \varepsilon^2)$ for monotone and low-curvature submodular functions.

Removing this exponential dependence on $M$ has been the focus of several recent breakthroughs. For sums of \emph{symmetric} submodular functions, Jambulapati, Lee, Liu, and Sidford~\cite{JLLS23} exploited the fact that the Lov\'{a}sz extension of a symmetric submodular function is a seminorm, obtaining polynomial-time sparsifiers of nearly optimal size $\widetilde{O}(n/\varepsilon^2)$ via Gibbs-distribution sampling. Building on~\cite{JLLS23}, Khanna, Putterman, and Sudan~\cite{KPS24} reduced \emph{monotone} submodular sparsification to the symmetric case by adding a single auxiliary vertex, achieving $\widetilde{O}(n/\varepsilon^2)$ size for all monotone submodular functions. For \emph{general} (non-monotone, asymmetric) submodular functions, where directed cut functions impose an $\Omega(n^2)$ lower bound~\cite{CohenKPPRSV17, OST23}, Kenneth and Krauthgamer~\cite{KK24} designed a polynomial-time algorithm to upper bound the sensitivities $s_i = \max_S \frac{f_i(S)}{F(S)}$ by pairwise directed min-cut scores $\rho_i$ with tight total sum $\sum_{i=1}^m \rho_i = O(n^2)$. Combining this with a naive union bound over all $2^n$ queries yielded a sparsifier of size $O(n^3/\varepsilon^2)$, leaving an $O(n)$ gap to the $\Omega(n^2)$ lower bound (while~\cite{KPS24,Qua24} obtained $\widetilde{O}(n^2/\varepsilon^2)$ in the special case of $\{0,1\}$-valued submodular functions and directed hyperedges). Our Theorem~\ref{inform:sparsification_all_subsets} closes this gap up to a logarithmic factor for arbitrary nonnegative submodular functions, resolving~\cite[Conjecture~1.14]{KK24}.

\paragraph{Constrained submodular sparsification and minimum-ratio decompositions.}
Sparsification of submodular functions under cardinality and matroid constraints was first studied by Rafiey and Yoshida~\cite{RafieyY22submodular}, who obtained bounds depending on the rank $k$ and the base-polytope parameter $M$ for monotone functions. While the unconditional $\widetilde{O}(n/\varepsilon^2)$ bound of~\cite{KPS24} superseded those results for monotone functions when $M$ is large, whether a cardinality constraint $|S| \le k$ can reduce the $\Omega(n^2)$ barrier for general (non-monotone, asymmetric) submodular functions---or reduce the $\widetilde{O}(n)$ bound for symmetric/monotone functions to $o(n)$---remained open; our Theorem~\ref{inform:sparsification_cardinality_constraint} and Section~\ref{sec:lower_bounds} answer both questions. Technically, our zero-temperature cooling of local Gibbs distributions connects importance sampling to classical minimum-ratio chains in the theory of principal partitions of submodular functions~\cite{Narayanan91, Fujishige05} and to Bencz\'{u}r--Karger strength decompositions~\cite{BK96}, which were also recently used by Quanrud~\cite{Qua24} for quotient sparsification of a single monotone submodular function.

\subsection{Discussion}

Our work leaves several intriguing open questions for future research:

\begin{enumerate}[leftmargin=*]
    \item \textbf{Faster and nearly linear-time algorithms}: The preprocessing time of our deterministic algorithm in Theorem~\ref{inform:sparsification_all_subsets} is $O(mn^3 (\mathcal{T}_{\mathrm{eval}} + n) \mathcal{T}_{\min}^2)$, which relies on general submodular function minimization. By contrast, many sparsification algorithms for graphs, hypergraphs, and matrices run in nearly linear time~\cite{BK96, ST04, CP15_Lewis, JLLS23}. Can one compute $\widetilde{O}(n^2/\varepsilon^2)$-size submodular sparsifiers in nearly linear time (or with a nearly linear number of value-oracle queries), either for general submodular functions or for structured families such as directed hypergraphs and welfare functions?
    
    \item \textbf{Structured query families and sublinear sparsifiers}: Theorem~\ref{inform:sparsification_cardinality_constraint} shows that restricting queries to the cardinality-constrained family $\binom{[n]}{\le k}$ unconditionally reduces the size of general submodular sparsifiers from $\widetilde{O}(n^2)$ to $\widetilde{O}(nk)$, whereas our lower bounds in Section~\ref{sec:lower_bounds} show that $\Omega(n)$ size remains necessary even for regular symmetric or almost-regular monotone submodular functions. What other natural constraints on the query family $\mathcal{Q} \subseteq \{0,1\}^n$ can reduce the size of submodular sparsifiers? Furthermore, what structural conditions on $f_1, \dots, f_m$ (beyond monotonicity and singleton regularity) are needed to achieve RIP-like sparsifiers of size $(k \log n)^{O(1)}$ under $|S| \le k$?
    
    \item \textbf{Broader applications of local Gibbs distributions and positive chaining}: Our framework combines local Gibbs distributions (and their zero-temperature minimum-ratio limits) with positive chaining over monotone envelopes $f_i^{\to}(A, B)$ to sparsify objectives that lack symmetry or norm structure. Can these techniques find further applications in other asymmetric or non-linear sparsification problems, such as $k$-submodular functions or directed spectral hypergraph sparsification?
\end{enumerate}

\paragraph{Organization.}
The remainder of the paper is organized as follows. Section~\ref{sec:overview} provides a self-contained technical overview of our constructions, counterexamples, and proof techniques. Section~\ref{sec:preliminaries} introduces notation, basic properties of submodular functions and directed block min-cut functions, and our positive-chaining concentration lemma. In Section~\ref{sec:sparsification-n2}, we prove Theorem~\ref{inform:sparsification_all_subsets} via zero-temperature terminal decompositions. In Section~\ref{sec:sparsification-nk}, we extend this framework  to prove Theorem~\ref{inform:sparsification_cardinality_constraint} under cardinality constraints. Section~\ref{sec:lower_bounds} presents our lower bounds under cardinality constraints. Finally, Appendix~\ref{sec:Gibbs} presents the proof of Theorem~\ref{inform:sparsification_all_subsets} via finite-temperature local Gibbs distributions, and Appendices~\ref{sec:preprocess} and~\ref{appendix:proofs} contain preprocessing reductions and omitted technical proofs.

\section{Overview}\label{sec:overview}

We provide an overview of our algorithms and proof techniques. Sparsifying a sum $F(S) = \sum_{i=1}^m f_i(S)$ of general nonnegative submodular functions via importance sampling requires addressing two main questions: (i)~how to design importance weights $p_i$ with total sum $O(n^2)$ (or $O(nk)$ when $|S| \le k$) when the functions $f_i$ are neither symmetric nor monotone, and (ii)~how to bound the sampling error simultaneously across all $2^n$ (or $\binom{n}{\le k}$) queries. Our approach proceeds in five steps:
\begin{itemize}[leftmargin=*]
\item \textbf{Directed block min-cut functions and positive chaining (Section~\ref{subsec:overview-block})}: We introduce the \emph{directed block min-cut function} $f_i^{\to}(A, B) := \min_{A \subseteq T \subseteq [n] \setminus B} f_i(T)$, which is monotone and submodular in $A$ (for fixed $B$) and in $B$ (for fixed $A$). Decomposing each $f_i(S) = f_i^{\to}(S, \bar{S})$ into a telescoping sum of \emph{nonnegative} increments $f_i^{\to}(I_j, O_j) - f_i^{\to}(I_{j-1}, O_{j-1}) \ge 0$ over nested high-influence subsets $I_j \subseteq S$ and $O_j \subseteq \bar{S}$ bypasses the Rademacher symmetrization and Dudley chaining machinery of~\cite{JLLS23, KPS24}, reducing all-queries concentration to a residual bound \textbf{(P1)} and a total influence budget \textbf{(P2)} (listed in \eqref{eq:two_properties} below).

\item \textbf{Bottlenecks of previous sampling weights (Section~\ref{subsec:overview-counterexample})}: On an explicit two-block counterexample requiring $\Omega(n^3)$ samples, we show why existing weights fail: classical sensitivity $s_i$ lacks homogeneity across disjoint blocks, the single global Gibbs distribution of~\cite{JLLS23} fails under non-monotonicity and has total budget only $O(n)$, and the pairwise min-cut weights $\rho_i$ of~\cite{KK24} evaluate pairs $(u, v)$ independently and can over-amplify a function's weight by $\Theta(n^2)$.

\item \textbf{Importance weights from $2n$ local Gibbs distributions (Section~\ref{subsec:overview-gibbs})}: We unify~\cite{JLLS23} and~\cite{KK24} by applying Gibbs-distribution sampling to the $2n$ monotone block min-cut families $f_i^{\to}(u, \cdot)$ ($u \in [n]$) and $f_i^{\to}(\cdot, v)$ ($v \in [n]$). These weights recover~\cite{KK24} in the additive case while capturing submodular diminishing returns across subsets with total weight at most $2n(n-1)$, and their marginal densities at zero naturally define coordinate influences satisfying \textbf{(P1)} and \textbf{(P2)}.

\item \textbf{Cooling the Gibbs distributions to zero temperature (Section~\ref{subsec:overview-cooling})}: After a logarithmic change of variables, the Gibbs distributions at any integer inverse temperature $\beta \ge 1$ yield valid sampling weights with sparsity $\widetilde{O}(n^2)$ (Theorem~\ref{thm:gibbs-near-quadratic}). Cooling $\beta \to +\infty$ collapses each Gibbs distribution to a single deterministic ground state computable in polynomial time via submodular minimization, derandomizing the weight computation.

\item \textbf{Extension to cardinality constraints $|S| \le k$ (Section~\ref{subsec:overview-nk})}: When $|S| \le k$, we exploit the size asymmetry between $S$ ($|S| \le k$) and $\bar{S}$ ($|\bar{S}| \ge n - k$) to compress the total weight to $O(nk)$: for $f_i^{\to}(u, \cdot)$ ($u \in S$) we cap the weight budget per $u$ at $O(k)$, and for $f_i^{\to}(\cdot, v)$ ($v \in \bar{S}$) we replace the $n-k$ outside vertices by $O(k)$ blocks via low-collision partitions.
\end{itemize}

\paragraph{Conventions.}
For ease of exposition, we assume throughout this section that every component is normalized at the endpoints, i.e., $f_i(\varnothing) = f_i([n]) = 0$ for all $i \in [m]$; as shown in Appendix~\ref{sec:preprocess}, this holds without loss of generality by extending the ground set by two auxiliary elements. For a vector $z$, we write $z_v$ or $z(v)$ interchangeably for its $v$-th coordinate. For any submodular function $g : \{0,1\}^n \to \mathbb{R}_{\ge 0}$ with $g(\varnothing) = 0$, we denote its Lov\'asz extension~\cite{Lovasz83} by $\widehat{g}(z) := \int_0^\infty g(\{v \in [n] : z_v \ge \lambda\})\,\mathrm{d}\lambda$ for $z \in \mathbb{R}_{\ge 0}^n$ (see Section~\ref{sec:preliminaries}).

\subsection{Directed block min-cut functions and positive chaining}\label{subsec:overview-block}\label{subsec:overview-chaining}

The first tool of this work is the \emph{directed block min-cut function} of a submodular function, which extends Fujishige's definition $g(S) := \min_{T \subseteq S} f(T)$~\cite{Fujishige05} and Kenneth and Krauthgamer's pairwise directed min-cut $f^{u \to v} := \min_{T \subseteq [n],\, u \in T,\, v \notin T} f(T)$~\cite{KK24} to two disjoint subsets $A$ and $B$.

\begin{definition}[Directed block min-cut function]\label{def:directed_block_mincut_function}
Given a submodular function $f : \{0,1\}^n \to \mathbb{R}_{\ge 0}$ and disjoint subsets $A, B \subseteq [n]$, define
\[
f^{\to}(A, B) := \min_{T \subseteq [n] :\, A \subseteq T,\, T \cap B = \varnothing} f(T).
\]
When $A = \{u\}$ (or $B = \{v\}$), we write $f^{\to}(u, B)$ (or $f^{\to}(A, v)$) in place of $f^{\to}(\{u\}, B)$ (or $f^{\to}(A, \{v\})$).
\end{definition}

Compared to the original function $f$---which is neither symmetric nor monotone in general---$f^{\to}(A, B)$ enjoys two key properties (proved in Lemmas~\ref{lem:block-minimum-separate-submodularity} and~\ref{lem:block-residual-two-block}):
\begin{enumerate}[leftmargin=*]
\item \textbf{Monotonicity and Submodularity}: For fixed $B$, $f^{\to}(A, B)$ is monotone non-decreasing and submodular in $A \subseteq [n] \setminus B$; symmetrically, for fixed $A$, it is monotone non-decreasing and submodular in $B \subseteq [n] \setminus A$. Since $f(\varnothing) = f([n]) = 0$, we also have $f^{\to}(\varnothing, B) = f^{\to}(A, \varnothing) = 0$.
\item \textbf{Approximation via Subsets $(I, O)$}: When $A = S$ and $B = \bar{S} := [n] \setminus S$, the only feasible set is $T = S$, so $f_i^{\to}(S, \bar{S}) = f_i(S)$. If we approximate $f_i(S)$ by $f_i^{\to}(I, O)$ on subsets $I \subseteq S$ and $O \subseteq \bar{S}$ (leaving $U := \overline{I \cup O}$ free), the approximation error (residual) $\Delta_i(S, U) := f_i(S) - f_i^{\to}(I, O) = f_i^{\to}(S, \bar{S}) - f_i^{\to}(S \setminus U, \bar{S} \setminus U)$ is always nonnegative and satisfies (via Lemma~\ref{lem:block-residual-two-block} and submodularity of $f_i^{\to}$):
\begin{equation}
0 \le \Delta_i(S, U) \le f_i^{\to}(S, \bar{S} \setminus O) + f_i^{\to}(S \setminus I, \bar{S}) \le \sum_{u \in S} f_i^{\to}(u, \bar{S} \setminus O) + \sum_{v \in \bar{S}} f_i^{\to}(S \setminus I, v). \label{eq:bound_Delta_S1_S2}
\end{equation}
\end{enumerate}

\paragraph{Positive chaining and all-queries concentration via (P1) and (P2).}
Equation~\eqref{eq:bound_Delta_S1_S2} allows us to bound the sampling error across all $2^n$ queries via a \emph{positive chaining} argument. For each query $S$, we approximate $f_i(S)$ across gradually enlarging subsets $\varnothing = (I_0, O_0) \subseteq (I_1, O_1) \subseteq \dots \subseteq (I_J, O_J) = (S, \bar{S})$ and write $f_i(S) = \sum_{j=1}^J z_{i,j}$. Unlike prior works~\cite{JLLS23, KPS24} that chain directly on the non-monotone functions $f_i(S_j) - f_i(S_{j-1})$ (whose signed increments require Rademacher symmetrization and Dudley's entropy integral), the monotonicity of $f_i^{\to}$ guarantees that every increment $z_{i,j} := f_i^{\to}(I_j, O_j) - f_i^{\to}(I_{j-1}, O_{j-1})$ is \emph{nonnegative} and bounded by $\Delta_i(S, U_{j-1})$ (where $U_j := \overline{I_j \cup O_j}$). This lets us use an elementary Bernstein bound (Lemma~\ref{lem:Bernstein}) and a union bound at each scale.

Specifically, suppose that given importance weights $p_i$, for any query $S \subseteq [n]$ we can assign each coordinate $w \in [n]$ a nonnegative influence $\operatorname{inf}_S(w) \ge 0$ satisfying a \emph{residual bound} \textbf{(P1)} and a \emph{total influence budget} \textbf{(P2)} for every threshold $\theta > 0$ and $U_\theta := \{w \in [n] : \operatorname{inf}_S(w) \le \theta\}$:
\begin{equation}\label{eq:two_properties}
\textbf{(P1)}\;\; \Delta_i(S, U_\theta) \le O(\theta \cdot p_i) \quad (\forall i \in [m]),
\qquad\qquad
\textbf{(P2)}\;\; \sum_{w \in [n]} \operatorname{inf}_S(w) \le 4 F(S).
\end{equation}
Once (P1) and (P2) hold, independent sampling with probabilities $\min\{1, \Lambda p_i\}$ yields an $O(\Lambda \sum_i p_i)$-size sparsifier by a three-step positive chaining argument (Section~\ref{sec:sparsification-n2} and Appendix~\ref{sec:Gibbs}):
\begin{enumerate}[leftmargin=*]
\item \textbf{Multiscale decomposition}: For a query $S$ with $F(S) > 0$, set dyadic thresholds $\theta_j = 2^{1-j} F(S)$ for $j = 0, 1, \dots, J$, let $U_j := U_{\theta_j}$, $I_j := S \setminus U_j$, $O_j := \bar{S} \setminus U_j$, and $\ell_j := |I_j| + |O_j|$. Then $f_i(S) = \sum_{j=1}^J z_{i,j}$ with $z_{i,j} := f_i^{\to}(I_j, O_j) - f_i^{\to}(I_{j-1}, O_{j-1}) \ge 0$, and (P1) gives $0 \le z_{i,j} \le \Delta_i(S, U_{j-1}) \le 2\theta_{j-1} p_i$ (shown for constant $1$ in~\textbf{(P1)}), i.e., $\max_i \frac{z_{i,j}}{p_i} \le 2\theta_{j-1}$.
\item \textbf{Geometric complexity bound}: Since $w \in I_j \cup O_j$ only while $\theta_j < \operatorname{inf}_S(w)$, geometric summation and (P2) give
\[
\sum_{j=1}^J \ell_j \max_{i \in [m]} \frac{z_{i,j}}{p_i} \le 2 \sum_{j=1}^J \ell_j \theta_{j-1} \le 8 \sum_{w \in [n]} \operatorname{inf}_S(w) \le 32 F(S).
\]
\item \textbf{Per-scale union bound}: An increment vector $(z_{i,j})_{i=1}^m$ with $\ell_j = \ell$ fixed coordinates is determined by $(I_{j-1}, O_{j-1}, I_j, O_j)$, which belongs to a \emph{query-independent} dictionary $\mathcal{D}_\ell$ of size at most $\binom{n}{\ell} 4^\ell \le (4n)^\ell$. Applying Lemma~\ref{lem:Bernstein} (with failure probability $\delta / (32n)^\ell$) and taking a union bound over $\mathcal{D}_\ell$ for all $\ell \ge 1$ guarantees that with probability at least $1 - \delta$, every query $S \subseteq [n]$ satisfies
\[
\bigl|\widetilde{F}(S) - F(S)\bigr| \le \sum_{j=1}^J \left|\sum_{i=1}^m (w_i - 1) z_{i,j}\right| \le \frac{\varepsilon}{2} \sum_{j=1}^J \sum_{i=1}^m z_{i,j} + O\!\left(\frac{\log(2n/\delta)}{\varepsilon \Lambda}\right) \sum_{j=1}^J \ell_j \max_{i \in [m]} \frac{z_{i,j}}{p_i} \le \varepsilon F(S)
\]
for $\Lambda = \Theta(\varepsilon^{-2} \log(2n/\delta))$, since $\sum_{j=1}^J \sum_{i=1}^m z_{i,j} = F(S)$ telescopes without absolute values. When $\sum_{i=1}^m p_i = O(n^2)$, this yields sparsity $O(\Lambda \sum_i p_i) = O(\varepsilon^{-2} n^2 \log(2n/\delta))$.
\end{enumerate}

Thus, the entire sparsification problem reduces to choosing importance weights $p_i$ with $\sum_{i=1}^m p_i = O(n^2)$ and coordinate influences $\operatorname{inf}_S(w)$ satisfying \textbf{(P1)} and \textbf{(P2)}. Crucially, \eqref{eq:bound_Delta_S1_S2} already reduces bounding $\Delta_i(S, U)$ to bounding $2n$ \emph{monotone} submodular functions: $f_i^{\to}(u, \cdot)$ for $u \in [n]$ and $f_i^{\to}(\cdot, v)$ for $v \in [n]$.

\subsection{Bottlenecks of previous sampling weights}\label{subsec:overview-counterexample}

We first construct an explicit counterexample showing why three standard importance-sampling approaches fail to achieve near-quadratic sparsity $\widetilde{O}(n^2)$ in~\eqref{eq:def_sparsification}: (1)~classical sensitivity, (2)~Kenneth and Krauthgamer's pairwise directed min-cut weight~\cite{KK24} (which satisfies $\rho_i \ge s_i$ and $\sum_{i=1}^m \rho_i \le n(n-1)$), and (3)~the single global Gibbs distribution $\mu_F(x) \propto e^{-\widehat{F}(x)}$ on $\mathbb{R}_{\ge 0}^n$ of~\cite{JLLS23} (which satisfies $\sum_{i=1}^m p_i^{\mathrm{JLLS}} = \mathbb{E}[\widehat{F}(X)] \le n$):
\[
s_i := \max_{S \subseteq [n]} \frac{f_i(S)}{F(S)}, \qquad\qquad
\rho_i := \sum_{u \ne v} \frac{f_i^{\to}(u, v)}{\sum_{j=1}^m f_j^{\to}(u, v)}, \qquad\qquad
p_i^{\mathrm{JLLS}} := \underset{X \sim \mu_F}{\mathbb{E}}\bigl[\widehat{f}_i(X)\bigr].
\]

Throughout, we use the convention $\frac{0}{0}:=0$.

\begin{claim}\label{clm:counterexample}\label{clm:counter_examples}
There exist nonnegative submodular functions $f_1, \dots, f_m : \{0,1\}^n \to \mathbb{R}_{\ge 0}$ such that sampling proportional to either $s_i$ or $\rho_i$ requires $\Omega(n^3)$ random samples (or $\Omega(n^3)$ expected support under independent Bernoulli sampling if each unit edge on $V_2$ is split into $d$ parallel copies of weight $1/d$) to achieve relative error $0.2$ with probability at least $0.99$.
\end{claim}

\begin{proof}
Set $n = 2d$ and construct two collections of submodular functions on disjoint variable blocks $V_1 = \{1, \dots, d\}$ and $V_2 = \{d+1, \dots, 2d\}$ (Figure~\ref{fig:counterexample}):
\begin{itemize}[leftmargin=*]
\item \textbf{Block $V_1$}: Place the all-or-nothing cut $f_1(S) := \mathbf{1}[\varnothing \subsetneq S \cap V_1 \subsetneq V_1]$ together with $d-1$ weak directed edges $f_i(S) := \frac{1}{d}\operatorname{cut}_{\{i\} \to \{1\}}(S) = \frac{1}{d}\mathbf{1}[i \in S \wedge 1 \notin S]$ for $i \in \{2, \dots, d\}$.
\item \textbf{Block $V_2 = L \sqcup R$ ($|L| = |R| = d/2$)}: Include all $d^2/4$ unit directed edges $\operatorname{cut}_{a \to b}(S) := \mathbf{1}[a \in S \wedge b \notin S]$ for $(a, b) \in L \times R$.
\end{itemize}
Let $F := f_1 + \sum_{i=2}^d f_i + \sum_{a \in L,\, b \in R} \operatorname{cut}_{a \to b}$. Any $0.2$-sparsifier $\widetilde{F}$ must retain at least $0.4d$ of the weak edges $f_2, \dots, f_d$: if $T \subseteq \{2, \dots, d\}$ is the retained subset and $|T| \le 0.4d$, then on $S := \{i \in \{2, \dots, d\} : i \notin T\}$ every retained $f_i$ ($i \in T$) vanishes so $\widetilde{F}(S) = \widetilde{F}(\{1\}) = w_1 f_1(\{1\})$, whereas $F(\{1\}) = 1$ and $F(S) = 1 + \frac{d-1-|T|}{d} \ge 1.55$ (for $d \ge 20$), violating $0.8 F(S) \le \widetilde{F}(S) = \widetilde{F}(\{1\}) \le 1.2 F(\{1\})$.

\begin{figure}[t]
\centering
\begin{tikzpicture}[
    >=Stealth,
    vnode/.style={circle, draw=black!80, fill=white, thick, minimum size=6.5mm, inner sep=0pt, font=\footnotesize},
    diredge/.style={->, thick, draw=blue!70!black},
    bidedge/.style={->, semithick, draw=teal!75!black}
]
\node[font=\bfseries\small] at (0, 3.35) {Block $V_1 = \{1, \dots, d\}$};

\draw[rounded corners=8pt, dashed, thick, draw=red!70!black, fill=red!3] (-3.4, -0.65) rectangle (3.4, 2.85);
\node[text=red!75!black, font=\footnotesize, anchor=north] at (0, 2.75) {All-or-nothing cut $f_1(S) = \mathbf{1}[\varnothing \subsetneq S \cap V_1 \subsetneq V_1]$};

\node[vnode] (u2) at (-1.8, 1.75) {$2$};
\node[vnode] (u3) at (-1.8, 0.90) {$3$};
\node[font=\footnotesize] at (-1.8, 0.30) {$\vdots$};
\node[vnode] (ud) at (-1.8, -0.20) {$d$};

\node[vnode, fill=gray!15] (u1) at (1.8, 0.80) {$1$};

\draw[diredge] (u2) -- node[above, sloped, font=\scriptsize, text=blue!75!black] {$f_2$ (wt $1/d$)} (u1);
\draw[diredge] (u3) -- node[below=-1pt, pos=0.45, sloped, font=\scriptsize, text=blue!75!black] {$f_3$ (wt $1/d$)} (u1);
\draw[diredge] (ud) -- node[below, sloped, font=\scriptsize, text=blue!75!black] {$f_d$ (wt $1/d$)} (u1);

\node[align=left, font=\scriptsize, anchor=north] at (0, -0.80) {
    $\bullet$ Sensitivity $s$: $s_1 = 1,\ s_i = \Theta(1/d) \;\Rightarrow\; \sum_{i=2}^d s_i = \Theta(1)$\\[2pt]
    $\bullet$ Pairwise $\rho$: $\rho_1 = \Theta(d^2),\ \rho_i = \Theta(1/d) \;\Rightarrow\; \sum_{i=2}^d \rho_i = \Theta(1)$\\[2pt]
    $\bullet$ Our weight $p$: $p_1 = \Theta(d^2),\ p_i = \Theta(1) \;\Rightarrow\; \sum_{i=2}^d p_i = \Theta(d)$
};

\node[font=\bfseries\small] at (8.0, 3.35) {Block $V_2 = L \sqcup R$ ($|L| = |R| = d/2$)};

\draw[rounded corners=8pt, thick, draw=black!35, fill=gray!4] (4.6, -0.65) rectangle (11.4, 2.85);
\node[text=teal!70!black, font=\footnotesize, anchor=north] at (8.0, 2.75) {$d^2/4$ unit directed edges $\operatorname{cut}_{a \to b}(S)$ (weight $1$)};

\node[vnode] (l1) at (6.2, 1.75) {$a_1$};
\node[vnode] (l2) at (6.2, 0.90) {$a_2$};
\node[font=\footnotesize] at (6.2, 0.30) {$\vdots$};
\node[vnode] (lk) at (6.2, -0.20) {$a_{d/2}$};

\node[vnode] (r1) at (9.8, 1.75) {$b_1$};
\node[vnode] (r2) at (9.8, 0.90) {$b_2$};
\node[font=\footnotesize] at (9.8, 0.30) {$\vdots$};
\node[vnode] (rk) at (9.8, -0.20) {$b_{d/2}$};

\foreach \src in {l1, l2, lk} {
    \foreach \dst in {r1, r2, rk} {
        \draw[bidedge] (\src) -- (\dst);
    }
}

\node[align=left, font=\scriptsize, anchor=north] at (8.0, -0.80) {
    $\bullet$ Sensitivity $s$: $s_{a \to b} = 1 \;\Rightarrow\; \sum_{a,b} s_{a \to b} = \Theta(d^2)$\\[2pt]
    $\bullet$ Pairwise $\rho$: $\rho_{a \to b} = \Theta(1) \;\Rightarrow\; \sum_{a,b} \rho_{a \to b} = \Theta(d^2)$\\[2pt]
    $\bullet$ Our weight $p$: $p_{a \to b} = \Theta(1) \;\Rightarrow\; \sum_{a,b} p_{a \to b} = \Theta(d^2)$
};
\end{tikzpicture}
\caption{The two-block counterexample of Claim~\ref{clm:counterexample} ($n = 2d$). Under both $s_i$ and $\rho_i$, $\{f_2, \dots, f_d\}$ receives only a $\Theta(1/d^2)$ fraction of the total weight, requiring $\Omega(d^3) = \Omega(n^3)$ samples to retain $\Omega(d)$ weak edges.}
\label{fig:counterexample}
\end{figure}

\noindent\textbf{Failure of $s_i$ and $\rho_i$:}
On $V_1$, $s_1 = 1$ (at $S = \{1\}$) and $s_i = \frac{1/d}{1 + 1/d} = \Theta(1/d)$ (at $S = \{i\}$), so $\sum_{i=2}^d s_i = \Theta(1)$; on $V_2$, each $\operatorname{cut}_{a \to b}$ is isolated by $S_{a,b} := \{a\} \cup (R \setminus \{b\})$, giving $s_{a \to b} = 1$ and $\sum_{a,b} s_{a \to b} = d^2/4 = \Theta(d^2)$. Under $\rho_i$, $f_1^{\to}(u, v) = 1$ for \emph{every} ordered pair $u \ne v$ in $V_1$ whereas $f_i^{\to}(u, v) = 1/d$ only for $(u, v) = (i, 1)$, giving $\rho_1 \ge d(d-1)/2 = \Theta(d^2)$ and $\rho_i = \frac{1/d}{1 + 1/d} = \Theta(1/d)$ for $i \in \{2, \dots, d\}$. In both cases, $\sum_{i=2}^d s_i = \Theta(1)$ and $\sum_{i=2}^d \rho_i = \Theta(1)$ are drowned out by a total weight of $\Theta(d^2)$, forcing $\Omega(d^3) = \Omega(n^3)$ random draws.
\end{proof}

\paragraph{Failure of a single global Gibbs distribution~\cite{JLLS23} and takeaway.}
The global Gibbs distribution $\mu_F(x) \propto e^{-\widehat{F}(x)}$ of~\cite{JLLS23} also fails on Block $V_2$: because $\operatorname{cut}_{a \to b}$ is non-monotone, $\sum_{a \in L, b \in R} \widehat{\operatorname{cut}}_{a \to b}(x) = \sum_{a \in L, b \in R} \max\{x_a - x_b, 0\}$ vanishes on the unbounded cone $\{\max_{a \in L} x_a \le \min_{b \in R} x_b\}$ (so $e^{-\widehat{F}(x)}$ is not integrable), and even after regularization a single $1$-homogeneous distribution has total budget $\sum_{i=1}^m \mathbb{E}[\widehat{f}_i(X)] = \mathbb{E}[\widehat{F}(X)] \le n$, which cannot upper-bound the unit sensitivities $s_{a \to b} = 1$ of all $\Theta(n^2)$ directed edges in $V_2$ (a typical edge receives weight only $O(1/n)$). Together, these three failures highlight three requirements for our importance weights $p_i$:
\begin{enumerate}[leftmargin=*]
\item \textbf{Homogeneity across blocks}: Unlike $s_i$ (whose sum varies from $\Theta(1)$ on $V_1$ to $\Theta(d^2)$ on $V_2$, starving $V_1$), the weight allocation must be homogeneous so disjoint blocks receive comparable budgets.
\item \textbf{$2n$ monotone terminal families with $O(n^2)$ budget}: Unlike~\cite{JLLS23}, decomposing into the $2n$ directed block min-cut families $f_i^{\to}(u, \cdot)$ and $f_i^{\to}(\cdot, v)$ restores monotonicity and provides an $O(n^2)$ total budget.
\item \textbf{Evaluating $f_i^{\to}(u, D)$ over subsets rather than singletons}: While $\rho_i$~\cite{KK24} allocates $d-1$ units of weight per terminal $u = i \in \{2, \dots, d\}$, it splits this budget independently across the $d-1$ singletons $v \in V_1 \setminus \{i\}$, allowing $f_1$ (with $f_1^{\to}(i, v) = 1$ for all $v$) to absorb $d-2$ units on $v \ne 1$ and leave $f_i$ with $\Theta(1/d)$. Because $D \mapsto f_1^{\to}(i, D)$ has zero marginal growth after the first vertex ($f_1^{\to}(i, V_1 \setminus \{i\}) = 1$), evaluating on the full subset $D = V_1 \setminus \{i\}$ keeps $f_1$'s value at $1$ rather than $d-1$ and boosts each weak edge $f_i$ by a factor of $d-1$ to $\Theta(1)$.
\end{enumerate}

\subsection{Importance weights from $2n$ local Gibbs distributions}\label{subsec:overview-gibbs}

Motivated by~\eqref{eq:bound_Delta_S1_S2} and the three requirements above, we construct $p_i$ as the sum of $2n$ homogeneous share vectors---one for each inside terminal $u \in [n]$ on the monotone submodular family $\bigl(f_i^{\to}(u, \cdot)\bigr)_{i=1}^m$, and one for each outside terminal $v \in [n]$ on $\bigl(f_i^{\to}(\cdot, v)\bigr)_{i=1}^m$---each having total weight at most $n-1$.

\paragraph{Local Gibbs distributions and how they unify~\cite{JLLS23, KK24}.}
For each $u \in [n]$, define the aggregate Lov\'asz potential $H_u(x) := \sum_{i=1}^m \widehat{f_i^{\to}}(u, x)$ on $\mathbb{R}_{\ge 0}^{[n] \setminus \{u\}}$, and let $E_u := \{v \in [n] \setminus \{u\} : H_u(\mathbf{1}_{\{v\}}) > 0\}$ be its effective coordinates. We define the Gibbs distribution $\mu_{H_u}(x) \propto e^{-H_u(x)}$ on $\mathbb{R}_{\ge 0}^{E_u}$. Symmetrically, for each $v \in [n]$, define $G_v(y) := \sum_{i=1}^m \widehat{f_i^{\to}}(y, v)$ on $\mathbb{R}_{\ge 0}^{[n] \setminus \{v\}}$, with effective coordinates $E_v^{\mathrm{out}} := \{w \in [n] \setminus \{v\} : G_v(\mathbf{1}_{\{w\}}) > 0\}$ and Gibbs distribution $\mu_{G_v}(y) \propto e^{-G_v(y)}$ on $\mathbb{R}_{\ge 0}^{E_v^{\mathrm{out}}}$. Throughout, vectors on effective coordinates are extended by zero outside them. By positive $1$-homogeneity (Euler's identity), each local Gibbs distribution has total share $\sum_{i=1}^m \mathbb{E}_{X \sim \mu_{H_u}}[\widehat{f_i^{\to}}(u, X)] = \mathbb{E}[H_u(X)] = |E_u| \le n - 1$. We define the sampling weight $p_i$ by
\begin{equation}
p_i := \sum_{u \in [n]} \underbrace{\underset{X \sim \mu_{H_u}}{\mathbb{E}}\bigl[\widehat{f_i^{\to}}(u, X)\bigr]}_{\eta_i(u)} + \sum_{v \in [n]} \underbrace{\underset{Y \sim \mu_{G_v}}{\mathbb{E}}\bigl[\widehat{f_i^{\to}}(Y, v)\bigr]}_{\gamma_i(v)}, \qquad \text{with } \sum_{i=1}^m p_i \le 2n(n-1). \label{eqdef:weights_sum_2n_Gibbs}
\end{equation}
This definition resolves all three bottlenecks in Section~\ref{subsec:overview-counterexample}:
\begin{itemize}[leftmargin=*]
\item \textbf{Monotonicity and homogeneity}: Each potential $H_u$ and $G_v$ is monotone and $1$-homogeneous, allocating a fixed budget of at most $n-1$ per terminal (and $O(n^2)$ in total) so disjoint blocks cannot starve one another and Block $V_2$ is fully covered.
\item \textbf{From additive singletons~\cite{KK24} to submodular subsets}: If each $D \mapsto f_i^{\to}(u, D)$ \emph{were additive} across singletons $v \in D$, then $\widehat{f_i^{\to}}(u, x) = \sum_{v \in E_u} f_i^{\to}(u, v) x_v$ and $X_v \sim \operatorname{Exp}\bigl(\sum_{j=1}^m f_j^{\to}(u, v)\bigr)$ independently for $v \in E_u$ under $\mu_{H_u}$, so $\eta_i(u) = \sum_{v \in E_u} f_i^{\to}(u, v) \mathbb{E}[X_v] = \sum_{v \in E_u} \frac{f_i^{\to}(u, v)}{\sum_j f_j^{\to}(u, v)}$ would sum over $u \in [n]$ to \emph{precisely} the pairwise weight $\rho_i$ of~\cite{KK24}. Using the true Lov\'asz extension $\widehat{f_i^{\to}}(u, X)$ \emph{replaces this linear approximation} by submodular evaluation over subsets, boosting each weak edge $f_2, \dots, f_d$ in Claim~\ref{clm:counterexample} from $\Theta(1/d)$ to $\Theta(1)$.
\end{itemize}

\paragraph{Achieving (P1) and (P2) via marginal densities at zero.}
Recall from~\eqref{eq:bound_Delta_S1_S2} that $\Delta_i(S, U) \le \sum_{u \in S} f_i^{\to}(u, U \setminus S) + \sum_{v \in \bar{S}} f_i^{\to}(S \cap U, v)$. Under $\mu_{H_u}$ (with partition function $Z_{H_u} := \int_{\mathbb{R}_{\ge 0}^{E_u}} e^{-H_u(x)}\,\mathrm{d}x$), we measure the influence of coordinate $v \in E_u$ on $H_u$ by the sensitivity of the free energy $-\log Z_{H_u}$ under a coordinate shift $x \mapsto x + t e_v$ at $t = 0$:
\[
q_v^{(u)}(0) := \left.-\frac{\mathrm{d}}{\mathrm{d}t} \log \int_{\mathbb{R}_{\ge 0}^{E_u}} e^{-H_u(x + t e_v)}\,\mathrm{d}x\right|_{t=0} = \underset{X \sim \mu_{H_u}}{\mathbb{E}}\bigl[\partial_v H_u(X)\bigr] = \frac{1}{Z_{H_u}} \int_{\mathbb{R}_{\ge 0}^{E_u \setminus \{v\}}} e^{-H_u(x_v = 0,\, x_{-v})}\,\mathrm{d}x_{-v},
\]
which is the 1D marginal density of $X_v \in [0, \infty)$ at $x_v = 0$ (Appendix~\ref{sec:gibbs-tau-one}). Set $q_v^{(u)}(0) := 0$ for $v \notin E_u$, and symmetrically $\tilde{q}_w^{(v)}(0) := 0$ for $w \notin E_v^{\mathrm{out}}$. By log-concavity of $\mu_{H_u}$, $q_v^{(u)}(0)$ satisfies the two properties needed for (P1) and (P2) (Appendix~\ref{sec:gibbs-tau-one}, Lemmas~\ref{lem:gibbs-density-greedy} and~\ref{lem:gibbs-density-influence}):
\begin{enumerate}[leftmargin=*]
\item \textbf{Scale control for (P1)}: For $v \in E_u$, the 1D marginal density of $X_v$ is log-concave and non-increasing on $[0, \infty)$, so it satisfies $\frac{1}{2 q_v^{(u)}(0)} \le \mathbb{E}[X_v] \le \frac{1}{q_v^{(u)}(0)}$. Since $\widehat{f_i^{\to}}(u, \cdot)$ is convex, coordinatewise monotone, and $1$-homogeneous, Jensen's inequality gives $\eta_i(u) = \mathbb{E}[\widehat{f_i^{\to}}(u, X)] \ge \widehat{f_i^{\to}}(u, \mathbb{E}[X])$, and since $\mathbb{E}[X_v] \ge \frac{1}{2 \max_{w \in D} q_w^{(u)}(0)}$ for all $v \in D \cap E_u$ whenever $D \cap E_u \ne \varnothing$, we obtain $f_i^{\to}(u, D) = \widehat{f_i^{\to}}(u, \mathbf{1}_D) \le 2 \eta_i(u) \max_{v \in D} q_v^{(u)}(0)$; the inequality is trivial if $D \cap E_u = \varnothing$. Summing over $u \in A$ yields
\begin{equation}
f_i^{\to}(A, D) \le 2 \sum_{u \in A} \left(\eta_i(u) \cdot \max_{v \in D} q_v^{(u)}(0)\right). \label{eq:beta_i_density_submodular}
\end{equation}
\item \textbf{Marginal identity for (P2)}: Since $H_u$ is a Lov\'asz extension, its partial derivatives $\partial_x H_u(x)$ telescope to $\sum_{v \in E_u} \partial_v H_u(x) = H_u(E_u) = H_u([n] \setminus \{u\})$ a.e., so $\sum_{v \in E_u} q_v^{(u)}(0) = \mathbb{E}[\sum_{v \in E_u} \partial_v H_u(X)] = H_u([n] \setminus \{u\})$. Comparing $\mu_{H_u}$ ($u \in S$) against the Gibbs distribution of $H_S(x) := \sum_{i=1}^m \widehat{f_i^{\to}}(S, x)$ on its effective coordinates (which satisfies $H_S(\bar{S}) = F(S)$) via positive association (Appendices~\ref{sec:gibbs-budget} and~\ref{sec:gibbs-tau-one}) bounds the maximum coordinate influences across $u \in S$, and symmetrically for the marginal densities $\tilde{q}_u^{(v)}(0)$ of $\mu_{G_v}$ at zero:
\begin{equation}
\forall S \subseteq [n], \qquad \sum_{v \notin S} \max_{u \in S} q_v^{(u)}(0) \le 2 F(S) \qquad\text{and}\qquad \sum_{u \in S} \max_{v \notin S} \tilde{q}_u^{(v)}(0) \le 2 F(S). \label{eq:split_budgets}
\end{equation}
\end{enumerate}
Setting $\operatorname{inf}_S(w) := \max_{v \in \bar{S}} \tilde{q}_w^{(v)}(0)$ for $w \in S$ and $\operatorname{inf}_S(w) := \max_{u \in S} q_w^{(u)}(0)$ for $w \in \bar{S}$, equations~\eqref{eq:bound_Delta_S1_S2}, \eqref{eq:beta_i_density_submodular}, and~\eqref{eq:split_budgets} immediately yield (P1) and (P2).

\subsection{Cooling the Gibbs distributions to zero temperature (ground states)}\label{subsec:overview-cooling}

While one can estimate the Gibbs weights~\eqref{eqdef:weights_sum_2n_Gibbs} via randomized convex-body sampling~\cite{JLLS23} (Appendix~\ref{sec:gibbs-polynomial-time}), in Appendix~\ref{sec:Gibbs} (Theorem~\ref{thm:gibbs-near-quadratic}) we prove a stronger result: after a logarithmic change of variables, the Gibbs distributions at any integer inverse temperature $\beta \ge 1$ yield valid sampling weights with uniform $\widetilde{O}(n^2)$ sparsity bounds. Next taking the zero-temperature limit $\beta \to +\infty$ collapses each Gibbs distribution to a single deterministic point (its ground state), completely derandomizing the weight computation and simplifying the analysis.

\paragraph{Why cool in logarithmic coordinates?}
To derandomize $\mathbb{E}_{X \sim \mu_{H_u}}[\widehat{f_i^{\to}}(u, X)]$, we seek a deterministic ground state $x^*$ so that we can replace $\mathbb{E}_{X \sim \mu_{H_u}}[\widehat{f_i^{\to}}(u, X)]$ by $\widehat{f_i^{\to}}(u, x^*)$. Naively introducing $\beta > 0$ into $e^{-\beta H_u(x)}\,\mathrm{d}x$ and sending $\beta \to +\infty$ simply rescales by $1/\beta$ and collapses to $\arg\min_{x \in \mathbb{R}_{\ge 0}^{E_u}} H_u(x) = \mathbf{0}$. Thus, we instead seek a regularizer $g(x)$ such that $x^* = \arg\min_{x \in \mathbb{R}_{>0}^{E_u}} \{ H_u(x) - g(x) \}$ yields valid importance shares $\widehat{f_i^{\to}}(u, x^*)$.

Since the Lov\'asz extension is piecewise linear (linear on every cone with a fixed coordinate ordering), it is natural to calibrate $g(x)$ on the linear case where $\widehat{f_i^{\to}}(u, \cdot)$ are linear and $H_u(x) = \sum_{v \in E_u} \alpha_v x_v$. In this case, linearity gives $\mathbb{E}_{X \sim \mu_{H_u}}[\widehat{f_i^{\to}}(u, X)] = \widehat{f_i^{\to}}(u, \mathbb{E}[X])$, where $X_v \sim \operatorname{Exp}(\alpha_v)$ has mean $\mathbb{E}[X_v] = 1/\alpha_v$ for $v \in E_u$. For the minimizer $x^* = \arg\min_{x \in \mathbb{R}_{>0}^{E_u}} \{\sum_{v \in E_u} \alpha_v x_v - g(x)\}$ to coincide with $\mathbb{E}[X] = (1/\alpha_v)_{v \in E_u}$ for all $\alpha_v > 0$, the first-order optimality condition $\alpha_v - \partial_v g(x^*) = 0$ at $x_v^* = 1/\alpha_v$ forces $\partial_v g(x) = 1/x_v$, which uniquely determines $g(x) = \sum_{v \in E_u} \log x_v$!

Probabilistically, this logarithmic barrier $g(x) = \sum_{v \in E_u} \log x_v$ arises naturally from a \emph{logarithmic change of variables} in $\mu_{H_u}$. Set $y_v := \log x_v$ for $v \in E_u$ ($x_v = e^{y_v}$); the Jacobian $\mathrm{d}x = (\prod_{v \in E_u} e^{y_v})\,\mathrm{d}y = \exp(\sum_{v \in E_u} \log x_v)\,\mathrm{d}y$ automatically produces the $\sum_{v \in E_u} \log x_v$ term, giving $e^{-H_u(x)}\,\mathrm{d}x = \exp\bigl(-(H_u(e^y) - \sum_{v \in E_u} y_v)\bigr)\,\mathrm{d}y$. Introducing the inverse temperature $\beta > 0$ into this log-coordinate density gives $\nu_\beta(y) \propto \exp\bigl(-\beta [H_u(e^y) - \sum_{v \in E_u} y_v]\bigr)$ (or $e^{-\beta H_u(x)} \prod_{v \in E_u} x_v^{\beta - 1}\,\mathrm{d}x$ in $x$-coordinates, analyzed for all $\beta \ge 1$ in Appendix~\ref{sec:Gibbs}). As $\beta \to +\infty$, $\nu_\beta$ concentrates around its unique mode, which in $x$-coordinates is the unique minimizer of the strictly convex program
\begin{equation}
x^* := \argmin_{x \in \mathbb{R}_{>0}^{E_u}} \left\{ H_u(x) - \sum_{v \in E_u} \log x_v \right\}. \label{eq:minimizer_Gibbs}
\end{equation}

\paragraph{Computing the ground state $x^*$ via submodular minimization.}
Minimizing a Lov\'asz extension plus a separable convex function as in~\eqref{eq:minimizer_Gibbs} is a classical problem related to the principal partition and lexicographically optimal base of a submodular function~\cite{Fujishige80, Fujishige05, Bach13}. To see how $x^*$ is obtained directly from the greedy formula of the Lov\'asz extension~\eqref{eq:Lovasz_ext_per}, suppose $x \in \mathbb{R}_{>0}^{E_u}$ takes $L$ distinct positive values $z_1 > z_2 > \dots > z_L > 0$ on an ordered partition $(B_1, \dots, B_L)$ of $E_u$, and let $P_0 := \varnothing$ and $P_j := B_1 \cup \dots \cup B_j$ for $j \in [L]$, so that $P_L = E_u$. Then $H_u(x) - \sum_{v \in E_u} \log x_v = \sum_{j=1}^L \bigl[(H_u(P_j) - H_u(P_{j-1})) z_j - |B_j| \log z_j\bigr]$, and setting the derivative in $z_j$ to zero gives $z_j = 1/\tau_j$, where $\tau_j = \frac{H_u(P_j) - H_u(P_{j-1})}{|B_j|}$ must satisfy $0 < \tau_1 < \tau_2 < \dots < \tau_L$. Because $H_u$ is submodular, we inductively choose $B_j \subseteq E_u \setminus P_{j-1}$ to minimize the average marginal ratio (taking the inclusion-wise largest minimizer):
\[
z_j = \frac{1}{\tau_j}, \qquad \text{where } \tau_j := \min_{\varnothing \ne B \subseteq E_u \setminus P_{j-1}} \frac{H_u(P_{j-1} \cup B) - H_u(P_{j-1})}{|B|},
\]
and set the coordinate price $b^{(u,\mathrm{in})}(v) := \tau_j$ and ground-state coordinate $x_v^* := 1/\tau_j$ for $v \in B_j$, extending both by zero to $([n] \setminus \{u\}) \setminus E_u$. This sequence has strictly increasing ratios $0 < \tau_1 < \dots < \tau_L$ and is computed in polynomial time via submodular function minimization $\min_{B \subseteq E_u \setminus P_{j-1}} \{ H_u(P_{j-1} \cup B) - t |B| \}$ (Algorithm~\ref{alg:n2-terminal-shells}). As proved in Lemma~\ref{lem:correctness_decomposition} (Section~\ref{sec:sparsification-n2}), $b^{(u,\mathrm{in})} \in \mathcal{B}(H_u)$ with $H_u(x^*) = \langle b^{(u,\mathrm{in})}, x^* \rangle$, $x^*$ uniquely minimizes~\eqref{eq:minimizer_Gibbs}, and the ground-state shares $s_i^{(u)} := \widehat{f_i^{\to}}(u, x^*)$ and prices $b^{(u,\mathrm{in})}(v) = 1/x_v^*$ for $v \in E_u$ (together with the symmetric shares $t_i^{(v)}$ and prices $b^{(v,\mathrm{out})}(u)$ from $G_v$) satisfy (P1) and (P2) without the factor of $2$ from~\eqref{eq:beta_i_density_submodular}:
\begin{equation}
\sum_{i=1}^m s_i^{(u)} \le n - 1, \qquad f_i^{\to}(u, D) \le s_i^{(u)} \cdot \max_{v \in D} b^{(u,\mathrm{in})}(v) \quad \forall D \subseteq [n] \setminus \{u\}, \label{eq:bound_f_i_x^*}
\end{equation}
and defining $\operatorname{inf}_S(w) := \max_{v \in \bar{S}} b^{(v,\mathrm{out})}(w)$ for $w \in S$ and $\operatorname{inf}_S(w) := \max_{u \in S} b^{(u,\mathrm{in})}(w)$ for $w \in \bar{S}$ gives $\Delta_i(S, U_\theta) \le \theta p_i$ and $\sum_{w \in [n]} \operatorname{inf}_S(w) \le 2 F(S)$ (Claims~\ref{clm:bound_gap} and~\ref{clm:sum_influences}). Setting $p_i := \sum_{u \in [n]} s_i^{(u)} + \sum_{v \in [n]} t_i^{(v)}$ yields our deterministic-weight algorithm for unrestricted queries (Algorithm~\ref{alg:n2-sparsification}) with sparsity $O(\varepsilon^{-2} n^2 \log(2n/\delta))$.

\subsection{Extension to sparsification under cardinality constraints ($|S| \le k$)}\label{subsec:overview-nk}

Next, we extend our framework to cardinality-constrained queries $|S| \le k$ in~\eqref{eq:def_sparsification_cardinality}, compressing the total weight from $\sum_{i=1}^m p_i = O(n^2)$ to $O(nk)$ while preserving $\Delta_i(S, U_\theta) \le \theta p_i$ for all scales $\theta \ge F(S)/k$.

\paragraph{Connection to Bencz\'ur--Karger / Quanrud strength decomposition and $|S| \le k$ asymmetry.}
Writing out $s_i^{(u)} = \widehat{f_i^{\to}}(u, x^*)$ along the chain $P_0 = \varnothing \subsetneq P_1 \subsetneq \dots \subsetneq P_L = E_u$ (where $x_v^* = 1/\tau_j$ on $B_j$) gives the explicit allocation formula and telescoped budget
\begin{equation}
s_i^{(u)} = \sum_{j=1}^L \frac{f_i^{\to}(u, P_j) - f_i^{\to}(u, P_{j-1})}{\tau_j}, \qquad \sum_{i=1}^m s_i^{(u)} = \sum_{j=1}^L \frac{H_u(P_j) - H_u(P_{j-1})}{\tau_j} = \sum_{j=1}^L |B_j| \le n - 1, \label{eq:allocation_s_i}
\end{equation}
generalizing the classical edge-strength / quotient decomposition of~\cite{BK96, Qua24} to directed block min-cuts.

When $|S| \le k \ll n$, the inside and outside sets are asymmetric: \emph{$S$ is small ($|S| \le k$), whereas $\bar{S}$ is co-small ($|\bar{S}| \ge n - k$)}. For a fixed query $S$, the ground-state coordinate prices $b_S(w)$ of the two query-level aggregate functions $A \mapsto \sum_{i=1}^m f_i^{\to}(A, \bar{S})$ on $S$ and $D \mapsto \sum_{i=1}^m f_i^{\to}(S, D)$ on $\bar{S}$ satisfy $\sum_{w \in [n]} b_S(w) = 2F(S)$ (Proposition~\ref{prop:nk-terminal-dominates-residual}). With $U_\theta := \{w \in [n] : b_S(w) \le \theta\}$, $I_\theta := S \setminus U_\theta$, and $O_\theta := \bar{S} \setminus U_\theta$, positive chaining only needs to run down to the stopping scale $\theta_J \asymp F(S)/k$: at that scale, $\sum_{w \in [n]} b_S(w) = 2F(S)$ implies $\ell_J = |I_J| + |O_J| \le 2k$, so the final residual $\Delta_i(S, U_J) \le \frac{2F(S)}{k} p_i$ depends on at most $|S| + \ell_J \le 3k$ coordinates and can be union-bounded over at most $(4n)^{3k}$ configurations (Section~\ref{sec:sparsification-nk}). Thus, we only need to bound $\Delta_i(S, U_\theta) \le \sum_{u \in S} f_i^{\to}(u, \bar{S} \setminus O_\theta) + f_i^{\to}(S \setminus I_\theta, \bar{S})$ for $\theta \ge F(S)/k$, using two complementary compressions:

\begin{enumerate}[leftmargin=*]
\item \textbf{Inside-to-Outside Families $f_i^{\to}(u, \cdot)$ ($u \in S$): Capping the Budget per $u$ at $O(k)$.}
For $u \in S$ and $D := \bar{S} \setminus O_\theta \subseteq E := [n] \setminus \{u\}$ at scale $\theta \ge F(S)/k$, since $|S \setminus \{u\}| \le k - 1$ and $|O_\theta| \le F(S)/\theta \le k$, $D$ misses at most $|E \setminus D| \le 2k - 1 =: K$ elements of $E$. In Section~\ref{sec:sparsification-nk}, we replace the cardinality penalty $|B_j| = |E \setminus P_{j-1}| - |E \setminus P_j|$ in~\eqref{eq:allocation_s_i} with the \emph{capped submodular penalty} $\kappa(P) := \min\{K,\, |E \setminus P|\}$. For $\theta \ge F(S)/k$, the uncapped objective satisfies
\[
\min_{C \subseteq E}\{H_u(C) + \theta |E \setminus C|\} \le H_u(\bar{S}) + \theta(|S| - 1) \le F(S) + \theta(k - 1) \le (2k - 1)\theta = K\theta,
\]
so the capped and uncapped objectives share the same inclusion-wise largest minimizer. Moreover, optimality of the query-level ground state and increasing differences of block minima give
\[
H_u(D) - H_u(C) \le \sum_{i=1}^m \bigl(f_i^{\to}(S,D) - f_i^{\to}(S,C)\bigr) \le \theta |D \setminus C| \qquad (\forall C \subseteq D),
\]
so $D$ is contained in this minimizer; taking $P_0 \subsetneq \dots \subsetneq P_L = E$ to be the chain of distinct largest minimizers of $H_u(P) + t\kappa(P)$ (where $|P_1| \ge n - 2k$ and $\kappa(P_{j-1}) > \kappa(P_j)$), the capped shares (with $\tau_j := \frac{H_u(P_j) - H_u(P_{j-1})}{\kappa(P_{j-1}) - \kappa(P_j)}$) satisfy $f_i^{\to}(u,D) \le \theta s_i^{(u)}$ while telescoping the per-terminal budget to $\sum_{i=1}^m s_i^{(u)} = \sum_{j=1}^L (\kappa(P_{j-1}) - \kappa(P_j)) = \kappa(P_0) \le 2k - 1$, giving $\sum_{u \in [n]} \sum_{i=1}^m s_i^{(u)} \le n(2k - 1) = O(nk)$.

\item \textbf{Outside-to-Inside Families $f_i^{\to}(\cdot, \bar{S})$: Compressing $\bar{S}$ into $O(k)$ Low-Collision Blocks.}
For $f_i^{\to}(A, \bar{S})$ ($A = S \setminus I_\theta$), summing over $v \in \bar{S}$ would incur $\Theta(n)$ outside terminals. Instead, we partition $[n]$ into $b = \Theta(k)$ blocks $X_1, \dots, X_b$ ($2k \le b < 4k$) with pairwise collision probability at most $\frac{1}{2k}$. For any query $|S| \le k$, each $v \in \bar{S}$ belongs to a block $X_c$ disjoint from $S$ with probability at least $1 - \frac{|S|}{2k} \ge \frac{1}{2}$. Letting $Z := \bigcup_{X_c \cap S = \varnothing} X_c \subseteq \bar{S}$, a standard random-restriction inequality for monotone submodular functions (Lemma~\ref{lem:residual-nk-random-restriction}: $\mathbb{E}[h(R)] \ge q\,h(E)$ when $\min_{v \in E} \Pr[v \in R] \ge q$) and subadditivity across $X_c \subseteq Z$ give
\[
f_i^{\to}(A, \bar{S}) \le 2\,\mathbb{E}\bigl[f_i^{\to}(A, Z)\bigr] \le 2\,\mathbb{E}\left[\sum_{c=1}^b \mathbf{1}[X_c \cap S = \varnothing] \, f_i^{\to}(A, X_c)\right].
\]
In Section~\ref{sec:sparsification-nk} (Claim~\ref{clm:nk-low-collision-partitions}), we derandomize this via universal linear hashing into an explicit family $\mathcal{H}$ of $< 4n$ partitions $\Pi = (X_1, \dots, X_b)$, and compute ground-state shares $t_i^{(X)}$ for $C \mapsto f_i^{\to}(C, X)$ on $[n] \setminus X$ ($\sum_i t_i^{(X)} \le n - |X|$). For any block $X$ disjoint from $S$, increasing differences and optimality of the query-level ground state give $\sum_{i=1}^m (f_i^{\to}(A,X) - f_i^{\to}(B,X)) \le \sum_{i=1}^m (f_i^{\to}(A,\bar{S}) - f_i^{\to}(B,\bar{S})) \le \theta |A \setminus B|$ for all $B \subseteq A$, so $f_i^{\to}(A,X) \le \theta t_i^{(X)}$. Defining
\[
p_i := \sum_{u \in [n]} s_i^{(u)} + \frac{2}{|\mathcal{H}|} \sum_{\Pi \in \mathcal{H}} \sum_{X \in \Pi} t_i^{(X)}
\]
gives $\sum_{i=1}^m p_i \le n(2k-1) + 2(b-1)n \le 10nk$, yielding the $O(\varepsilon^{-2} nk \log(2n/\delta))$ sparsifier of Theorem~\ref{thm:small-query-submodular-sparsification}.
\end{enumerate}

\section{Preliminaries}
\label{sec:preliminaries}
 Given an error $\epsilon$, for two real numbers $a$ and $b$, we use $a=b \pm\epsilon$ to denote $a \in [(1-\epsilon)\cdot b, (1+\epsilon)\cdot b]$.

Let $[n]:=\{1,2,\ldots,n\}$ and ${[n] \choose k}$ denote the family of size-$k$ subsets in $[n]$. Similarly, let ${[n] \choose \le k}$ denote the family of subsets with maximum size most $k$. As mentioned earlier, we identify Boolean vectors in $\{0,1\}^n$ with subsets of $[n]$. Hence, we use $S$ both for a subset of $[n]$ and for its corresponding Boolean vector. Moreover, $x \in {[n] \choose \le k}$ is equivalent to $x \in \{0,1\}^n$ has $|x| \le k$.

 For vectors in $\mathbb{R}^n$, we use $x$ to denote a vector of real numbers and $\mathbf{1}_S$ to denote the corresponding indicator vector in $\mathbb{R}^n$.

For a vector $x \in \mathbb{R}^n$, let $x(i)$ denote its $i$-th coordinate and let $x(T)$ denote the sub-vector in $\mathbb{R}^T$ for $T \subseteq [n]$. For two vectors $x,y \in \mathbb{R}^n$, we write $x\geq y$ to indicate $x(i)\geq y(i)$ for all $i\in [n]$. For an event $E$, let $\mathbb{I}(E) \in \{0,1\}$ denote the indicator function. 

\paragraph{Concentration bound for relative errors.} 

We will use the following concentration bound extensively in this work. While our analysis relies on chaining arguments, we use Lemma~\ref{lem:Bernstein} instead of (sub-)Gaussian variables and Gaussian processes for simplicity. In particular, because $F(S)=\sum_i f_i(S)$ as a summation of positive functions is not normalized, Lemma~\ref{lem:Bernstein} is stated in the mixed multiplicative-additive form. It is a reformulation of the Chernoff bound in \cite{Qua24}. For completeness, we provide a self-contained proof in Appendix~\ref{sec:proof_Bernstein}.

\begin{lemma}\label{lem:Bernstein}
    Let $p_1,\ldots,p_m$ be positive weights and $\Lambda>1$. Suppose each $f_i$ is sampled independently with probability $\min\{1,p_i\Lambda\}$. Define the random weight
    \[
        W_i:=
        \begin{cases}
            \dfrac{1}{\min\{1,p_i\Lambda\}}, & \text{if } f_i \text{ is sampled},\\[6pt]
            0, & \text{otherwise}.
        \end{cases}
    \]
    For any $z_1,\ldots,z_m\in\R_{\ge0}$, define $Y_i:=(W_i-1)z_i$ and $R:=\max_{i\in[m]}\frac{z_i}{p_i}$. Then, for any $0<\eta\le1$ and $x>0$,
    \[
        \Pr\left[\left|\sum_{i=1}^m Y_i\right|>\eta\sum_{i=1}^m z_i+\frac{xR}{\eta\Lambda}\right]\le2e^{-x}.
    \]
\end{lemma}

\paragraph{Submodular functions.} A function $f:\{0,1\}^n \xrightarrow{} \mathbb{R}$ is submodular iff $f(S) + f(T) \ge f(S \cup T)+ f(S \cap T)$ for any $S \subseteq [n]$ and $T \subseteq [n]$. Equivalently, $f$ has diminishing marginal returns: $f(S \cup \{i\})-f(S) \ge f(T \cup \{i\})-f(T)$ for any $S \subseteq T$ and $i\notin T$.

The Lov\'{a}sz extension \cite{Lovasz83} of a submodular function $f$ with $f(\varnothing)=0$ is
\begin{equation}\label{eq:Lovasz_ext}
    \widehat f(z):=\int_0^{\infty} f\big( \{i:z(i) \ge \lambda\} \big) \mathrm{d} \lambda \text{ for } z \in \R_{\ge 0}^n.
\end{equation}
Equivalently, let $\pi$ denote the permutation of $[n]$ with $z\big( \pi_1 \big) \ge z\big( \pi_2 \big) \ge \cdots \ge z\big( \pi_n \big)$ and $\mathrm{Pre}_\pi(k)=\{\pi_1,\pi_2, \cdots,\pi_k\}$ such that

\begin{equation}\label{eq:Lovasz_ext_per}
\widehat f(z)=\sum_{i=1}^{n-1} \bigg( z\big( \pi_i \big)-z\big( \pi_{i+1} \big) \bigg) \cdot f\big( \mathrm{Pre}_\pi(i) \big) + z\big( \pi_n \big) \cdot f\big( \mathrm{Pre}_\pi(n) \big).
\end{equation}

We will use these two equivalent definitions alternatively. Here are basic properties of the Lov\'{a}sz Extension $\widehat{f}$.

\begin{enumerate}
    \item If $f$ is monotone, then $\widehat{f}$ is monotone: $\widehat{f}(x) \ge \widehat{f}(y)$ for any $x \ge y$.
    \item $\widehat f(\mathbf{1}_S)=f(S), \forall S\subseteq [n]$.
    \item The Lov\'{a}sz extension $\widehat{f}$ is convex iff $f$ is submodular.
    \item $\widehat{f}(cx) = c\widehat{f}(x)$ for any scalar $c \ge 0$ and vector $x$.
\end{enumerate}

For a submodular function $f$, let its \emph{base polytope} in $\mathbb{R}^n$ be 
\[
        \mathcal{B}(f) := \left\{ y \in \mathbb{R}^n \;\middle|\; \sum_{v \in A} y(v) \le f(A)\ \forall A \subseteq [n],\ \sum_{v \in [n]} y(v) = f([n]) \right\}.\]
Then $\wh{f}(x) = \max_{s \in \mathcal{B}(f)} \langle s ,x \rangle$ and its sub-gradient $\partial \wh{f}(x)=\operatorname{conv}\{\underset{s \in \mathcal{B}(f)}{\arg\max} \langle s ,x \rangle\}$ is a convex set.

\paragraph{Submodular Minimization.} We use  submodular minimization algorithms as a primitive to compute the sampling weights $p_i$. Let $\mathcal{T}_{min}(f)$ denote the time complexity of minimizing $f(S)$. We state the best bounds \cite{JLSZ24} on $\mathcal{T}_{min}(f)$.
\begin{lemma}\label{lem:submodular_min}    
    For any submodular function $f$, 
    $\mathcal{T}_{min}(f) \le O(n^3 \log n \cdot \mathcal{T}_{eval}(f) + n^4 \log n)$.

\end{lemma}

\paragraph{Properties of Directed Block Min-Cut Functions}
We show two useful properties about the block min-cut functions. Recall that its definition is 
\begin{equation}\label{eq:def_block_min_cut}
    f^\to(A,B):=\min_{S: A \subseteq S, B \cap S=\emptyset} f(S) \text{ for any submodular } f.
\end{equation}
\begin{lemma}
 \label{lem:block-minimum-separate-submodularity}
 Given any submodular $f:\{0,1\}^n \to \R_{\ge 0}$, $f^\to$ is monotone non-decreasing and submodular: for any fixed $A \subseteq [n]$,
    \begin{enumerate}
        \item $f^{\to}(A,B)$ is monotone non-decreasing function for $B \subseteq [n] \setminus A$.
        \item $f^{\to}(A,B)$ is submodular for $B \subseteq [n] \setminus A$.
    \end{enumerate}  
    The same is true when we fix $B$ and consider $f^{\to}(A,B)$ as a function of $A \subseteq [n] \setminus B$.
\end{lemma}

Because $f^\to(S,\overline{S})=f(S)$, a simple corollary of Lemma~\ref{lem:block-minimum-separate-submodularity} is $ \max_{u \in S} f^\to(u,\overline{S})\le f(S) \le \sum_{u \in S} f^\to(u,\overline{S})$. Next we provide finer methods to bound $f(S)-f(T)$.

Since our chaining of $f(S)=f^\to(S,\overline{S})$ is based on $f^{\to}(I,O)$ where $I \subset S$ and $O \subset \overline{S}$ records important indices with respect to $S$. We state a useful bound about $f^\to(S,\overline{S})-f^{\to}(I,O)$. It would be easier to consider $U=\overline{I \cup O}$ as the set of free variables and the residual $\Delta(S,U):=f(S)-f^{\to}(I,O)=f(S) - f^{\to}(S\setminus U,\bar S\setminus U)$ in the analyses of Section~\ref{sec:residual-n2} and Section~\ref{sec:residual-nk}.
\begin{lemma}
\label{lem:block-residual-two-block}
Consider a fixed submodular function $f:\{0,1\}^V \to \R_{ \ge 0}$ with $f(\varnothing)=f(V)=0$. For any $S$ and $T$,  $f(S)-f(T) \le f^\to(S \setminus T,\overline{S}) + f^\to(S \cap T, T \setminus S)$. Moreover, for any $S$ and $U$, the residual $\Delta(S,U):=f(S)-f^{\to}(S\setminus U,\bar S\setminus U)$ satisfies
\begin{enumerate}
    \item $0 \le \Delta(S,U) \le  f^{\to}(S\cap U,\bar S)+f^{\to}(S,U\setminus S)$.
    \item $\Delta(S,U)$ is monotone with respect to $U$.
    \item $\Delta(S,\emptyset)=0$ and $\Delta(S,V)=f(S)$.
\end{enumerate}
\end{lemma}

\begin{figure}[htbp]
\centering
\begin{tikzpicture}[font=\small]
    \fill[gray!5] (-2.85,-1.87) rectangle (2.85,1.83);
    \draw[gray!60] (-2.85,-1.87) rectangle (2.85,1.83);
    \fill[blue!13] (-0.86,0) circle (1.5);
    \fill[orange!13] (0.86,0) circle (1.5);
    \begin{scope}
        \clip (-0.86,0) circle (1.5);
        \fill[teal!22] (0.86,0) circle (1.5);
    \end{scope}
    \draw[thick,blue!55!black] (-0.86,0) circle (1.5);
    \draw[thick,orange!60!black] (0.86,0) circle (1.5);
    \node at (-1.58,0) {$S\setminus T$};
    \node at (0,0) {$S\cap T$};
    \node at (1.58,0) {$T\setminus S$};
    \node at (-1.15,1.16) {$S$};
    \node at (1.15,1.16) {$T$};
    \node[font=\scriptsize,gray!70!black] at (0,-1.68) {$V\setminus(S\cup T)$};
    \node[font=\scriptsize,gray!70!black,anchor=north west] at (-2.78,1.75) {$V$};
\end{tikzpicture}
\caption{The regions used in Lemma~\ref{lem:block-residual-two-block}. A minimizer $X$ of $f^\to(S\setminus T,\bar S)$ contains the left region and lies within $S$. A minimizer $Y$ of $f^\to(S\cap T,T\setminus S)$ contains the overlap and excludes the right region. Thus $X\cup(T\cap Y)=S$, and submodularity and nonnegativity give $f(S)=f(X\cup(T\cap Y))\le f(X)+f(T\cap Y)\le f(X)+f(T)+f(Y)$.}
\label{fig:block-residual-two-block}
\end{figure}

Figure~\ref{fig:block-residual-two-block} illustrates the bound $f(S)-f(T)\le f^\to(S\setminus T,\bar S)+f^\to(S\cap T,T\setminus S)$ from Lemma~\ref{lem:block-residual-two-block}. Formal proofs of Lemma~\ref{lem:block-minimum-separate-submodularity} and Lemma~\ref{lem:block-residual-two-block} are deferred to Appendix~\ref{appendix:proofs} for completeness.

\section{Submodular Sparsification without Cardinality Constraints}
\label{sec:residual-n2}
\label{sec:sparsification-n2}

In this section, we present our polynomial-time deterministic-weight algorithm for sparsifying a sum $F = \sum_{i=1}^m f_i$ of general nonnegative submodular functions without cardinality constraints, achieving near-quadratic sparsity $O(\varepsilon^{-2} n^2 \log(2n/\delta))$. By the preprocessing reduction in Appendix~\ref{sec:preprocess} (Lemma~\ref{lem:endpoint-normalization}), we assume throughout this section that every component is normalized at the endpoints, i.e., $f_i(\varnothing) = f_i([n]) = 0$ for all $i \in [m]$.

\begin{theorem}
\label{thm:near-quadratic-submodular-sparsification}
Let $F = \sum_{i=1}^m f_i$ be the sum of $m$ nonnegative submodular functions $f_1, \dots, f_m : \{0,1\}^n \to \mathbb{R}_{\ge 0}$ with $f_i(\varnothing) = f_i([n]) = 0$. Algorithm~\ref{alg:n2-sparsification} (using Algorithm~\ref{alg:n2-terminal-shells} as a subroutine) deterministically computes nonnegative sampling weights $p_1, \dots, p_m$ satisfying $\sum_{i=1}^m p_i = O(n^2)$ (specifically, $\sum_{i=1}^m p_i \le 2n(n-1)$ for normalized inputs) in time
\[
    O\bigl(mn^3 \mathcal{T}_{\min}(n)^2 (\mathcal{T}_{\mathrm{eval}} + n)\bigr).
\]
Moreover, given any $\varepsilon, \delta \in (0,1)$, set $\Lambda := 56\varepsilon^{-2}\log(32n/\delta)$ and construct the reweighted sum $\widetilde{F} := \sum_{i=1}^m w_i f_i$ by independently setting each $w_i = 1/\min\{1, \Lambda p_i\}$ with probability $\min\{1, \Lambda p_i\}$ and $w_i = 0$ otherwise. Then with probability at least $1 - \delta$, $\widetilde{F}$ is an $\varepsilon$-sparsifier of $F$---i.e., $|\widetilde{F}(S) - F(S)| \le \varepsilon F(S)$ simultaneously for all $S \subseteq [n]$---with sparsity $|\supp(w)| = O(\varepsilon^{-2} n^2 \log(2n/\delta))$.
\end{theorem}

As outlined in Section~\ref{sec:overview}, Algorithm~\ref{alg:n2-sparsification} computes the importance weight $p_i$ of each component $f_i$ by decomposing $2n$ monotone directed block min-cut families. Specifically, for each vertex $u \in [n]$, it considers the inside-to-outside family $\bigl(f_i^{\to}(u, \cdot)\bigr)_{i=1}^m$ and the outside-to-inside family $\bigl(f_i^{\to}(\cdot, u)\bigr)_{i=1}^m$ on the ground set $E = [n] \setminus \{u\}$, and invokes the subroutine \textsc{Decomposition} (Algorithm~\ref{alg:n2-terminal-shells}) to compute their ground-state component shares $\bigl(s_i^{(u)}\bigr)_{i=1}^m$ and $\bigl(t_i^{(u)}\bigr)_{i=1}^m$, each having total weight at most $|E| = n - 1$. Summing these $2n$ share vectors gives $p_i = \sum_{u \in [n]} s_i^{(u)} + \sum_{v \in [n]} t_i^{(v)}$ with $\sum_{i=1}^m p_i \le 2n(n-1)$, after which the algorithm independently samples each component $f_i$ with probability $\min\{1, \Lambda p_i\}$.

\begin{algorithm}[hbt!]
\caption{Sparsification of a sum of nonnegative submodular functions}
\label{alg:n2-sparsification}
\begin{algorithmic}[1]
\Procedure{ComputeSampling}{$f_1, \dots, f_m : \{0,1\}^n \to \mathbb{R}_{\ge 0},\, \varepsilon,\, \delta$}
 \ForAll{$u \in [n]$}
  \State $E \gets [n] \setminus \{u\}$
  \State $(s_1^{(u)}, \dots, s_m^{(u)}) \gets \Call{Decomposition}{E,\, (f_i^{\to}(u, \cdot))_{i=1}^m}$
  \State $(t_1^{(u)}, \dots, t_m^{(u)}) \gets \Call{Decomposition}{E,\, (f_i^{\to}(\cdot, u))_{i=1}^m}$
 \EndFor
 \State $p_i \gets \sum_{u \in [n]} s_i^{(u)} + \sum_{v \in [n]} t_i^{(v)}$ for every $i \in [m]$
 \State $\Lambda \gets 56\varepsilon^{-2}\log(32n/\delta)$
 \ForAll{$i \in [m]$}
  \State Independently draw $\xi_i \sim \operatorname{Bernoulli}(\min\{1, \Lambda p_i\})$
  \State $w_i \gets 1 / \min\{1, \Lambda p_i\}$ if $\xi_i = 1$, and $0$ otherwise
 \EndFor
 \State \Return $(w_1, \dots, w_m)$
\EndProcedure
\end{algorithmic}
\end{algorithm}

The remainder of this section is organized into three subsections:
\begin{itemize}[leftmargin=*]
    \item In \textbf{Section~\ref{sec:n2-decomposition}}, we present and analyze the subroutine \textsc{Decomposition} (Algorithm~\ref{alg:n2-terminal-shells}) for an arbitrary family of nonnegative monotone submodular functions $h_1, \dots, h_m : 2^E \to \mathbb{R}_{\ge 0}$. Along with the component shares $s_i = \widehat{h}_i(x^*)$ (where $x^*$ is the zero-temperature ground state of $H = \sum_{i=1}^m \widehat{h}_i$), the greedy minimum-ratio chain of \textsc{Decomposition} naturally assigns each coordinate $v \in E$ a price $b(v) = 1/x^*(v)$ such that $b \in \mathcal{B}(h)$ and $h_i(C) \le s_i \max_{v \in C} b(v)$ for all $C \subseteq E$ (Lemma~\ref{lem:correctness_decomposition}). We then prove that \textsc{Decomposition} runs in polynomial time via divide-and-conquer submodular minimization (Claim~\ref{clm:n2-decomposition-recovery}) and that its coordinate prices are monotone under marginal dominance (Lemma~\ref{lem:n2-price-comparison}).
    \item In \textbf{Section~\ref{sec:proofs1}}, we use the coordinate price vectors associated with the $2n$ terminal families $\bigl(f_i^{\to}(u, \cdot)\bigr)_{i=1}^m$ and $\bigl(f_i^{\to}(\cdot, v)\bigr)_{i=1}^m$ to define the query-specific coordinate influences $\operatorname{inf}_S(w)$, and prove the two core properties required by our chaining framework: the residual bound \textbf{(P1)}, $\Delta_i(S, U_\theta) \le \theta p_i$ (Claim~\ref{clm:bound_gap}), and the total influence budget \textbf{(P2)}, $\sum_{w \in [n]} \operatorname{inf}_S(w) \le 2F(S)$ (Claim~\ref{clm:sum_influences}).
    \item In \textbf{Section~\ref{sec:n2-chaining-proof}}, we combine \textbf{(P1)} and \textbf{(P2)} with the positive chaining argument and Lemma~\ref{lem:Bernstein} to complete the proof of Theorem~\ref{thm:near-quadratic-submodular-sparsification}.
\end{itemize}

\subsection{Ground-state decomposition for monotone submodular families}
\label{sec:n2-decomposition}
\label{sec:connections_Gibbs}

Let $E$ be a finite ground set, and let $h_1, \dots, h_m : 2^E \to \mathbb{R}_{\ge 0}$ be nonnegative monotone submodular functions satisfying $h_i(\varnothing) = 0$. Write $h := \sum_{i=1}^m h_i$ for their sum and $H := \sum_{i=1}^m \widehat{h}_i$ for the sum of their Lov\'asz extensions on $\mathbb{R}_{\ge 0}^E$. (In Algorithm~\ref{alg:n2-sparsification}, we apply this general setup to $E = [n] \setminus \{u\}$ with $h_i = f_i^{\to}(u, \cdot)$ or $h_i = f_i^{\to}(\cdot, u)$.)

\paragraph{Greedy minimum-ratio chain, prices, and shares.}
We decompose $E$ into an ordered chain of subsets $P_0 \subsetneq P_1 \subsetneq \dots \subsetneq P_L = E$ as follows:
\begin{enumerate}[leftmargin=*]
    \item Let $P_0 \subseteq E$ be the inclusion-wise largest zero set of $h$, i.e., the largest subset with $h(P_0) = 0$. (Because $h$ is nonnegative and submodular, its zero sets are closed under union: $h(A) = h(B) = 0 \implies 0 \le h(A \cup B) \le h(A) + h(B) = 0$; hence $P_0$ is well defined.)
    \item For $j = 1, 2, \dots, L$, while $P_{j-1} \subsetneq E$, choose a nonempty subset $B_j \subseteq E \setminus P_{j-1}$ that minimizes the average marginal increment
    \begin{equation}\label{eq:n2-min-ratio-step}
        \tau_j := \min_{\varnothing \ne B \subseteq E \setminus P_{j-1}} \frac{h(P_{j-1} \cup B) - h(P_{j-1})}{|B|},
    \end{equation}
    breaking ties by taking the inclusion-wise largest minimizer, and set $P_j := P_{j-1} \cup B_j$. Repeat until $P_L = E$ (if $P_0 = E$, the chain has $L = 0$ blocks).
\end{enumerate}
From this chain, we define the \emph{coordinate price vector} $b_h \in \mathbb{R}_{\ge 0}^E$ (written simply as $b$ when $h$ is clear from context), the \emph{ground-state vector} $x^* \in \mathbb{R}_{\ge 0}^E$, and the \emph{component shares} $s_1, \dots, s_m \in \mathbb{R}_{\ge 0}$ by
\begin{equation}\label{eq:n2-prices-and-shares-def}
    b_h(v) := \begin{cases}
        0, & v \in P_0, \\
        \tau_j, & v \in B_j,
    \end{cases}
    \qquad
    x^*(v) := \begin{cases}
        0, & v \in P_0, \\
        1/\tau_j, & v \in B_j,
    \end{cases}
    \qquad
    s_i := \sum_{j=1}^L \frac{h_i(P_j) - h_i(P_{j-1})}{\tau_j} \quad (i \in [m]),
\end{equation}
where $s_i = 0$ if $P_0 = E$. The following lemma establishes the structural and variational properties of this decomposition.

\begin{lemma}[Properties of the ground-state decomposition]
\label{lem:correctness_decomposition}
Let $h_1, \dots, h_m : 2^E \to \mathbb{R}_{\ge 0}$ be monotone submodular functions with $h_i(\varnothing) = 0$, and let $P_0, \dots, P_L$, $B_1, \dots, B_L$, $\tau_1, \dots, \tau_L$, $b_h$, $x^*$, and $(s_1, \dots, s_m)$ be defined by~\eqref{eq:n2-min-ratio-step} and~\eqref{eq:n2-prices-and-shares-def}. Write $E' := E \setminus P_0$. Then:
\begin{enumerate}[leftmargin=*]
    \item \textbf{(Strictly increasing ratios and base polytope)} If $E' \ne \varnothing$, then $0 < \tau_1 < \tau_2 < \dots < \tau_L$. Moreover, $b_h$ belongs to the base polytope $\mathcal{B}(h)$ and is tight on every prefix $P_j$:
    \begin{equation}\label{eq:n2-price-base-inequalities}
        \sum_{v \in C} b_h(v) \le h(C) \quad (\forall C \subseteq E), \qquad \sum_{v \in P_j} b_h(v) = h(P_j) \quad (0 \le j \le L), \qquad H(x^*) = \langle b_h, x^* \rangle.
    \end{equation}
    \item \textbf{(Canonical parametric minimizers)} For every parameter $t \ge 0$, the minimizers of $h(C) - t|C|$ over $C \subseteq E$ are closed under union, and the inclusion-wise largest minimizer $P_h(t) := \operatorname*{argmin}_{C \subseteq E} \{h(C) - t|C|\}$ is given by the low-price prefix
    \begin{equation}\label{eq:n2-canonical-prefix}
        P_h(t) = \{v \in E : b_h(v) \le t\}.
    \end{equation}
    \item \textbf{(Zero-temperature Gibbs ground state)} With coordinates outside $E'$ set to zero in $H$, the restriction $x^*|_{E'}$ is the unique minimizer of the strictly convex objective
    \begin{equation}\label{eq:n2-log-gibbs-objective}
        H(y) - \sum_{v \in E'} \log y(v), \qquad y \in \mathbb{R}_{>0}^{E'},
    \end{equation}
    or equivalently the unique maximizer of the logarithmically transformed Gibbs density $\nu(z)\propto\exp\left(-H(e^z)+\sum_{v\in E'}z(v)\right)$.

    \item \textbf{(Component shares, budget, and low-price bound)} For every $i \in [m]$, the component share $s_i$ satisfies
    \begin{equation}\label{eq:n2-low-price-shares}
        s_i = \widehat{h}_i(x^*), \qquad \sum_{i=1}^m s_i = |E'| \le |E|, \qquad \text{and} \qquad h_i(C) \le s_i \cdot \max_{v \in C} b_h(v) \quad (\forall C \subseteq E),
    \end{equation}
    where the maximum over the empty set is defined to be zero.
\end{enumerate}
\end{lemma}

\begin{proof}
At each step $j \ge 1$, the minimizers of~\eqref{eq:n2-min-ratio-step} are the nonempty zero sets of the nonnegative submodular function $B \mapsto h(P_{j-1} \cup B) - h(P_{j-1}) - \tau_j |B|$ on $E \setminus P_{j-1}$; since these zero sets are closed under union, the inclusion-wise largest minimizer $B_j$ is uniquely defined. If $E' = \varnothing$, then $h_i \equiv 0$ for all $i \in [m]$ and all assertions hold trivially with $b_h = x^* = 0$ and $s_i = 0$. Assume henceforth that $E' \ne \varnothing$.

\smallskip\noindent\textbf{Proof of Item 1.}
By the maximality of the zero set $P_0$, we have $h(P_0 \cup B) > 0$ for every nonempty $B \subseteq E \setminus P_0$, so $\tau_1 > 0$. Suppose for contradiction that $\tau_{j+1} \le \tau_j$ for some $j \in [L-1]$. Since $B_j$ and $B_{j+1}$ are disjoint, the average marginal increment of $B_j \cup B_{j+1}$ from $P_{j-1}$ is
\[
    \frac{h(P_{j+1}) - h(P_{j-1})}{|B_j| + |B_{j+1}|} = \frac{\tau_j |B_j| + \tau_{j+1} |B_{j+1}|}{|B_j| + |B_{j+1}|} \le \tau_j.
\]
If $\tau_{j+1} < \tau_j$, this average is strictly smaller than $\tau_j$, contradicting the optimality of $B_j$ in~\eqref{eq:n2-min-ratio-step}; if $\tau_{j+1} = \tau_j$, then $B_j \cup B_{j+1}$ achieves the minimum ratio $\tau_j$, contradicting the inclusion-wise maximality of $B_j$. Thus $0 < \tau_1 < \tau_2 < \dots < \tau_L$.

Next, for any $C \subseteq E$, set $C_j := C \cap P_j$ for $0 \le j \le L$. Submodularity of $h$ and the minimum-ratio property of $B_j$ imply that for each $j \in [L]$ (whether $C \cap B_j$ is empty or not),
\[
    h(C_j) - h(C_{j-1}) \ge h\bigl(P_{j-1} \cup (C \cap B_j)\bigr) - h(P_{j-1}) \ge \tau_j |C \cap B_j| = \sum_{v \in C \cap B_j} b_h(v).
\]
Since $C_0 \subseteq P_0$ and $h$ is nonnegative and monotone with $h(P_0) = 0$, we have $h(C_0) = 0 = \sum_{v \in C_0} b_h(v)$. Summing over $j = 1, \dots, L$ gives $\sum_{v \in C} b_h(v) \le h(C)$ for all $C \subseteq E$. When $C = P_j$, both inequalities in each step $1 \le r \le j$ hold with equality by definition of $\tau_r$, giving $\sum_{v \in P_j} b_h(v) = h(P_j)$ (and in particular $\sum_{v \in E} b_h(v) = h(E)$), so $b_h \in \mathcal{B}(h)$. Finally, because $x^*$ takes the strictly decreasing positive values $1/\tau_1 > \dots > 1/\tau_L > 0$ on $B_1, \dots, B_L$ and vanishes on $P_0$, the greedy formula~\eqref{eq:Lovasz_ext_per} for the Lov\'asz extension gives
\begin{equation}\label{eq:n2-cooling-subgradient}
    H(x^*) = \sum_{j=1}^L \frac{h(P_j) - h(P_{j-1})}{\tau_j} = \sum_{j=1}^L |B_j| = \langle b_h, x^* \rangle.
\end{equation}

\smallskip\noindent\textbf{Proof of Item 2.}
For any $t \ge 0$, the function $C \mapsto h(C) - t|C|$ is submodular, so its minimizers are closed under union. Let $P := \{v \in E : b_h(v) \le t\}$, which equals the prefix $P_j$ where $j \in \{0, \dots, L\}$ is the largest index with $\tau_j \le t$ (taking $\tau_0 := 0$). By Item~1, for every $C \subseteq E$,
\[
    h(C) - t|C| \ge \sum_{v \in C} \bigl(b_h(v) - t\bigr) \ge \sum_{v \in P} \bigl(b_h(v) - t\bigr) = h(P) - t|P|,
\]
where the second inequality holds because $P$ contains all elements with $b_h(v) - t \le 0$ and excludes all elements with $b_h(v) - t > 0$, and the final equality uses $\sum_{v \in P_j} b_h(v) = h(P_j)$. Moreover, if $C \not\subseteq P$, then $C$ contains at least one element $v \in E \setminus P$ with $b_h(v) > t$, making the second inequality strict. Hence $P$ is the inclusion-wise largest minimizer $P_h(t)$.

\smallskip\noindent\textbf{Proof of Item 3.}
Every coordinate $v \in P_0$ satisfies $h_i(\{v\}) = 0$ for all $i \in [m]$ and hence has zero marginal contribution everywhere. On $E'$, we have $h(\{v\}) > 0$ for all $v \in E'$ (by maximality of $P_0$), so $H(y) \ge \max_{v \in E'} h(\{v\}) y(v)$ grows at least linearly as $\|y\|_\infty \to \infty$, while $-\sum_{v \in E'} \log y(v)$ is strictly convex and diverges at the boundary of $\mathbb{R}_{>0}^{E'}$. Under the logarithmic change of variables $z(v) = \log y(v)$ on $E'$, we have $e^{-H(y)} \prod_{v \in E'} \mathrm{d}y(v) = \exp\bigl(-H(y) + \sum_{v \in E'} \log y(v)\bigr) \prod_{v \in E'} \mathrm{d}(\log y(v))$, so maximizing the transformed Gibbs density is equivalent to minimizing~\eqref{eq:n2-log-gibbs-objective}. By Item~1, $b_h \in \mathcal{B}(h)$ and $H(x^*) = \langle b_h, x^* \rangle$, which implies $b_h|_{E'} \in \partial H(x^*|_{E'})$. Since $b_h(v) = 1/x^*(v)$ for all $v \in E'$, $x^*|_{E'}$ satisfies the first-order optimality condition $0 \in \partial H(x^*|_{E'}) - (1/x^*(v))_{v \in E'}$, and strict convexity guarantees that $x^*|_{E'}$ is the unique minimizer.

\smallskip\noindent\textbf{Proof of Item 4.}
Evaluating the Lov\'asz extension $\widehat{h}_i$ at $x^*$ via~\eqref{eq:Lovasz_ext_per} along the same ordered partition $(B_1, \dots, B_L, P_0)$ (using $h_i(P_0) = 0$) immediately gives $\widehat{h}_i(x^*) = \sum_{j=1}^L \frac{h_i(P_j) - h_i(P_{j-1})}{\tau_j} = s_i$. Summing over $i \in [m]$ gives $\sum_{i=1}^m s_i = H(x^*) = \sum_{j=1}^L |B_j| = |E'| \le |E|$. Finally, for any nonempty $C \subseteq E$, let $t := \max_{v \in C} b_h(v)$. Then $C \subseteq \{v \in E : b_h(v) \le t\} = P_h(t)$ by~\eqref{eq:n2-canonical-prefix}. Monotonicity of $h_i$ and $h_i(P_0) = 0$ then yield
\[
    h_i(C) \le h_i\bigl(P_h(t)\bigr) = \sum_{j :\, \tau_j \le t} \bigl(h_i(P_j) - h_i(P_{j-1})\bigr) \le t \sum_{j=1}^L \frac{h_i(P_j) - h_i(P_{j-1})}{\tau_j} = t s_i = s_i \cdot \max_{v \in C} b_h(v),
\]
and the bound for $C = \varnothing$ holds trivially since $h_i(\varnothing) = 0$.
\end{proof}

\paragraph{Polynomial-time computation via divide-and-conquer.}
Property~\eqref{eq:n2-canonical-prefix} of Lemma~\ref{lem:correctness_decomposition} shows that for any parameter $t \ge 0$, the inclusion-wise largest minimizer $P_h(t) = \operatorname*{argmin}_{C \subseteq E} \{h(C) - t|C|\}$ is always one of the prefixes $P_j$ in the chain $P_0 \subsetneq P_1 \subsetneq \dots \subsetneq P_L = E$. This allows Algorithm~\ref{alg:n2-terminal-shells} to recover the entire chain in $O(|E|)$ calls to \textsc{FindMinimizer} via a standard divide-and-conquer procedure \textsc{Recover}.

\begin{algorithm}[hbt!]
\caption{Computing the ground-state decomposition of a monotone submodular family}
\label{alg:n2-terminal-shells}
\label{alg:n2-canonical-state}
\begin{algorithmic}[1]
\Procedure{Decomposition}{$E,\, (h_i)_{i=1}^m$} \Comment{Return the component shares $(s_1, \dots, s_m)$}
 \State $h \gets \sum_{i=1}^m h_i$
 \State $P_0 \gets \Call{FindMinimizer}{h, 0}$
 \State $(P_0, \dots, P_L) \gets \Call{Recover}{h, P_0, E}$
 \For{$j = 1, \dots, L$}
  \State $B_j \gets P_j \setminus P_{j-1}$
  \State $\tau_j \gets \dfrac{h(P_j) - h(P_{j-1})}{|B_j|}$
 \EndFor
 \State $s_i \gets \displaystyle\sum_{j=1}^L \dfrac{h_i(P_j) - h_i(P_{j-1})}{\tau_j}$ for every $i \in [m]$
 \State \Return $(s_1, \dots, s_m)$
\EndProcedure
\Procedure{Recover}{$h, A, B$} \Comment{Return the subchain of prefixes from $A$ to $B$}
 \If{$A = B$}
  \State \Return $(A)$
 \EndIf
 \State $t \gets \dfrac{h(B) - h(A)}{|B| - |A|}$
 \State $C \gets \Call{FindMinimizer}{h, t}$
 \If{$C = B$}
  \State \Return $(A, B)$
 \EndIf
 \State \Return $\Call{Recover}{h, A, C} \cup \Call{Recover}{h, C, B}$ \Comment{Concatenate chains, keeping one copy of $C$}
\EndProcedure
\Procedure{FindMinimizer}{$h : 2^E \to \mathbb{R}_{\ge 0},\, t$} \Comment{Return the inclusion-wise largest minimizer $P_h(t)$}
 \State $\mu \gets \min_{C \subseteq E} \{h(C) - t|C|\}$
 \State $Q \gets \varnothing$
 \ForAll{$v \in E$}
  \State $\mu' \gets \min_{C :\, Q \cup \{v\} \subseteq C \subseteq E} \{h(C) - t|C|\}$
  \If{$\mu' = \mu$}
   \State $Q \gets Q \cup \{v\}$
  \EndIf
 \EndFor
 \State \Return $Q$
\EndProcedure
\end{algorithmic}
\end{algorithm}

\begin{claim}[Correctness and running time of Algorithm~\ref{alg:n2-terminal-shells}]
\label{clm:n2-decomposition-recovery}
Given monotone submodular functions $h_1, \dots, h_m : 2^E \to \mathbb{R}_{\ge 0}$ with $h_i(\varnothing) = 0$, Algorithm~\ref{alg:n2-terminal-shells} computes the exact chain $(P_0, \dots, P_L)$ and component shares $(s_1, \dots, s_m)$ defined in~\eqref{eq:n2-min-ratio-step}--\eqref{eq:n2-prices-and-shares-def} using $O(|E|)$ calls to \textsc{FindMinimizer}. When applied to any of the $2n$ terminal families in Algorithm~\ref{alg:n2-sparsification} (where $|E| = n - 1$), Algorithm~\ref{alg:n2-terminal-shells} runs in deterministic time
\[
    O\bigl(mn^2 \mathcal{T}_{\min}(n)^2 (\mathcal{T}_{\mathrm{eval}} + n)\bigr).
\]
\end{claim}

\begin{proof}
Because the minimizers of $C \mapsto h(C) - t|C|$ are closed under union, an element $v \in E$ belongs to the inclusion-wise largest minimizer $P_h(t)$ (together with the already identified subset $Q \subseteq P_h(t)$) if and only if $\min_{Q \cup \{v\} \subseteq C \subseteq E} \{h(C) - t|C|\} = \mu$. Thus \textsc{FindMinimizer}$(h, t)$ returns $P_h(t)$; in particular, at $t = 0$ it returns $P_h(0) = P_0$.

Next, suppose \textsc{Recover}$(h, A, B)$ is called on two prefixes $A = P_a$ and $B = P_b$ with $0 \le a < b \le L$. By~\eqref{eq:n2-price-base-inequalities} and strict monotonicity $\tau_{a+1} < \dots < \tau_b$, the average slope $t$ between $A$ and $B$ satisfies
\[
    t = \frac{h(B) - h(A)}{|B \setminus A|} = \frac{\sum_{j=a+1}^b \tau_j |B_j|}{\sum_{j=a+1}^b |B_j|} \in [\tau_{a+1}, \tau_b].
\]
By~\eqref{eq:n2-canonical-prefix}, the set $C = P_h(t) = \{v \in E : b_h(v) \le t\}$ returned by \textsc{FindMinimizer}$(h, t)$ is the prefix $P_c$ where $c \in \{a+1, \dots, b\}$ is the largest index with $\tau_c \le t$. Consequently, $A \subsetneq C \subseteq B$ holds automatically without any interval constraints in \textsc{FindMinimizer}. If $b = a + 1$, then $t = \tau_{a+1}$ and $C = P_{a+1} = B$, so \textsc{Recover} correctly returns the adjacent pair $(P_a, P_{a+1})$. If $b > a + 1$, the strict inequality $\tau_{a+1} < \tau_b$ implies $t < \tau_b$, so $a < c < b$ and $A \subsetneq C \subsetneq B$; the two recursive calls on $(P_a, P_c)$ and $(P_c, P_b)$ therefore split the index interval $[a, b]$ nontrivially and recover $(P_a, \dots, P_b)$. Since the recursion tree has $L \le |E|$ leaves $(P_{j-1}, P_j)$, it makes at most $2L - 1 = O(|E|)$ calls to \textsc{FindMinimizer}.

Finally, for a terminal family $h_i(C) = f_i^{\to}(u, C)$ or $h_i(C) = f_i^{\to}(C, u)$ on $E = [n] \setminus \{u\}$, evaluating each $h_i(C)$ is a constrained submodular minimization of $f_i$ on at most $n$ elements, so one evaluation of $h(C) = \sum_{i=1}^m h_i(C)$ takes $O\bigl(m \mathcal{T}_{\min}(n)(\mathcal{T}_{\mathrm{eval}} + n)\bigr)$ time. Each call to \textsc{FindMinimizer} solves $|E| + 1 = n$ outer submodular minimizations of $h(C) - t|C|$ and thus takes $O\bigl(mn \mathcal{T}_{\min}(n)^2 (\mathcal{T}_{\mathrm{eval}} + n)\bigr)$ time. Multiplying by the $O(n)$ calls to \textsc{FindMinimizer} yields the claimed running time.
\end{proof}

\paragraph{Monotonicity of coordinate prices under marginal dominance.}
We conclude this subsection with a general comparison lemma for the coordinate prices $b_h$ produced by the decomposition: if a monotone submodular function $h$ has smaller marginal increments than $g$ on a subset $E_0 \subseteq E$, then $h$ assigns no larger coordinate prices than $g$ on $E_0$. This property will be the key to bounding the total influence budget in Section~\ref{sec:proofs1} (and again in Section~\ref{sec:sparsification-nk}).

\begin{lemma}[Monotonicity of decomposition prices under marginal dominance]
\label{lem:n2-price-comparison}
Let $E_0 \subseteq E$, and let $h : 2^E \to \mathbb{R}_{\ge 0}$ and $g : 2^{E_0} \to \mathbb{R}_{\ge 0}$ be monotone submodular functions with $h(\varnothing) = g(\varnothing) = 0$. Suppose that $h$ is marginally dominated by $g$ on $E_0$:
\begin{equation}\label{eq:n2-marginal-dominance-assumption}
    h(B) - h(A) \le g(B) - g(A) \qquad \text{for all } A \subseteq B \subseteq E_0.
\end{equation}
Then their coordinate price vectors satisfy
\begin{equation}\label{eq:n2-price-comparison}
    b_h(v) \le b_g(v) \qquad \text{for all } v \in E_0.
\end{equation}
\end{lemma}

\begin{proof}
Fix any parameter $t \ge 0$, and let $P := P_h(t) \subseteq E$ and $C := P_g(t) \subseteq E_0$ be the inclusion-wise largest minimizers of $h(X) - t|X|$ (over $X \subseteq E$) and $g(Y) - t|Y|$ (over $Y \subseteq E_0$), respectively. By submodularity of $h$, the marginal dominance assumption~\eqref{eq:n2-marginal-dominance-assumption} applied to $C \cap P \subseteq C \subseteq E_0$, and the optimality of $C$ for $g(\cdot) - t|\cdot|$ (which gives $g(C) - t|C| \le g(C \cap P) - t|C \cap P|$), we have
\[
    h(P \cup C) - h(P) \le h(C) - h(C \cap P) \le g(C) - g(C \cap P) \le t|C \setminus P| = t\bigl(|P \cup C| - |P|\bigr).
\]
Rearranging gives $h(P \cup C) - t|P \cup C| \le h(P) - t|P|$, so $P \cup C$ is also a minimizer of $h(X) - t|X|$. By the inclusion-wise maximality of $P = P_h(t)$, we must have $C \subseteq P$, i.e., $P_g(t) \subseteq P_h(t)$ for every $t \ge 0$. By the characterization~\eqref{eq:n2-canonical-prefix}, for any $v \in E_0$, setting $t = b_g(v)$ gives $v \in P_g(t) \subseteq P_h(t)$, and hence $b_h(v) \le t = b_g(v)$.
\end{proof}

\subsection{Coordinate influences: residual bound and total influence budget}
\label{sec:proofs1}

We now apply the decomposition of Section~\ref{sec:n2-decomposition} to the $2n$ single-vertex terminal families used in Algorithm~\ref{alg:n2-sparsification}. For each vertex $u \in [n]$, Lemma~\ref{lem:block-minimum-separate-submodularity} and the endpoint normalization $f_i(\varnothing) = f_i([n]) = 0$ ensure that the two families on $E = [n] \setminus \{u\}$,
\[
    \bigl(f_i^{\to}(u, \cdot)\bigr)_{i=1}^m \qquad \text{and} \qquad \bigl(f_i^{\to}(\cdot, u)\bigr)_{i=1}^m,
\]
consist of nonnegative monotone submodular functions vanishing at the empty set. Let $\bigl(s_i^{(u)}\bigr)_{i=1}^m$ and $b^{(u,\mathrm{in})} \in \mathbb{R}_{\ge 0}^{[n] \setminus \{u\}}$ denote the component shares and coordinate price vector produced by \textsc{Decomposition} on the inside-terminal family $\bigl(f_i^{\to}(u, \cdot)\bigr)_{i=1}^m$ (where $u$ is fixed inside the set), and let $\bigl(t_i^{(v)}\bigr)_{i=1}^m$ and $b^{(v,\mathrm{out})} \in \mathbb{R}_{\ge 0}^{[n] \setminus \{v\}}$ denote the component shares and coordinate price vector produced on the outside-terminal family $\bigl(f_i^{\to}(\cdot, v)\bigr)_{i=1}^m$ (where $v$ is fixed outside the set). By Item~4 of Lemma~\ref{lem:correctness_decomposition}, each of these $2n$ decompositions has total share at most $n - 1$, so the aggregate sampling weights $p_i = \sum_{u \in [n]} s_i^{(u)} + \sum_{v \in [n]} t_i^{(v)}$ satisfy
\begin{equation}\label{eq:n2-total-weight-bound}
    \sum_{i=1}^m p_i = \sum_{u \in [n]} \sum_{i=1}^m s_i^{(u)} + \sum_{v \in [n]} \sum_{i=1}^m t_i^{(v)} \le 2n(n - 1).
\end{equation}

\paragraph{Query-specific coordinate influences.}
Recall from Section~\ref{subsec:overview-block} and Lemma~\ref{lem:block-residual-two-block} that when we approximate a query $S \subseteq [n]$ by fixing a subset of high-influence coordinates $I \subseteq S$ and $O \subseteq \bar{S}$ and leaving the remaining coordinates $U := \overline{I \cup O}$ free, the approximation error (residual) of component $f_i$ is
\[
    \Delta_i(S, U) := f_i(S) - f_i^{\to}(S \setminus U, \bar{S} \setminus U) = f_i^{\to}(S, \bar{S}) - f_i^{\to}(I, O),
\]
which satisfies
\begin{equation}\label{eq:n2-residual-two-block-recall}
    0 \le \Delta_i(S, U) \le f_i^{\to}(S \cap U, \bar{S}) + f_i^{\to}(S, U \setminus S).
\end{equation}
Using the coordinate prices $b^{(u,\mathrm{in})}$ and $b^{(v,\mathrm{out})}$ from the $2n$ terminal decompositions, we define the \emph{influence} of each coordinate $w \in [n]$ with respect to a query $S \subseteq [n]$ by
\begin{equation}\label{eq:n2-influence-def}
    \operatorname{inf}_S(w) := \begin{cases}
        \displaystyle\max_{v \in \bar{S}} b^{(v,\mathrm{out})}(w), & w \in S, \\[8pt]
        \displaystyle\max_{u \in S} b^{(u,\mathrm{in})}(w), & w \in \bar{S},
    \end{cases}
\end{equation}
where the maximum over the empty set is taken to be zero. Note that $\operatorname{inf}_S(w)$ is used only in the analysis; the algorithm itself is oblivious to $S$.

We now prove that the weights $p_i$ and coordinate influences $\operatorname{inf}_S(w)$ satisfy the two core properties \textbf{(P1)} and \textbf{(P2)} from Section~\ref{subsec:overview-chaining}.

\begin{claim}[Residual bound \textbf{(P1)}]
\label{clm:bound_gap}
For any query $S \subseteq [n]$ and threshold $\theta \ge 0$, let $U_\theta := \{w \in [n] : \operatorname{inf}_S(w) \le \theta\}$. Then for every $u \in S$, $v \in \bar{S}$, and $i \in [m]$, we have
\[
    f_i^{\to}(u, U_\theta \setminus S) \le \theta s_i^{(u)} \qquad \text{and} \qquad f_i^{\to}(U_\theta \cap S, v) \le \theta t_i^{(v)},
\]
and consequently
\[
    \Delta_i(S, U_\theta) \le \theta \cdot p_i.
\]
\end{claim}

\begin{proof}
Write $A := U_\theta \cap S$ and $D := U_\theta \setminus S$. For any $u \in S$ and $w \in D \subseteq \bar{S}$, the definition~\eqref{eq:n2-influence-def} of $\operatorname{inf}_S(w)$ implies $b^{(u,\mathrm{in})}(w) \le \operatorname{inf}_S(w) \le \theta$. Applying the low-price bound~\eqref{eq:n2-low-price-shares} of Lemma~\ref{lem:correctness_decomposition} to the inside-terminal family $\bigl(f_i^{\to}(u, \cdot)\bigr)_{i=1}^m$ gives
\[
    f_i^{\to}(u, D) \le s_i^{(u)} \cdot \max_{w \in D} b^{(u,\mathrm{in})}(w) \le \theta s_i^{(u)}.
\]
Symmetrically, for any $v \in \bar{S}$ and $w \in A \subseteq S$, we have $b^{(v,\mathrm{out})}(w) \le \operatorname{inf}_S(w) \le \theta$, so applying~\eqref{eq:n2-low-price-shares} to $\bigl(f_i^{\to}(\cdot, v)\bigr)_{i=1}^m$ gives
\[
    f_i^{\to}(A, v) \le t_i^{(v)} \cdot \max_{w \in A} b^{(v,\mathrm{out})}(w) \le \theta t_i^{(v)}.
\]
Combining~\eqref{eq:n2-residual-two-block-recall} with the subadditivity of the block min-cut function in each argument (which follows from Lemma~\ref{lem:block-minimum-separate-submodularity}), we obtain
\[
    \Delta_i(S, U_\theta) \le f_i^{\to}(A, \bar{S}) + f_i^{\to}(S, D) \le \sum_{v \in \bar{S}} f_i^{\to}(A, v) + \sum_{u \in S} f_i^{\to}(u, D) \le \theta \left( \sum_{v \in \bar{S}} t_i^{(v)} + \sum_{u \in S} s_i^{(u)} \right) \le \theta p_i. \qedhere
\]
\end{proof}

\begin{claim}[Total influence budget \textbf{(P2)}]
\label{clm:sum_influences}
For every query $S \subseteq [n]$, we have $\sum_{w \in S} \operatorname{inf}_S(w) \le F(S)$ and $\sum_{w \in \bar{S}} \operatorname{inf}_S(w) \le F(S)$, and hence
\[
    \sum_{w \in [n]} \operatorname{inf}_S(w) \le 2 F(S).
\]
\end{claim}

\begin{proof}
If $S = \varnothing$ or $S = [n]$, then $\operatorname{inf}_S(w) = 0$ for all $w \in [n]$ by the empty-maximum convention, so the claim holds trivially. Assume $\varnothing \subsetneq S \subsetneq [n]$.

To bound the influences of coordinates inside $S$, consider the query-specific monotone submodular function $g : 2^S \to \mathbb{R}_{\ge 0}$ defined by
\[
    g(A) := \sum_{i=1}^m f_i^{\to}(A, \bar{S}) \qquad (A \subseteq S),
\]
which satisfies $g(\varnothing) = 0$ and $g(S) = \sum_{i=1}^m f_i^{\to}(S, \bar{S}) = F(S)$. Fix any outside vertex $v \in \bar{S}$, and let $h_v(A) := \sum_{i=1}^m f_i^{\to}(A, v)$ on $E = [n] \setminus \{v\}$ be the aggregate function of the outside-terminal family defining $b^{(v,\mathrm{out})}$. For any $A \subseteq B \subseteq S$ and each $i \in [m]$, let $T_1$ and $T_2$ be optimal sets attaining $f_i^{\to}(A, v)$ and $f_i^{\to}(B, \bar{S})$, respectively. Then $T_1 \cup T_2$ is feasible for $f_i^{\to}(B, v)$ (since $B \subseteq T_1 \cup T_2 \subseteq [n] \setminus \{v\}$) and $T_1 \cap T_2$ is feasible for $f_i^{\to}(A, \bar{S})$ (since $A \subseteq T_1 \cap T_2 \subseteq S$), so submodularity of $f_i$ yields
\[
    f_i^{\to}(B, v) + f_i^{\to}(A, \bar{S}) \le f_i(T_1 \cup T_2) + f_i(T_1 \cap T_2) \le f_i(T_1) + f_i(T_2) = f_i^{\to}(A, v) + f_i^{\to}(B, \bar{S}).
\]
Rearranging and summing over $i \in [m]$ shows that $h_v$ is marginally dominated by $g$ on $S \subseteq E$:
\[
    h_v(B) - h_v(A) \le g(B) - g(A) \qquad \text{for all } A \subseteq B \subseteq S.
\]
Applying Lemma~\ref{lem:n2-price-comparison} (with $E_0 = S$) gives $b^{(v,\mathrm{out})}(w) \le b_g(w)$ for every $w \in S$ and $v \in \bar{S}$. Consequently, by the base-polytope equality~\eqref{eq:n2-price-base-inequalities} for $g$,
\[
    \sum_{w \in S} \operatorname{inf}_S(w) = \sum_{w \in S} \max_{v \in \bar{S}} b^{(v,\mathrm{out})}(w) \le \sum_{w \in S} b_g(w) = g(S) = F(S).
\]

Symmetrically, on $\bar{S}$, define $g'(D) := \sum_{i=1}^m f_i^{\to}(S, D)$ for $D \subseteq \bar{S}$. For each $u \in S$, the same cross-monotonicity argument shows that $h_u'(D) := \sum_{i=1}^m f_i^{\to}(u, D)$ is marginally dominated by $g'$ on $\bar{S}$, so Lemma~\ref{lem:n2-price-comparison} gives $b^{(u,\mathrm{in})}(w) \le b_{g'}(w)$ for all $w \in \bar{S}$ and
\[
    \sum_{w \in \bar{S}} \operatorname{inf}_S(w) = \sum_{w \in \bar{S}} \max_{u \in S} b^{(u,\mathrm{in})}(w) \le \sum_{w \in \bar{S}} b_{g'}(w) = g'(\bar{S}) = F(S).
\]
Adding the two inequalities proves $\sum_{w \in [n]} \operatorname{inf}_S(w) \le 2F(S)$.
\end{proof}

\subsection{Positive chaining and proof of Theorem~\ref{thm:near-quadratic-submodular-sparsification}}
\label{sec:n2-chaining-proof}

With the weight budget~\eqref{eq:n2-total-weight-bound}, the running time bound of Claim~\ref{clm:n2-decomposition-recovery}, and the two influence guarantees \textbf{(P1)} (Claim~\ref{clm:bound_gap}) and \textbf{(P2)} (Claim~\ref{clm:sum_influences}) in place, we now complete the proof of Theorem~\ref{thm:near-quadratic-submodular-sparsification} via positive chaining.

\begin{proofof}{Theorem~\ref{thm:near-quadratic-submodular-sparsification}}
By the preprocessing reduction in Appendix~\ref{sec:preprocess} (Lemma~\ref{lem:endpoint-normalization}), we may assume $f_i(\varnothing) = f_i([n]) = 0$ for every $i \in [m]$ (at the cost of increasing $n$ by at most $2$, which preserves all asymptotic bounds). If $F \equiv 0$, every $f_i \equiv 0$ and the theorem holds trivially; assume henceforth that $F$ is not identically zero.

By~\eqref{eq:n2-total-weight-bound}, the deterministic sampling weights satisfy $\sum_{i=1}^m p_i \le 2n(n-1)$. Moreover, for any query $S \subseteq [n]$ with $F(S) > 0$, Claim~\ref{clm:sum_influences} implies $\operatorname{inf}_S(w) \le F(S)$ for all $w \in [n]$, so setting $\theta = F(S)$ in Claim~\ref{clm:bound_gap} gives $U_{F(S)} = [n]$ and
\[
    f_i(S) = \Delta_i(S, [n]) \le F(S) p_i \qquad (\forall i \in [m]).
\]
In particular, $p_i > 0$ for every nonzero component $f_i$, and summing over $i \in [m]$ gives $\sum_{i=1}^m p_i \ge 1$. (Components with $p_i = 0$ vanish identically and can be omitted from expressions involving division by $p_i$.)

\smallskip\noindent\textbf{Step 1: Multi-scale positive chaining decomposition.}
Fix any query $S \subseteq [n]$ with $F(S) > 0$. For each integer $j \ge 0$, define the dyadic threshold and low-influence free set
\[
    \theta_j := 2^{-j} F(S), \qquad U_j := \{w \in [n] : \operatorname{inf}_S(w) \le \theta_j\},
\]
so that $U_0 = [n]$. Let $J \ge 1$ be large enough that $\theta_J$ is strictly smaller than every positive value among $\{\operatorname{inf}_S(w) : w \in [n]\}$. Then $U_J = \{w \in [n] : \operatorname{inf}_S(w) = 0\}$, and Claim~\ref{clm:bound_gap} applied at $\theta = 0$ gives $\Delta_i(S, U_J) = 0$ for every $i \in [m]$.

For each $0 \le j \le J$, let
\[
    I_j := S \setminus U_j, \qquad O_j := \bar{S} \setminus U_j, \qquad \text{and} \qquad \ell_j := |I_j| + |O_j|
\]
denote the high-influence inside and outside coordinates revealed by scale $j$ (with $I_0 = O_0 = \varnothing$). For $1 \le j \le J$ and $i \in [m]$, define the scale-$j$ increment
\[
    z_{i,j} := f_i^{\to}(I_j, O_j) - f_i^{\to}(I_{j-1}, O_{j-1}).
\]
Because $I_{j-1} \subseteq I_j$, $O_{j-1} \subseteq O_j$, and the block min-cut function is monotone in both arguments (Lemma~\ref{lem:block-minimum-separate-submodularity}), every increment is nonnegative ($z_{i,j} \ge 0$) and upper bounded by $f_i(S) - f_i^{\to}(I_{j-1}, O_{j-1}) = \Delta_i(S, U_{j-1})$. Furthermore, since $f_i^{\to}(\varnothing, \varnothing) = 0$ and $f_i^{\to}(I_J, O_J) = f_i(S) - \Delta_i(S, U_J) = f_i(S)$, the increments telescope to $f_i(S)$ and, by Claim~\ref{clm:bound_gap}, satisfy
\begin{equation}\label{eq:n2-chain-increments}
    \sum_{j=1}^J z_{i,j} = f_i(S), \qquad 0 \le z_{i,j} \le \Delta_i(S, U_{j-1}) \le \theta_{j-1} p_i \quad (\forall i \in [m],\, 1 \le j \le J).
\end{equation}

\smallskip\noindent\textbf{Step 2: Bounding the total description budget.}
Each coordinate $w \in [n]$ with $\operatorname{inf}_S(w) > 0$ first enters $I_h \cup O_h$ at the unique scale $h \ge 1$ where $\theta_h < \operatorname{inf}_S(w) \le \theta_{h-1}$. Its total contribution to the weighted sum $\sum_{j=1}^J \ell_j \theta_{j-1}$ across all scales $j \ge h$ is therefore
\[
    \sum_{j=h}^J \theta_{j-1} < 2\theta_{h-1} = 4\theta_h < 4 \operatorname{inf}_S(w).
\]
Summing over all $w \in [n]$ and applying Claim~\ref{clm:sum_influences} yields
\begin{equation}\label{eq:n2-influence-chain-budget}
    \sum_{j=1}^J \ell_j \theta_{j-1} \le 4 \sum_{w \in [n]} \operatorname{inf}_S(w) \le 8 F(S).
\end{equation}
Combining~\eqref{eq:n2-chain-increments} and~\eqref{eq:n2-influence-chain-budget} bounds the total description-weighted relative scale of the chain:
\begin{equation}\label{eq:n2-increment-description-budget}
    \sum_{j=1}^J \ell_j \max_{i \in [m]} \frac{z_{i,j}}{p_i} \le 8 F(S).
\end{equation}

\smallskip\noindent\textbf{Step 3: Concentration and union bound over the increment dictionary.}
Observe that any increment vector $z = (z_1, \dots, z_m)$ of the form $z_i = f_i^{\to}(I, O) - f_i^{\to}(I', O')$ with $I' \subseteq I$, $O' \subseteq O$, $I \cap O = \varnothing$, and $|I| + |O| = \ell \in [n]$ is completely determined by choosing the $\ell$ fixed coordinates $I \cup O \subseteq [n]$ and assigning each of them one of four labels ($I'$, $I \setminus I'$, $O'$, or $O \setminus O'$). Thus, the query-independent dictionary $\mathcal{D}_\ell$ of all possible scale-$\ell$ increment vectors has size
\[
    |\mathcal{D}_\ell| \le \binom{n}{\ell} 4^\ell \le (4n)^\ell.
\]
For each $\ell \in [n]$ and each $z \in \mathcal{D}_\ell$, we apply Lemma~\ref{lem:Bernstein} with weights $p_i$, $\eta = \varepsilon/2$, $x = \ell \log(32n/\delta)$, and $\Lambda = 56\varepsilon^{-2}\log(32n/\delta)$. Taking a union bound over all $\ell \in [n]$ and $z \in \mathcal{D}_\ell$, the probability that Lemma~\ref{lem:Bernstein} fails for any increment vector in $\bigcup_{\ell=1}^n \mathcal{D}_\ell$ is at most
\[
    2 \sum_{\ell=1}^n |\mathcal{D}_\ell| e^{-\ell \log(32n/\delta)} \le 2 \sum_{\ell=1}^n (4n)^\ell \left(\frac{\delta}{32n}\right)^\ell \le 2 \sum_{\ell=1}^\infty \left(\frac{\delta}{8}\right)^\ell < \frac{\delta}{2}.
\]
Conditioned on the complementary event (which holds with probability $> 1 - \delta/2$), every query $S \subseteq [n]$ with $F(S) > 0$ simultaneously satisfies, by~\eqref{eq:n2-chain-increments} and~\eqref{eq:n2-increment-description-budget},
\begin{align*}
    \bigl|\widetilde{F}(S) - F(S)\bigr|
    &\le \sum_{j=1}^J \left|\sum_{i=1}^m (w_i - 1) z_{i,j}\right| \\
    &\le \frac{\varepsilon}{2} \sum_{j=1}^J \sum_{i=1}^m z_{i,j} + \frac{2\log(32n/\delta)}{\varepsilon \Lambda} \sum_{j=1}^J \ell_j \max_{i \in [m]} \frac{z_{i,j}}{p_i} \\
    &\le \left(\frac{\varepsilon}{2} + \frac{16\log(32n/\delta)}{\varepsilon \Lambda}\right) F(S) = \frac{11\varepsilon}{14} F(S) < \varepsilon F(S).
\end{align*}
(For queries with $F(S) = 0$, nonnegativity implies $f_i(S) = 0$ for all $i \in [m]$, so $\widetilde{F}(S) = F(S) = 0$ holds deterministically.)

\smallskip\noindent\textbf{Step 4: Sparsity and running time.}
Let $N := |\supp(w)| = \sum_{i=1}^m \xi_i$ denote the number of sampled components. By~\eqref{eq:n2-total-weight-bound}, $\mathbb{E}[N] \le \Lambda \sum_{i=1}^m p_i \le 2\Lambda n(n-1)$. Moreover, independence of the Bernoulli variables $\xi_i$ gives $\mathbb{E}[2^N] = \prod_{i=1}^m \bigl(1 + \min\{1, \Lambda p_i\}\bigr) \le \exp\bigl(\Lambda \sum_{i=1}^m p_i\bigr)$, so Markov's inequality and $\sum_{i=1}^m p_i \ge 1$ yield
\[
    \Pr\left[N > 2\Lambda \sum_{i=1}^m p_i\right] = \Pr\left[2^N > 4^{\Lambda \sum_i p_i}\right] \le \exp\left(-(2\log 2 - 1)\Lambda \sum_{i=1}^m p_i\right) < \frac{\delta}{2}.
\]
By a union bound with the concentration event of Step~3, with probability at least $1 - \delta$, $\widetilde{F}$ is an $\varepsilon$-sparsifier of $F$ with sparsity $N \le 4\Lambda n(n-1) = O(\varepsilon^{-2} n^2 \log(2n/\delta))$. Finally, multiplying the running time of Claim~\ref{clm:n2-decomposition-recovery} by the $2n$ terminal calls in Algorithm~\ref{alg:n2-sparsification} gives the total deterministic preprocessing time $O\bigl(mn^3 \mathcal{T}_{\min}(n)^2 (\mathcal{T}_{\mathrm{eval}} + n)\bigr)$.
\end{proofof}

\section{Sparsification for Cardinality-Bounded Queries}
\label{sec:residual-nk}
\label{sec:sparsification-nk}

\begin{theorem}\label{thm:small-query-submodular-sparsification}
    Let $F=\sum_{i=1}^m f_i$ with nonnegative submodular functions $f_1,\ldots,f_m:\{0,1\}^n\to\mathbb R_{\ge0}$ with $f_i(\varnothing) = f_i([n]) = 0$, and let $1\le k\le n$. Algorithm~\ref{alg:sampling_nk} computes weights $p_1,\ldots,p_m$ deterministically with $\sum_i p_i\le10nk$ in time $\!\left(mn^3k\,\mathcal{T}_{\min}(n)^2(\mathcal{T}_{\mathrm{eval}}+n)\right)$.

    Moreover, for any $\varepsilon,\delta\in(0,1)$, set $\Lambda=56\varepsilon^{-2}\log(32n/\delta)$ and let $\widetilde F=\sum_{i=1}^m w_i f_i$, where independently $w_i=1/\min\{1,\Lambda p_i\}$ with probability $\min\{1,\Lambda p_i\}$ (and $0$ otherwise). Then with probability at least $1-\delta$, $\tilde{F}$ is a size-$O\!\left(nk\varepsilon^{-2}\log(2n/\delta)\right)$ sparsifier of $F$ under the cardinality constraint $x \in \binom{[n]}{\le k}$.
\end{theorem}

\par\addvspace{\intextsep}
\hrule height.8pt depth0pt\nobreak\vskip2pt
\refstepcounter{algorithm}
\noindent\textbf{Algorithm~\thealgorithm} Sparsification under a cardinality constraint%
\label{alg:sampling_nk}
\addcontentsline{loa}{algorithm}{\protect\numberline{\thealgorithm}Sparsification under a cardinality constraint}
\par\nobreak\kern2pt\hrule\nobreak\kern2pt
\begin{algorithmic}[1]
\Procedure{ComputeSampling}{$f_1,\ldots,f_m,F:\{0,1\}^n\to\mathbb R_{\ge0}, \, \varepsilon,\, \delta$} \Comment{Compute $p_i$ and sample}
    \State $V\gets[n]$
    \ForAll{$u\in V$}
        \State $(s_i^{(u)})_{i=1}^m
          \gets\Call{Decomposition}{V\setminus\{u\},
          (f_i^\to(u,\cdot))_{i=1}^m,\,2k-1}$
    \EndFor
    \State $\mathcal H\gets\Call{LowCollisionPartitions}{V,k}$
    \ForAll{$\Pi\in\mathcal H$}
        \ForAll{labeled blocks $X\in\Pi$}
            \State $(t_i^{(X)})_{i=1}^m
              \gets\Call{Decomposition}{V\setminus X,
              (f_i^\to(\cdot,X))_{i=1}^m,\,|V\setminus X|}$
        \EndFor
    \EndFor
    \State $p_i\gets
       \displaystyle\sum_{u\in V}s_i^{(u)}
       +\frac{2}{|\mathcal H|}
          \sum_{\Pi\in\mathcal H}\sum_{X\in\Pi}t_i^{(X)}$
          for $i\in[m]$
    \State $\Lambda\gets56\varepsilon^{-2}\log(32n/\delta)$
    \ForAll{$i\in[m]$}
        \State Independently draw
          $\xi_i\sim\operatorname{Bernoulli}(\min\{1,\Lambda p_i\})$
        \State $w_i \gets 1 / \min\{1, \Lambda p_i\}$ if $\xi_i = 1$, and $0$ otherwise
    \EndFor
    \State \Return $(w_1,\ldots,w_m)$
\EndProcedure
\end{algorithmic}
\par\nobreak\kern2pt\hrule
\par\addvspace{\intextsep}

\paragraph{Proof outline.}
    As in the unrestricted case, we use terminal decompositions to bound the residual $\Delta_i(S,U)$ and then apply positive chaining. The difference is how we obtain shares with total weight $O(nk)$. Using all $2n$ single-vertex terminal families, as in Section~\ref{sec:sparsification-n2}, would give an $O(n^2)$ bound. For the families $f_i^\to(u,\cdot)$, we cap the total share of each family at $K=2k-1$. In the other direction, we replace $f_i^\to(\cdot,v)$ by the block families $f_i^\to(\cdot,X)$, where $X$ ranges over the blocks of partitions in a fixed low-collision family $\mathcal H$. Each partition has fewer than $4k$ blocks, and two distinct vertices share a block with probability at most $1/(2k)$. Averaging the corresponding shares gives
\[
    \sum_i p_i\le10nk.
\]

    The choice $K=2k-1$ also leaves the minimizers needed for small queries unchanged. Fix $S$ with $|S|\le k$ and $F(S)>0$, and let $u\in S$. For the uncapped objective associated with $h_u(C)=\sum_i f_i^\to(u,C)$, choosing $C=\bar S$ leaves at most $k-1$ coordinates outside $C$, and $h_u(\bar S)\le F(S)$. Thus, whenever $t\ge F(S)/k$,
\[
    h_u(\bar S)+t(|S|-1)
    \le F(S)+t(k-1)
    \le(2k-1)t.
\]
    The two terms account for the function value and the at most $k-1$ other elements of $S$. The uncapped minimum is therefore at most $Kt$, so Lemma~\ref{lem:nk-capped-decomposition} shows that the capped and uncapped largest minimizers agree. Hence the smaller shares still control the sets needed at these thresholds.

    As in the unrestricted proof, the two uncapped decompositions associated with $S$ give coordinate prices $b_S(v)$ with $\sum_{v\in V}b_S(v)=2F(S)$. With $U_t=\{v:b_S(v)\le t\}$, the capped decomposition bound and the low-collision property imply
\[
    \Delta_i(S,U_t)\le tp_i \qquad\text{for }t\ge F(S)/k.
\]
    Unlike in the unrestricted proof, we therefore stop the geometric chain at this scale rather than continue to zero. The ordinary increments have description cost at most $8F(S)$. The remaining term satisfies $z_{i,*}\le(2F(S)/k)p_i$ and involves at most $2k$ fixed coordinates. Consequently,
\[
    \sum_{j=1}^J\ell_j\max_i\frac{z_{i,j}}{p_i}
    +3k\max_i\frac{z_{i,*}}{p_i}
    \le14F(S).
\]
    As before, an ordinary increment with $\ell$ fixed coordinates belongs to a query-independent collection of size at most $(4n)^\ell$. The final residual requires a separate collection, of size at most $(4n)^{3k}$. Bernstein's inequality and a union bound now give uniform concentration with $\Lambda=56\varepsilon^{-2}\log(32n/\delta)$; the bound on $\sum_i p_i$ gives the sparsity.


    We specify the procedures used to compute the sampling scores. Let $h_1,\ldots,h_m:2^E\to\mathbb R_{\ge0}$ be normalized, monotone submodular functions, write $h=\sum_i h_i$, and let $K\ge1$. Put
\begin{equation}\label{eq:nk-capped-objective}
    \kappa(C)=\min\{K,|E\setminus C|\},
    \qquad
    P(t)=\operatorname*{argmin}_{C\subseteq E}
                \bigl\{h(C)+t\kappa(C)\bigr\}
    \quad(t\ge0).
\end{equation}
    Ties are broken by taking the inclusion-wise largest minimizer.

\begin{claim}\label{clm:nk-decomposition-implementation}
    Procedure~\textsc{Decomposition} in Algorithm~\ref{alg:nk-decomposition} recovers the chain of distinct minimizers $P(t)$ when $h$ is not identically zero, and returns the shares
\[
    s_i=\sum_{j=1}^L\frac{h_i(P_j)-h_i(P_{j-1})}{\tau_j},
    \qquad
    \tau_j=\frac{h(P_j)-h(P_{j-1})}{\kappa(P_{j-1})-\kappa(P_j)},
\]
    where $P_0\subsetneq\cdots\subsetneq P_L=E$ are the distinct sets attained by $P(t)$, $0<\tau_1<\cdots<\tau_L$, and $\kappa(P_{j-1})>\kappa(P_j)$ for every $j$. These shares satisfy $\sum_i s_i=\kappa(P_0)\le K$ and $h_i(C)\le ts_i$ whenever $C\subseteq P(t)$ and $t>0$. If $h\equiv0$, the procedure returns zero shares.
\end{claim}

\par\addvspace{\intextsep}
\hrule height.8pt depth0pt\nobreak\vskip2pt
\refstepcounter{algorithm}
\noindent\textbf{Algorithm~\thealgorithm} Computing the shares of a terminal family%
\label{alg:nk-decomposition}
\addcontentsline{loa}{algorithm}{\protect\numberline{\thealgorithm}Computing the shares of a terminal family}
\par\nobreak\kern2pt\hrule\nobreak\kern2pt
\begin{algorithmic}[1]
\Procedure{Decomposition}{$E,(h_i)_{i=1}^m,K$} \Comment{Similar to \textsc{Decomposition} in Algorithm~\ref{alg:n2-terminal-shells} except $\kappa$: compute the shares $s_i$}
    \If{$E=\varnothing$}
        \State \Return $(0,\ldots,0)$
    \EndIf
    \State $h\gets\sum_{i=1}^m h_i$; $\kappa(C)\gets\min\{K,|E\setminus C|\}$
    \State $P_0\gets\Call{FindMinimizer}{h,\kappa,0}$
    \If{$P_0=E$}
        \State \Return $(0,\ldots,0)$
    \EndIf
    \State $(P_0,\ldots,P_L)\gets\Call{Recover}{h,\kappa,P_0,E}$
    \For{$j=1,\ldots,L$}
        \State $\tau_j\gets\dfrac{h(P_j)-h(P_{j-1})}{\kappa(P_{j-1})-\kappa(P_j)}$
    \EndFor
    \State $s_i\gets\displaystyle\sum_{j=1}^L\dfrac{h_i(P_j)-h_i(P_{j-1})}{\tau_j}$ for $i\in[m]$
    \State \Return $(s_1,\ldots,s_m)$
\EndProcedure
\Procedure{Recover}{$h,\kappa,A,B$} \Comment{Recover the chain of minimizers from $A$ to $B$ with $\kappa$}
    \If{$A=B$}
        \State \Return $(A)$
    \EndIf
    \State $t\gets\dfrac{h(B)-h(A)}{\kappa(A)-\kappa(B)}$
    \State $C\gets\Call{FindMinimizer}{h,\kappa,t}$
    \If{$C=B$}
        \State \Return $(A,B)$
    \EndIf
    \State \Return $\Call{Recover}{h,\kappa,A,C}\cup\Call{Recover}{h,\kappa,C,B}$ \Comment{keep at most one $C$}
\EndProcedure
\Procedure{FindMinimizer}{$h,\kappa,t$} \Comment{Return the largest minimizer of $h+t\kappa$}
    \State $\mu\gets\min_{C\subseteq E}\{h(C)+t\kappa(C)\}$ \Comment{Submodular minimization}
    \State $Q\gets\varnothing$
    \ForAll{$v\in E$}
        \State $\mu'\gets\min_{C:\,Q\cup\{v\}\subseteq C\subseteq E}\{h(C)+t\kappa(C)\}$
        \If{$\mu'=\mu$}
            \State $Q\gets Q\cup\{v\}$
        \EndIf
    \EndFor
    \State \Return $Q$
\EndProcedure
\end{algorithmic}
\par\nobreak\kern2pt\hrule
\par\addvspace{\intextsep}

\begin{claim}\label{clm:nk-low-collision-partitions}
    For $V=[n]$ and $1\le k\le n$, Procedure~\textsc{LowCollisionPartitions} in Algorithm~\ref{alg:nk-low-collision-partitions} constructs an indexed family $\mathcal H$ of fewer than $4n$ partitions $\Pi=(X_1,\ldots,X_b)$ of $V$, where $2k\le b<4k$ and the blocks are labeled and may be empty. For all distinct $u,v\in V$,
\begin{equation}\label{eq:nk-low-collision-probability}
    \frac{1}{|\mathcal H|}\sum_{\Pi\in\mathcal H}
       \sum_{c=1}^b \mathbf{1}\bigl[\{u,v\}\subseteq X_c\bigr]
    \le\frac{1}{2k}.
\end{equation}
    The partitions are indexed, so repeated partitions are counted with multiplicity.
\end{claim}

Together, the two procedures compute the scores in~\eqref{eq:nk-score} in time $O\!\left(mn^3k\,\mathcal{T}_{\min}(n)^2(\mathcal{T}_{\mathrm{eval}}+n)\right)$, as shown in Section~\ref{subsec:nk-polynomial-time}.

\par\addvspace{\intextsep}
\hrule height.8pt depth0pt\nobreak\vskip2pt
\refstepcounter{algorithm}
\noindent\textbf{Algorithm~\thealgorithm} Constructing a low-collision partition family%
\label{alg:nk-low-collision-partitions}
\addcontentsline{loa}{algorithm}{\protect\numberline{\thealgorithm}Constructing a low-collision partition family}
\par\nobreak\kern2pt\hrule\nobreak\kern2pt
\begin{algorithmic}[1]
\Procedure{LowCollisionPartitions}{$V,k$} \Comment{Construct $\mathcal H$ from Claim~\ref{clm:nk-low-collision-partitions}}
    \State $b\gets2^{\lceil\log_2(2k)\rceil}$
    \State $M\gets2^{\lceil\log_2(\max\{n,b\})\rceil}$
    \State Construct $\mathbb F_M$ and assign distinct labels $x_v\in\mathbb F_M$ to $v\in V$
    \State Fix a surjective $\mathbb F_2$-linear map $\varphi:\mathbb F_M\to\mathbb F_2^{\log_2 b}$
    \State $\mathcal H\gets[\ ]$
    \ForAll{$a\in\mathbb F_M\setminus\{0\}$}
        \ForAll{$c\in\mathbb F_2^{\log_2 b}$}
            \State $X_{a,c}\gets\{v\in V:\varphi(ax_v)=c\}$
        \EndFor
        \State Append $(X_{a,c})_{c\in\mathbb F_2^{\log_2 b}}$ to $\mathcal H$
    \EndFor
    \State \Return $\mathcal H$
\EndProcedure
\end{algorithmic}
\par\nobreak\kern2pt\hrule
\par\addvspace{\intextsep}

\subsection{Sampling scores and certificates}
\label{subsec:nk-scores}

    We first record the properties of the capped decomposition in~\eqref{eq:nk-capped-objective}.

\begin{lemma}\label{lem:nk-capped-decomposition}
    If $h$ is not identically zero, the distinct sets attained by $P(t)$ form a chain $P_0\subsetneq P_1\subsetneq\cdots\subsetneq P_L=E$. Consecutive sets meet at strictly increasing positive breakpoints $\tau_1,\ldots,\tau_L$, where
\begin{equation}\label{eq:nk-capped-breakpoint}
    \tau_j
    =\frac{h(P_j)-h(P_{j-1})}
           {\kappa(P_{j-1})-\kappa(P_j)}
    \quad\text{and}\quad
    \kappa(P_{j-1})>\kappa(P_j).
\end{equation}
    The shares
\begin{equation}\label{eq:nk-capped-share}
    s_i=\sum_{j=1}^L
    \frac{h_i(P_j)-h_i(P_{j-1})}{\tau_j}
\end{equation}
    satisfy
\begin{equation}\label{eq:nk-capped-output}
    \sum_i s_i=\kappa(P_0)\le K,
    \qquad
    h_i(C)\le ts_i\quad\text{if }C\subseteq P(t),\ t>0.
\end{equation}
    Furthermore,
\begin{equation}\label{eq:nk-capped-optimum}
    \min_C\{h(C)+t\kappa(C)\}
    =\min\!\left\{
         \min_C\{h(C)+t|E\setminus C|\},Kt
      \right\}.
\end{equation}
    If the uncapped minimum is at most $Kt$, its inclusion-wise largest minimizer equals $P(t)$. Finally, if $D\subseteq E$ satisfies
\begin{equation}\label{eq:nk-density-criterion}
    h(D)-h(C)\le t|D\setminus C|
    \qquad(C\subseteq D),
\end{equation}
    then $D$ is contained in the inclusion-wise largest minimizer of $h(C)+t|E\setminus C|$. For $h\equiv0$, all shares are zero and $P(t)=E$.
\end{lemma}

    For $K\ge|E|$, the share formula and the bound $h_i(C)\le ts_i$ recover \eqref{eq:allocation_s_i} and~\eqref{eq:bound_f_i_x^*} from the overview. The cap preserves this bound while limiting the total share to $K$.

\begin{lemma}\label{lem:residual-nk-random-restriction}
    Let $h:2^E\to\mathbb R_{\ge0}$ be normalized, monotone, and submodular. If a random subset $R$ of $E$ satisfies $\Pr[v\in R]\ge\eta$ for every $v\in E$, then $\mathbb E h(R)\ge\eta h(E)$.
\end{lemma}

\begin{lemma}\label{lem:nk-partition-recovery}
    Let $\mathcal H$ be the indexed family of partitions from Claim~\ref{clm:nk-low-collision-partitions}, constructed by \textsc{LowCollisionPartitions} in Algorithm~\ref{alg:nk-low-collision-partitions}. For each $\Pi\in\mathcal H$, the notation $X\in\Pi$ ranges over its $b$ labeled blocks, including empty blocks. For every $i\in[m]$ and $A\subseteq S\subseteq V$ with $|S|\le k$, writing $\bar S=V\setminus S$, we have
\begin{equation}\label{eq:nk-partition-recovery}
    f_i^\to(A,\bar S)
    \le\frac{2}{|\mathcal H|}
       \sum_{\Pi\in\mathcal H}
       \sum_{\substack{X\in\Pi\\X\cap S=\varnothing}}
               f_i^\to(A,X).
\end{equation}
    The sums count partitions and blocks with their multiplicities; in particular, $\mathcal H$ is fixed independently of $S$.
\end{lemma}

    We now define the scores used in Algorithm~\ref{alg:sampling_nk}. For each $u\in V$, apply Lemma~\ref{lem:nk-capped-decomposition} to the functions $D\mapsto f_i^\to(u,D)$ on $V\setminus\{u\}$, with $K=2k-1$, and denote the resulting share of component $i$ by $s_i^{(u)}$.

    Similarly, fix a labeled block $X$ in a partition $\Pi\in\mathcal H$. Take $E=V\setminus X$ and apply the same lemma to the functions $C\mapsto f_i^\to(C,X)$ on $E$, with $K=|E|$. We write $t_i^{(X)}$ for the share of component $i$ given by~\eqref{eq:nk-capped-share}. Here the cap is inactive, so these are the ordinary shares of the terminal family with exclusion block $X$. If $X=V$, the ground set is empty and we set $t_i^{(V)}=0$. The sampling score is
\begin{equation}\label{eq:nk-score}
    p_i:=\sum_{u\in V}s_i^{(u)}
      +\frac{2}{|\mathcal H|}
        \sum_{\Pi\in\mathcal H}\sum_{X\in\Pi}t_i^{(X)}.
\end{equation}
    Each occurrence of a labeled block contributes to this sum, including repeated and empty blocks.

\begin{lemma}\label{lem:nk-score-budget}
    The scores in~\eqref{eq:nk-score} satisfy
\begin{equation}\label{eq:nk-score-budget}
    \sum_i p_i\le10nk.
\end{equation}
\end{lemma}

    For a query $S\subseteq V$ write $\bar S=V\setminus S$. As in~\eqref{eq:bound_Delta_S1_S2}, define the residual
\begin{equation}\label{eq:residual-nk-definition}
    \Delta_i(S,U)
    :=f_i(S)-f_i^\to(S\setminus U,\bar S\setminus U).
\end{equation}

\begin{proposition}\label{prop:nk-terminal-dominates-residual}
    For every $S$ with $|S|\le k$ and $F(S)>0$, there are nonnegative coordinate prices $b_S(v)$, $v\in V$, satisfying
\begin{equation}\label{eq:nk-query-price-budget}
    \sum_{v\in V}b_S(v)=2F(S).
\end{equation}
    For $t\ge F(S)/k$, let $U_t=\{v\in V:b_S(v)\le t\}$. Then for every $i\in[m]$,
\begin{equation}\label{eq:nk-terminal-dominates-residual}
    \Delta_i(S,U_t)\le tp_i.
\end{equation}
\end{proposition}

\begin{lemma}\label{lem:nk-stopped-chaining}
    After discarding identically zero components, for every $S$ with $|S|\le k$ and $F(S)>0$ there is a sequence $V=U_0\supsetneq U_1\supsetneq\cdots\supsetneq U_J$ (possibly $J=0$) with the following properties. Put
\[
    I_j=S\setminus U_j,\qquad O_j=\bar S\setminus U_j,
    \qquad \ell_j=|I_j|+|O_j|.
\]
    For $1\le j\le J$, let
\begin{equation}\label{eq:residual-nk-ordinary-increment}
    z_{i,j}:=f_i^\to(I_j,O_j)
                  -f_i^\to(I_{j-1},O_{j-1}),
\end{equation}
    and put
\begin{equation}\label{eq:residual-nk-final-residual}
    z_{i,*}:=\Delta_i(S,U_J).
\end{equation}
    Then all increments are nonnegative, $\ell_J\le2k$, and
\begin{equation}\label{eq:residual-nk-exact-decomposition}
    f_i(S)=\sum_{j=1}^Jz_{i,j}+z_{i,*}.
\end{equation}
    Moreover,
\begin{equation}\label{eq:residual-nk-chain-budget}
    \sum_{j=1}^J\ell_j
       \max_i\frac{z_{i,j}}{p_i}
       +3k\max_i\frac{z_{i,*}}{p_i}
    \le14F(S).
\end{equation}
    For each $\ell$, the vectors $(z_{i,j})_{i=1}^m$ with $\ell_j=\ell$ belong to a query-independent set of at most $(4n)^\ell$ vectors, and all possible final residual vectors $(z_{i,*})_{i=1}^m$ belong to a query-independent set of at most $(4n)^{3k}$ vectors.
\end{lemma}

\subsection{Proof of Theorem~\ref{thm:small-query-submodular-sparsification}}
\label{subsec:nk-sampling}

\begin{proofof}{Theorem~\ref{thm:small-query-submodular-sparsification}}
    By preprocessing, we may assume $f_i(\varnothing)=f_i(V)=0$ and discard identically zero components. If none remain, the zero reweighting is exact. We work with the remaining components and the scores in~\eqref{eq:nk-score}.

    Lemma~\ref{lem:nk-score-budget} gives $\sum_i p_i\le10nk$. For any query $S$ with $F(S)>0$, Proposition~\ref{prop:nk-terminal-dominates-residual} gives $b_S(v)\le2F(S)$ for every $v\in V$. Taking $t=2F(S)$ therefore gives $U_t=V$ and
\begin{equation}\label{eq:nk-positive-scores}
    f_i(S)=\Delta_i(S,V)\le2F(S)p_i.
\end{equation}
    Every nonzero component is positive on some singleton: otherwise submodularity and $f_i(\varnothing)=0$ would imply $f_i\equiv0$. Thus all remaining $p_i$ are positive, and, whenever $F$ is nonzero, summing~\eqref{eq:nk-positive-scores} yields $\sum_i p_i\ge1/2$.

    Fix a query $S$ with $|S|\le k$ and $F(S)>0$. Use the decomposition in Lemma~\ref{lem:nk-stopped-chaining}. For an ordinary increment with $\ell$ fixed coordinates, there are at most $(4n)^\ell$ possible vectors. For the final residual, there are at most $(4n)^{3k}$ possible vectors. Both collections are determined by the input, independently of the query and the chosen decomposition.

    We apply Lemma~\ref{lem:Bernstein} with the scores $p_i$, $\eta=\varepsilon/2$, and $x=\ell\log(32n/\delta)$ for each ordinary increment, or $x=3k\log(32n/\delta)$ for the final residual. Taking
\begin{equation}\label{eq:nk-sampling-probabilities}
    \Lambda=56\varepsilon^{-2}\log(32n/\delta)
\end{equation}
    as in Algorithm~\ref{alg:sampling_nk}, a union bound over both collections has failure probability at most
\[
    2\sum_{\ell=1}^{2k}(4n)^\ell
       \left(\frac{\delta}{32n}\right)^\ell
    +2(4n)^{3k}\left(\frac{\delta}{32n}\right)^{3k}
    <\frac{\delta}{2}.
\]
    On the complementary event, the concentration bound holds simultaneously for all queries and all their increments. Writing $\widetilde F(S)=\sum_i w_i f_i(S)$ and using \eqref{eq:residual-nk-exact-decomposition} and \eqref{eq:residual-nk-chain-budget}, we obtain
\begin{align*}
    |\widetilde F(S)-F(S)|
    &\le\sum_{j=1}^J
       \left|\sum_i(w_i-1)z_{i,j}\right|
       +\left|\sum_i(w_i-1)z_{i,*}\right|\\
    &\le\frac{\varepsilon}{2}
       \left(\sum_{j=1}^J\sum_i z_{i,j}+\sum_i z_{i,*}\right)\\
    &\quad
       +\frac{2\log(32n/\delta)}{\varepsilon\Lambda}
        \left(
          \sum_{j=1}^J\ell_j
             \max_i\frac{z_{i,j}}{p_i}
          +3k\max_i\frac{z_{i,*}}{p_i}
        \right)\\
    &\le\left(
         \frac{\varepsilon}{2}
         +\frac{28\log(32n/\delta)}{\varepsilon\Lambda}
        \right)F(S)\\
    &=\varepsilon F(S).
\end{align*}
    If $F(S)=0$, nonnegativity implies $f_i(S)=0$ for every $i$, so the approximation is exact.

    Let $N=|\operatorname{supp}(w)|$. The sampling probabilities and score budget give $\mathbb E N\le\Lambda\sum_i p_i\le10\Lambda nk$. Independence also gives
\[
    \mathbb E 2^N
    =\prod_i\bigl(1+\min\{1,\Lambda p_i\}\bigr)
    \le\exp\!\left(\Lambda\sum_i p_i\right).
\]
    Markov's inequality and $\sum_i p_i\ge1/2$ imply
\[
    \Pr\left[N>2\Lambda\sum_i p_i\right]
    \le\exp\!\left(
        -(2\log2-1)\Lambda\sum_i p_i
    \right)
    <\frac{\delta}{2}.
\]
    Together with uniform concentration, this proves the high-probability support and approximation bounds. Since $\mathbb E w_i=1$ for each remaining component, the estimator is unbiased. The running time follows from the implementations of Claims~\ref{clm:nk-decomposition-implementation} and~\ref{clm:nk-low-collision-partitions} in Section~\ref{subsec:nk-polynomial-time}.
\end{proofof}

\subsection{Proofs of the auxiliary lemmas}

\begin{proofof}{Lemma~\ref{lem:nk-capped-decomposition}}
    The function $\kappa(C)=\min\{K,|E\setminus C|\}$ is submodular: the marginal of adding a coordinate is $0$ while more than $K$ coordinates remain outside, and $-1$ otherwise. Hence the minimizers of $H_t(C)=h(C)+t\kappa(C)$ are closed under union, and the largest minimizer $P(t)$ is well defined.

    The sets $P(t)$ are nested. Indeed, for $t_1<t_2$, write $A=P(t_1)$ and $B=P(t_2)$. Submodularity and the optimality of $A$ give
\[
    0
    \le H_{t_1}(A)+H_{t_1}(B)
       -H_{t_1}(A\cap B)-H_{t_1}(A\cup B)
    \le H_{t_1}(B)-H_{t_1}(A\cup B).
\]
    The last expression equals
\[
    H_{t_2}(B)-H_{t_2}(A\cup B)
    +(t_2-t_1)\bigl(\kappa(A\cup B)-\kappa(B)\bigr)
    \le0.
\]
    Thus $A\cup B$ is optimal at $t_2$, so $A\subseteq B$.

    At $t=0$, $P_0$ is the largest zero set of $h$; for $t>h(E)$, the unique minimizer is $E$. The finitely many objective lines therefore determine a chain of distinct minimizers and breakpoints. If two consecutive minimizers had the same $\kappa$-value, the tie at their transition would force equal $h$-values. Their objective lines would then coincide, contrary to the largest-minimizer rule. Thus $\kappa(P_{j-1})>\kappa(P_j)$, and equality of the two objective values at the transition gives \eqref{eq:nk-capped-breakpoint}. The breakpoints are positive and strictly increasing: at a tie, the largest-minimizer rule selects the larger set, including all simultaneous jumps.

    Summing the share formula over components gives
\begin{align*}
    \sum_i s_i
    &=\sum_{j=1}^L
       \frac{h(P_j)-h(P_{j-1})}{\tau_j}\\
    &=\sum_{j=1}^L
       \bigl(\kappa(P_{j-1})-\kappa(P_j)\bigr)\\
    &=\kappa(P_0)\le K.
\end{align*}
    Moreover, $h_i(P_0)=0$ by nonnegativity. If $C\subseteq P(t)$, monotonicity gives
\begin{align*}
    h_i(C)
    &\le h_i(P(t))\\
    &=\sum_{\tau_j\le t}
       \bigl(h_i(P_j)-h_i(P_{j-1})\bigr)\\
    &\le t s_i.
\end{align*}

    For~\eqref{eq:nk-capped-optimum}, the objectives coincide whenever $|E\setminus C|\le K$. Outside this range the capped value is at least $Kt$, while the empty set has capped value $Kt$ if $K<|E|$. When $K\ge|E|$, the cap never applies. This proves the identity.

    For $t=0$, the capped and uncapped objectives coincide. Suppose henceforth that $t>0$, the uncapped minimum is at most $Kt$, and let $A$ be its largest minimizer. Since $h(A)\ge0$ and $t>0$, $|E\setminus A|\le K$, so $A$ is also capped-optimal. Let $B=P(t)$. The largest-minimizer convention gives $A\subseteq B$, and hence $|E\setminus B|\le K$. Thus $B$ also minimizes the uncapped objective, proving $A=B$.

    Finally, let $Q$ be the largest minimizer of the uncapped objective and assume~\eqref{eq:nk-density-criterion}. By submodularity,
\[
    h(Q\cup D)-h(Q)
    \le h(D)-h(D\cap Q)
    \le t|D\setminus Q|.
\]
    Consequently, $Q\cup D$ is also optimal for the uncapped objective. By maximality, $D\subseteq Q$. If $h\equiv0$, nonnegativity makes every $h_i$ zero, and the convention in the statement applies.
\end{proofof}

\begin{proofof}{Lemma~\ref{lem:residual-nk-random-restriction}}
    Order $E=\{v_1,\ldots,v_d\}$, put $E_j=\{v_1,\ldots,v_j\}$, and let $b_{v_j}=h(E_j)-h(E_{j-1})$. These values are nonnegative and sum to $h(E)$. Submodularity gives $h(A)\ge\sum_{v\in A}b_v$ for every $A\subseteq E$. Taking expectations,
\[
    \mathbb E h(R)
    \ge\sum_{v\in E}b_v\Pr[v\in R]
    \ge\eta\sum_{v\in E}b_v
    =\eta h(E).
\]
\end{proofof}

\begin{proofof}{Lemma~\ref{lem:nk-partition-recovery}}
    Choose $\Pi$ uniformly from $\mathcal H$ and let $Z$ be the union of its blocks disjoint from $S$. For $v\in\bar S$, the collision guarantee and a union bound over $S$ give
\[
    \Pr[v\in Z]
    \ge1-\frac{|S|}{2k}
    \ge\frac12.
\]
    The function $D\mapsto f_i^\to(A,D)$ on $\bar S$ is normalized, monotone, and submodular. Lemma~\ref{lem:residual-nk-random-restriction} yields $\mathbb E f_i^\to(A,Z)\ge f_i^\to(A,\bar S)/2$. On the other hand, subadditivity in the exclusion block gives
\[
    f_i^\to(A,Z)
    \le\sum_{\substack{X\in\Pi\\X\cap S=\varnothing}}
          f_i^\to(A,X).
\]
    Taking expectations proves~\eqref{eq:nk-partition-recovery}.
\end{proofof}

\begin{proofof}{Lemma~\ref{lem:nk-score-budget}}
    Lemma~\ref{lem:nk-capped-decomposition} gives $\sum_i s_i^{(u)}\le2k-1$ for each $u$. For a labeled block $X$, the uncapped shares satisfy $\sum_i t_i^{(X)}\le n-|X|$. Hence, for every $\Pi$,
\[
    \sum_i\sum_{X\in\Pi}t_i^{(X)}
    \le\sum_{X\in\Pi}(n-|X|)
    =(b-1)n.
\]
    Using~\eqref{eq:nk-score} and $b<4k$,
\[
    \sum_i p_i
    \le n(2k-1)+2(b-1)n
    \le10nk.
\]
\end{proofof}

\begin{proofof}{Proposition~\ref{prop:nk-terminal-dominates-residual}}
    We use the endpoint normalization $f_i(\varnothing)=f_i(V)=0$. Fix $S$. Since $F(S)>0$, both $S$ and $\bar S$ are nonempty. Consider the two normalized, monotone submodular functions
\[
    h^{\mathrm{in}}(A)=\sum_i f_i^\to(A,\bar S)
    \quad(A\subseteq S),\qquad
    h^{\mathrm{out}}(D)=\sum_i f_i^\to(S,D)
    \quad(D\subseteq\bar S).
\]
    Both functions take value $F(S)$ on their full ground sets. Apply Lemma~\ref{lem:nk-capped-decomposition} without a cap, taking $K=|S|$ and $K=|\bar S|$, respectively. In each decomposition, assign price $\tau_j$ to each coordinate in $P_j\setminus P_{j-1}$, and zero to those in $P_0$. Denote these prices by $b_S(v)$. The breakpoint formula~\eqref{eq:nk-capped-breakpoint} gives
\[
    \sum_{v\in S}b_S(v)=h^{\mathrm{in}}(S)=F(S),
    \qquad
    \sum_{v\in\bar S}b_S(v)=h^{\mathrm{out}}(\bar S)=F(S).
\]
    In particular, the prices are nonnegative and satisfy~\eqref{eq:nk-query-price-budget}. Fix $t\ge F(S)/k$ and put $U_t=\{v:b_S(v)\le t\}$, $A=S\cap U_t$, and $D=\bar S\cap U_t$. By the chain in Lemma~\ref{lem:nk-capped-decomposition}, $A$ and $D$ are the largest minimizers of $h^{\mathrm{in}}(C)+t|S\setminus C|$ and $h^{\mathrm{out}}(C)+t|\bar S\setminus C|$, respectively.

    We use the increasing-differences property of block minima: for $A'\subseteq A$ and $D\subseteq D_0$ with $A\cap D_0=\varnothing$,
\begin{equation}\label{eq:nk-block-cross-inequality}
    f_i^\to(A,D)-f_i^\to(A',D)
    \le f_i^\to(A,D_0)-f_i^\to(A',D_0)
    \quad\text{if }D\subseteq D_0.
\end{equation}
    The analogous inequality holds on exchanging the two blocks. Indeed, take minimizers for $f_i^\to(A,D_0)$ and $f_i^\to(A',D)$. Their union is feasible for $f_i^\to(A,D)$ and their intersection for $f_i^\to(A',D_0)$; submodularity gives~\eqref{eq:nk-block-cross-inequality}.

    Comparing $A$ and $D$ with their subsets in the two minimization problems yields
\begin{align}
    \sum_i\bigl(f_i^\to(A,\bar S)
              -f_i^\to(A',\bar S)\bigr)
    &\le t|A\setminus A'|
       &&(A'\subseteq A),\label{eq:nk-inside-density}\\
    \sum_i\bigl(f_i^\to(S,D)
              -f_i^\to(S,D')\bigr)
    &\le t|D\setminus D'|
       &&(D'\subseteq D).\label{eq:nk-outside-density}
\end{align}

    We first bound $f_i^\to(A,\bar S)$. Fix a labeled block $X$ disjoint from $S$ and put $h_X(C)=\sum_i f_i^\to(C,X)$. For every $B\subseteq A$, the cross inequality and \eqref{eq:nk-inside-density} give
\[
    h_X(A)-h_X(B)
    \le\sum_i\bigl(f_i^\to(A,\bar S)
                    -f_i^\to(B,\bar S)\bigr)
    \le t|A\setminus B|.
\]
    The last assertion of Lemma~\ref{lem:nk-capped-decomposition}, applied with $K=|V\setminus X|$, places $A$ in the largest minimizer of the uncapped objective at parameter $t$. Its share bound gives $f_i^\to(A,X)\le t\,t_i^{(X)}$. Using Lemma~\ref{lem:nk-partition-recovery}, we obtain
\begin{align}
    f_i^\to(A,\bar S)
    &\le\frac{2}{|\mathcal H|}
       \sum_{\Pi\in\mathcal H}
       \sum_{\substack{X\in\Pi\\X\cap S=\varnothing}}
                 f_i^\to(A,X)\notag\\
    &\le t\,\frac{2}{|\mathcal H|}
       \sum_{\Pi\in\mathcal H}\sum_{X\in\Pi}t_i^{(X)}.
       \label{eq:nk-free-inside-bound}
\end{align}

    For the other term, fix $u\in S$ and put $h_u(C)=\sum_i f_i^\to(u,C)$ on $V\setminus\{u\}$. The cross inequality and~\eqref{eq:nk-outside-density} give, for $C\subseteq D$,
\[
    h_u(D)-h_u(C)
    \le\sum_i\bigl(f_i^\to(S,D)-f_i^\to(S,C)\bigr)
    \le t|D\setminus C|.
\]
    The same assertion places $D$ in the largest minimizer of the uncapped objective for $h_u$ at $t$. The set $\bar S$ is a feasible comparison set for this objective. Since $S$ is feasible for every block minimum defining $h_u(\bar S)$,
\[
    h_u(\bar S)+t(|S|-1)
    \le F(S)+t(k-1)
    \le(2k-1)t.
\]
    Lemma~\ref{lem:nk-capped-decomposition} therefore shows that the largest minimizers of the capped and uncapped objectives agree when $K=2k-1$. The share bound now gives $f_i^\to(u,D)\le t s_i^{(u)}$. Subadditivity in the inclusion block yields
\begin{equation}\label{eq:nk-free-outside-bound}
    f_i^\to(S,D)
    \le\sum_{u\in S}f_i^\to(u,D)
    \le t\sum_{u\in V}s_i^{(u)}.
\end{equation}

    Finally, Lemma~\ref{lem:block-residual-two-block} gives
\[
    \Delta_i(S,U)
    \le f_i^\to(A,\bar S)+f_i^\to(S,D).
\]
    Combining~\eqref{eq:nk-free-inside-bound} and \eqref{eq:nk-free-outside-bound} gives $\Delta_i(S,U_t)\le tp_i$, as required.
\end{proofof}

\begin{proofof}{Lemma~\ref{lem:nk-stopped-chaining}}
    Fix $S$ and take the prices $b_S(v)$ from Proposition~\ref{prop:nk-terminal-dominates-residual}. For $j\ge0$, let
\[
    \theta_j=2^{1-j}F(S),\qquad
    U_j=\{v\in V:b_S(v)\le\theta_j\}.
\]
    Since $\sum_v b_S(v)=2F(S)$, we have $U_0=V$. Let $J$ be the smallest index such that $\theta_J\le2F(S)/k$. Then
\begin{equation}\label{eq:residual-nk-rate-halving}
    \theta_{j+1}=\frac{\theta_j}{2},
    \qquad \frac{F(S)}{k}\le\theta_J\le\frac{2F(S)}{k}.
\end{equation}
    Every vertex outside $U_J$ has price greater than $\theta_J$. Thus
\begin{equation}\label{eq:residual-nk-potential-accounting}
    \theta_J\ell_J
    \le\sum_{v\notin U_J}b_S(v)
    \le2F(S),
\end{equation}
    and hence $\ell_J\le2k$.

    Suppose a coordinate $v$ is first excluded from $U_j$ at index $r$. Then $\theta_r<b_S(v)\le\theta_{r-1}$, and its contribution to $\sum_{j=1}^J\ell_j\theta_{j-1}$ is
\[
    \sum_{j=r}^J\theta_{j-1}
    <2\theta_{r-1}
    =4\theta_r
    <4b_S(v).
\]
    Summing over the coordinates and using~\eqref{eq:nk-query-price-budget} gives
\begin{equation}\label{eq:residual-nk-description-cost}
    \sum_{j=1}^J\ell_j\theta_{j-1}\le8F(S).
\end{equation}

    By monotonicity of block minima, all $z_{i,j}\ge0$. Since $f_i^\to(\varnothing,\varnothing)=0$, they telescope with the final residual to give~\eqref{eq:residual-nk-exact-decomposition}. Proposition~\ref{prop:nk-terminal-dominates-residual} implies, for each $1\le j\le J$,
\[
    z_{i,j}
    \le\Delta_i(S,U_{j-1})
    \le\theta_{j-1}p_i.
\]
    For the final residual, the same proposition gives
\[
    z_{i,*}=\Delta_i(S,U_J)
    \le\theta_Jp_i\le\frac{2F(S)}{k}p_i.
\]
    Together with~\eqref{eq:residual-nk-description-cost}, these inequalities give
\[
    \sum_{j=1}^J\ell_j\max_i\frac{z_{i,j}}{p_i}
    +3k\max_i\frac{z_{i,*}}{p_i}
    \le8F(S)+6F(S)
    =14F(S).
\]
    Removing repeated sets $U_j$ only deletes zero increments and does not change the budget or the final residual. The resulting chain is strictly decreasing, as required.

    The increment dictionaries are fixed independently of $S$ and the thresholds $\theta_j$. An ordinary increment with $\ell$ fixed coordinates is determined by its old and new inclusion and exclusion blocks. Each fixed coordinate has four possible labels, so the number of possible vectors is at most
\begin{equation}\label{eq:residual-nk-ordinary-dictionary}
    \binom n\ell4^\ell\le(4n)^\ell.
\end{equation}
    A final residual is determined by a query $S$ of size at most $k$ and disjoint fixed blocks $I\subseteq S$ and $O\subseteq\bar S$ with $|I|+|O|\le2k$. Listing these vertices in a fixed order, with the two block labels and padding, gives at most
\begin{equation}\label{eq:residual-nk-final-residual-dictionary}
    (2n+1)^{2k}(n+1)^k\le(4n)^{3k}
\end{equation}
    possible vectors. This proves the lemma.
\end{proofof}

\subsection{Polynomial-time construction}
\label{subsec:nk-polynomial-time}

\begin{proofof}{Claim~\ref{clm:nk-decomposition-implementation}}
    Lemma~\ref{lem:nk-capped-decomposition} gives the chain of distinct minimizers $P(t)$ and the shares in the claim. It remains to verify that Algorithm~\ref{alg:nk-decomposition} recovers them and to bound its running time.

    For an inside-terminal family indexed by $u$, each evaluation of $h_i(C)=f_i^\to(u,C)$ is obtained from
\[
    \min_{Z\subseteq E\setminus C}f_i(\{u\}\cup Z),
    \qquad E=V\setminus\{u\}.
\]
    For an outside block $X\ne V$, the corresponding evaluation of $h_i(C)=f_i^\to(C,X)$ is
\[
    \min_{Z\subseteq E\setminus C}f_i(E\setminus Z),
    \qquad E=V\setminus X.
\]
    Both are submodular minimizations on at most $n$ coordinates. Since each $h_i$ and $\kappa$ are submodular, so is the outer objective $h+t\kappa$. Its minimizers are closed under union. Procedure~\textsc{FindMinimizer} starts with $Q=\varnothing$ and adds a coordinate whenever some minimizer contains $Q$ and that coordinate. It therefore returns the inclusion-wise largest minimizer $P(t)$, with $t=0$ giving $P_0$.

    The remaining step is the recursion \textsc{Recover}. The endpoints are $P(0)=P_0$ and $P(h(E)+1)=E$. For two distinct sets $A\subsetneq B$ attained by $P(t)$, Lemma~\ref{lem:nk-capped-decomposition} gives $\kappa(A)>\kappa(B)$, so the intersection parameter
\begin{equation}\label{eq:nk-parametric-intersection}
    t^\star
    =\frac{h(B)-h(A)}{\kappa(A)-\kappa(B)}
\end{equation}
    is well defined. If $A=P(a)$ and $B=P(b)$, then $a<t^\star\le b$ by optimality at the endpoints and the largest-minimizer convention. Consequently,
\[
    A\subsetneq P(t^\star)\subseteq B.
\]
    If $P(t^\star)=B$, no set $P(s)$ lies strictly between $A$ and $B$. Otherwise, writing $D=P(s)$ for such a set, the strict decrease of $\kappa$ along the chain implies $\kappa(A)>\kappa(D)>\kappa(B)$. If $s<t^\star$, optimality of $D$ at $s$ gives
\[
    H_{t^\star}(D)-H_{t^\star}(A)
    \le(s-t^\star)\bigl(\kappa(A)-\kappa(D)\bigr)<0,
\]
    contrary to the optimality of $B$ at $t^\star$, where the lines of $A$ and $B$ meet. If $s>t^\star$, compare instead with $B$:
\[
    H_{t^\star}(D)-H_{t^\star}(B)
    \le(s-t^\star)\bigl(\kappa(B)-\kappa(D)\bigr)<0.
\]
    The case $s=t^\star$ contradicts $P(t^\star)=B$. Thus the recursion correctly records consecutive sets. Otherwise it splits the interval at the intermediate minimizer. The two returned chains are joined in increasing order, listing their common endpoint once. Since $\kappa$ decreases by at least one at each transition, there are at most $K+1$ distinct sets $P(t)$ and $O(K)$ recursive calls.

    There are $n$ inside-terminal families, each requiring $O(k)$ distinct minimizers, and $O(nk)$ outside-block families, each requiring $O(n)$ distinct minimizers. Thus the algorithm computes $O(n^2k)$ such sets in total.

    Let $\mathcal{T}_{\min}(n)$ denote the cost of submodular minimization with a unit-cost value oracle on at most $n$ coordinates. Each evaluation of $h(C)$ uses $m$ inner submodular minimizations, taking $O(m\mathcal{T}_{\min}(n)(\mathcal{T}_{\mathrm{eval}}+n))$ time. One outer minimization therefore takes $O(m\mathcal{T}_{\min}(n)^2(\mathcal{T}_{\mathrm{eval}}+n))$ time. Procedure~\textsc{FindMinimizer} in Algorithm~\ref{alg:nk-decomposition} performs at most $n+1$ outer minimizations per returned set, including the tie-breaking tests. Evaluating the component shares at the resulting minimizers takes no more time. The total preprocessing time is
\begin{equation}\label{eq:nk-preprocessing-time}
    O\!\left(
       mn^3k\,\mathcal{T}_{\min}(n)^2(\mathcal{T}_{\mathrm{eval}}+n)
    \right).
\end{equation}
    The partitions, component shares, scores, and final Bernoulli draws require only lower-order work. The endpoint preprocessing changes $n$ and $k$ by at most a constant, so it does not affect the stated asymptotic bounds.
\end{proofof}

\begin{proofof}{Claim~\ref{clm:nk-low-collision-partitions}}
    Let $b$ be the least power of two at least $2k$, and let $M$ be the least power of two at least $\max\{n,b\}$. Then $2k\le b<4k$ and $M<4n$. Choose distinct $x_v\in\mathbb F_M$ for $v\in V$ and a surjective $\mathbb F_2$-linear map $\varphi:\mathbb F_M\to\mathbb F_2^{\log_2 b}$. For each $a\in\mathbb F_M\setminus\{0\}$, define the labeled partition $\Pi_a=(X_{a,c})_{c\in\mathbb F_2^{\log_2 b}}$ by
\[
    X_{a,c}=\{v\in V:\varphi(ax_v)=c\}.
\]
    Let $\mathcal H=(\Pi_a)_{a\in\mathbb F_M\setminus\{0\}}$. For distinct $u,v\in V$, the elements $a(x_u-x_v)$ range over $\mathbb F_M\setminus\{0\}$ as $a$ varies. Since $|\ker\varphi|=M/b$, the left-hand side of~\eqref{eq:nk-low-collision-probability} equals
\[
    \frac{|\ker\varphi|-1}{M-1}
    =\frac{M/b-1}{M-1}
    \le\frac1b
    \le\frac1{2k}.
\]
    There are $M-1<4n$ indexed partitions. Since $M=O(n)$, the field can be constructed by testing monic binary polynomials of degree $\log_2 M$ for irreducibility using~\cite{Rabin80}. The partitions are then enumerated in polynomial time.
\end{proofof}

\section{Lower Bounds under Cardinality Constraints}
\label{sec:lower_bounds}

We now prove the lower bounds stated in Section~\ref{sec:introduction}.
Proposition~\ref{clm:lower_bound_arbitrary} shows that
Theorem~\ref{inform:sparsification_cardinality_constraint} is tight up to a
logarithmic factor. Propositions~\ref{clm:lower_bound_regular}
and~\ref{clm:lower_bound_monotone} show that for symmetric or monotone
functions, which admit $\widetilde O(n)$-size sparsifiers without any
constraint~\cite{JLLS23,KPS24}, a cardinality constraint cannot reduce the
size to $o(n)$, even under suitable flatness assumptions on the singleton
values.
The bounds are summarized in Table~\ref{tab:cardinality_special_cases}. Finally, we provide an upper bound for regular and monotone submodular functions.

\begin{table}[H]
    \centering
    \begin{tabular}{lcc}
        \hline
        Function family & Sparsity bound & Reference \\
        \hline
        General submodular ($k\ge1$) & $\Omega(nk)$ & Proposition~\ref{clm:lower_bound_arbitrary} \\
        Regular (Def.~\ref{def:regular_submodular}) and symmetric ($k\ge2$) & $\Omega(n)$ & Proposition~\ref{clm:lower_bound_regular} \\
        Almost-regular (Def.~\ref{def:regular_submodular}) and monotone ($k\ge1$) & $\Omega(n)$ & Proposition~\ref{clm:lower_bound_monotone} \\
        \hline
         Regular (Def.~\ref{def:regular_submodular}) and monotone & $O(k^2 \log n)$ & Lemma~\ref{lem:algorithm_monotone_regular} \\
        \hline
    \end{tabular}
    \caption{Sparsity bounds under cardinality constraints.}
    \label{tab:cardinality_special_cases}
\end{table}

The first proof restricts the directed-cut construction
of~\cite{CohenKPPRSV17} to queries of size at most $k$.

\begin{proposition}\label{clm:lower_bound_arbitrary}
    For every $n\ge2$ and $1\le k\le n$, there are nonnegative
    submodular functions on $[n]$ for which every sparsifier under the
    cardinality constraint $|S|\le k$ has size $\Omega(nk)$.
\end{proposition}

\begin{proof}
    Take disjoint sets $L,R\subseteq[n]$ with
    $|L|=\lfloor n/2\rfloor$ and
    $|R|=\min\{k,\lceil n/2\rceil\}$.
    For each $(u,v)\in L\times R$, let
    \[
        f_{u,v}(S)=\cut_{u\to v}(S)=\mathbf{1}[u\in S,\ v\notin S],
    \]
    and let $F=\sum_{(u,v)\in L\times R}f_{u,v}$, so that a sparsifier
    has the form $\widetilde F=\sum_{(u,v)\in L\times R}w_{u,v}f_{u,v}$
    with $w_{u,v}\ge0$.
    The query $S_{u,v}=\{u\}\cup(R\setminus\{v\})$ has size
    $|R|\le k$. Since $u$ is the only vertex of $L$ in $S_{u,v}$ and
    $v$ is the only vertex of $R$ outside $S_{u,v}$,
    \[
        f_{u,v}(S_{u,v})=1,\qquad
        f_{u',v'}(S_{u,v})=0
        \quad\text{for }(u',v')\ne(u,v).
    \]
    Thus $F(S_{u,v})=1$ and
    $\widetilde F(S_{u,v})=w_{u,v}$, so
    $w_{u,v}\ge1-\varepsilon>0$.
    Every component must be retained, and their number is
    $|L||R|=\Omega(nk)$.
\end{proof}

    


We now turn to two subclasses of submodular functions: symmetric ones
and monotone ones. Both admit
sparsifiers of size $\widetilde O(n)$ without any constraint, and we ask
whether a cardinality constraint reduces this to $o(n)$.

Without further assumptions it does not. For the monotone functions
$f_v(S)=\mathbf{1}[v\in S]$, $v\in[n]$, a sparsifier has the form
$\widetilde F=\sum_{v\in[n]}w_vf_v$, and every singleton query satisfies
\[
    F(\{v\})=f_v(\{v\})=1,\qquad
    \widetilde F(\{v\})=w_vf_v(\{v\})=w_v,
\]
so $w_v\ge1-\varepsilon>0$ for every $v$, and all $n$ components must be
retained already for $k=1$. The same holds for symmetric functions, e.g.,
for the cut functions of $n/2$ disjoint undirected edges.

This phenomenon also appears in matrix sparsification, where one keeps a
reweighted subset of the rows of a matrix $A$ so that $\|Ax\|_2$ is
approximately preserved. Without further assumptions, restricting to
$k$-sparse vectors $x$ does not reduce the number of rows needed.
However, for a normalized Hadamard matrix, whose rows $A_i$ are flat,
i.e., $|A_{ij}|=O(\|A_i\|_2/\sqrt n)$ for all $j$, only $k(\log n)^{O(1)}$
rows are needed~\cite{RV,CheraghchiGV13,HR16,CD20}. We therefore ask
whether analogous flatness conditions on the singleton values $f(\{v\})$
help in our setting.

\begin{definition}\label{def:regular_submodular}
    A submodular function $f:2^{[n]}\to\R_{\ge0}$ is \emph{regular} if
    $\max_{v\in[n]}f(\{v\})\le C\min_{v\in[n]}f(\{v\})$, and
    \emph{almost-regular} if
    $\max_{v\in[n]}f(\{v\})\le\frac{C}{n}\sum_{v\in[n]}f(\{v\})$,
    where $C\ge1$ is an absolute constant.
\end{definition}
Regularity says that $f$ is flat, in the sense that no element has a
singleton value much larger than the others. Almost-regularity is a weaker
form of flatness that only compares the largest singleton value with the
average. Intuitively, flatness prevents a component from concentrating on
a few elements, so one might expect a
cardinality constraint to help. However, the next two lower bounds show that this
is not the case.

\begin{proposition}\label{clm:lower_bound_regular}
    For every \emph{even} number $n\ge4$, there are $m=\Omega(n)$ regular symmetric
    submodular functions $f_1,\ldots,f_m:2^{[n]}\to\R_{\ge0}$ such that
    every sparsifier under the cardinality constraint $|S|\le2$ retains
    all $m$ functions.
\end{proposition}

\begin{proof}
    Let $m=\lfloor(n-2)/2\rfloor$ and label distinct elements of $[n]$ as
    $a_1,\ldots,a_m,b_1,\ldots,b_m,r_1,r_2$. For each $i\in[m]$, let $G_i$ be the graph with $n/2$ unit-weight
    edges (as a matching)
    \[
        \{a_j,b_j\}\ (j\ne i),\qquad \{a_i,r_1\},\qquad \{b_i,r_2\},
    \]
    and let
    $f_i=\sum_{e\in E(G_i)}\cut_e$ be its cut function. Each $f_i$ is
    symmetric and submodular, and $f_i(\{v\})$ is the degree of $v$ in
    $G_i$, which equals $1$.    Hence each $f_i$ is regular.
    For the query $S_i=\{a_i,b_i\}$,
    \[
        f_j(S_i)=
        \begin{cases}
            2,&j=i,\\
            0,&j\ne i,
        \end{cases}
    \]
    since in $G_i$ both edges at $a_i$ and $b_i$ leave $S_i$, whereas in
    $G_j$ the only edge at $a_i$ or $b_i$ is $\{a_i,b_i\}$.
    Hence $F(S_i)=2$ and $\widetilde F(S_i)=2w_i$, forcing
    $w_i\ge1-\varepsilon>0$ for every $i$.
\end{proof}

\begin{proposition}\label{clm:lower_bound_monotone}
    For every $n\ge2$, there are $m=\Omega(n)$ almost-regular monotone
    submodular functions $f_1,\ldots,f_m:2^{[n]}\to\R_{\ge0}$ such that
    every sparsifier under the cardinality constraint $|S|\le1$ retains
    all $m$ functions.
\end{proposition}

\begin{proof}
    Let $m=\lfloor n/2\rfloor$ and $R=[n]\setminus[m]$. For each $i\in[m]$,
    let
    \[
        f_i(S)=\mathbf{1}[S\cap(\{i\}\cup R)\ne\varnothing],
    \]
    the indicator that $S$ meets a fixed set, which is monotone and
    submodular. Since $f_i(\{v\})=\mathbf{1}[v\in\{i\}\cup R]$, we have
    $\max_vf_i(\{v\})=1$ and
    $\frac1n\sum_vf_i(\{v\})=\frac{|R|+1}{n}\ge\frac12$, so each $f_i$ is
    almost-regular. For the singleton query $S_i=\{i\}$, we have
    $f_j(S_i)=\mathbf{1}[j=i]$ since $i\notin R$. Hence $F(S_i)=1$ and
    $\widetilde F(S_i)=w_i$, forcing $w_i\ge1-\varepsilon>0$ for every $i$.
\end{proof}

We complement the two lower bounds with an upper bound for monotone submodular functions: regularity reduces the size of sparsifiers to $\tilde{O}(k^2)$ nonzero entries.

\begin{lemma}\label{lem:algorithm_monotone_regular}
    Let $f_1,\ldots,f_m:2^{[n]}\to\R_{\ge0}$ be regular monotone submodular
    functions with $f_i(\varnothing) = 0$, and let $1\le k\le n$ and $\varepsilon\in(0,1)$. There is a
    randomized algorithm that returns $w\in\R_{\ge0}^m$ with
    $O(k^2\varepsilon^{-2}\log n)$ nonzero entries such that, with
    probability at least $0.99$,
    \[
        \widetilde F(S)=(1\pm\varepsilon)F(S)
        \qquad\text{for all } S\in\tbinom{[n]}{\le k}.
    \]
\end{lemma}

\begin{proof}
Let $c_i=\min_{v\in[n]}f_i(\{v\})$. By regularity, $\max_{v\in[n]}f_i(\{v\})=O(c_i)$.
Let $\rho_i=Ck \cdot \frac{c_i}{\sum_j c_j}$, where $C$ is a constant with $\max_v f_i(\{v\})\le Cc_i$ for every $i$, which exists by regularity.
Since each $f_i$ is monotone and submodular, we have
\[
\max_{S \subseteq [n], |S|\leq k} \frac{f_i(S)}{F(S)}\leq \max_a \frac{k\cdot f_i(\{a\})}{F(\{a\})} \leq C\frac{kc_i}{\sum_j c_j} \le \rho_i.
\]
Apply Lemma~\ref{lem:Bernstein} with $p_i=\rho_i$, $z_i=f_i(S)$, and
$\eta=\varepsilon/2$, and take $x=\Theta(k\log n)$ and
$\Lambda=\Theta(k\varepsilon^{-2}\log n)$. A union bound over the at most
$(n+1)^k$ queries of size at most $k$ gives the guarantee for all of them.
The expected support size is $\Lambda\sum_i\rho_i=O(k^2\varepsilon^{-2}\log n)$,
and by Markov's inequality the support has this size, up to a constant
factor, with probability $0.99$ overall.
\end{proof}

\begingroup
\let\subsection\section
\endgroup

\appendix

\section{Importance Weights from Gibbs Distribution}\label{sec:Gibbs}

    In this appendix, we sparsify $F=\sum_{i=1}^m f_i$ by sampling with the weights from Gibbs distributions sketched in the overview. We first recall the necessary definitions, and then present the algorithm and its guarantee.

    As in the overview, we write $V=[n]$ and assume $f_i(\varnothing)=f_i(V)=0$ for every $i$, which is achieved by the preprocessing in Appendix~\ref{sec:preprocess}. For a submodular function $g$ with $g(\varnothing)=0$, $\widehat g$ denotes its Lov\'asz extension~\eqref{eq:Lovasz_ext}. For disjoint $A,B\subseteq V$, recall the directed block min-cut function $f_i^\to(A,B)=\min\{f_i(T):A\subseteq T,\ T\cap B=\varnothing\}$ from Definition~\ref{def:directed_block_mincut_function}. For each $u\in V$, the two set functions $D\mapsto f_i^\to(u,D)$ and $D\mapsto f_i^\to(D,u)$ on $V\setminus\{u\}$ are monotone and submodular by Lemma~\ref{lem:block-minimum-separate-submodularity}, and they vanish at $D=\varnothing$ because $f_i(\varnothing)=f_i(V)=0$. We write $\widehat{f_i^\to}(u,x)$ and $\widehat{f_i^\to}(x,u)$ for their Lov\'asz extensions at $x\in\mathbb R_{\ge0}^{V\setminus\{u\}}$.

Throughout this appendix, we fix an inverse temperature $\beta\in[1,\infty)$. For a normalized monotone submodular function $h$ on a ground set $E$, the Gibbs distribution $\mu_h$ at inverse temperature $\beta$ has density proportional to $e^{-\beta\widehat h(x)}\prod_{v\in E}x(v)^{\beta-1}$ on $(0,\infty)^E$, where $E$ is restricted to the coordinates on which $h$ depends; see~\eqref{eq:gibbs-distribution} for the precise definition. In the logarithmic coordinates $y(v)=\log x(v)$, this density is proportional to $\exp\bigl(\beta\bigl(\sum_v y(v)-\widehat h(e^y)\bigr)\bigr)$, which is the cooled Gibbs distribution in the overview. The constants in the approximation and sparsity bounds below are uniform in $\beta\ge1$, and we suppress $\beta$ from the notation.

    Our algorithm samples with the weights $p_i$ of~\eqref{eqdef:weights_sum_2n_Gibbs} at inverse temperature $\beta$. It is given in Algorithm~\ref{alg:sampling_n2_code}, and its guarantee is the following.

\begin{theorem}\label{thm:gibbs-near-quadratic}
    Let $\beta\ge1$ be an integer and $\varepsilon,\delta\in(0,1)$. The scores in Algorithm~\ref{alg:sampling_n2_code} satisfy $\sum_ip_i\le2n(n-1)$, and for fixed $\beta$, one can compute in polynomial time positive scores $\tilde p_1,\ldots,\tilde p_m$ with $\sum_i\tilde p_i\le4n(n-1)+1$ such that, with probability at least $1-\delta/4$, $\tilde p_i\ge p_i$ for every $i$. Moreover, run Algorithm~\ref{alg:sampling_n2_code} with the scores $p_i$, or with $\tilde p_i$ in place of $p_i$. Then with probability at least $1-\delta$, $\tilde F:=\sum_i w_if_i$ is an $\varepsilon$-approximation of $F$ with sparsity $O(n^2\varepsilon^{-2}\log(2n/\delta))$, where the constant does not depend on $\beta$.
\end{theorem}
    The idea of the proof extends to every fixed rational $\beta\ge1$. Moreover, the guarantees for the exact scores $p_i$ extend to every constant $\beta\in(0,1)$, with constants depending on $\beta$. For simplicity, we only consider integer $\beta\ge1$.

\begin{algorithm}[ht]
\caption{Sparsification with Gibbs scores at inverse temperature $\beta$}
\label{alg:sampling_n2_code}
\begin{algorithmic}[1]
\Procedure{ComputeSampling}{$f_1,\ldots,f_m:\{0,1\}^n\to\mathbb R_{\ge0}$; $\varepsilon,\delta\in(0,1)$; $\beta\in[1,\infty)$}
    \State Discard identically zero components; if none remain, return zero weights
    \State $V\gets[n]$ \Comment{After preprocessing, $f_i(\varnothing)=f_i(V)=0$}
    \ForAll{$u\in V$}
        \State $E\gets V\setminus\{u\}$
        \State $H_u(x)\gets\sum_{i=1}^m\widehat{f_i^\to}(u,x)$ for $x\in\mathbb R_{\ge0}^{E}$
        \State $G_u(x)\gets\sum_{i=1}^m\widehat{f_i^\to}(x,u)$ for $x\in\mathbb R_{\ge0}^{E}$
        \State $s_i^{(u)}\gets\mathbb E_{X_u\sim\mu_{H_u}}\widehat{f_i^\to}(u,X_u)$ for $i\in[m]$
        \State $t_i^{(u)}\gets\mathbb E_{Y_u\sim\mu_{G_u}}\widehat{f_i^\to}(Y_u,u)$ for $i\in[m]$
    \EndFor
    \State $p_i\gets\sum_{u\in V}(s_i^{(u)}+t_i^{(u)})$ for $i\in[m]$
    \State $\Lambda\gets128\varepsilon^{-2}\log(64n/\delta)$
    \ForAll{$i\in[m]$}
        \State Independently draw $\xi_i\sim\operatorname{Bernoulli}(\min\{1,\Lambda p_i\})$
        \State $w_i\gets\xi_i/\min\{1,\Lambda p_i\}$
    \EndFor
    \State \Return $(w_1,\ldots,w_m)$
\EndProcedure
\end{algorithmic}
\end{algorithm}
We prove Theorem~\ref{thm:gibbs-near-quadratic} by the chaining argument sketched in the overview. For a query $S$, write $\bar S=V\setminus S$. Since $f_i(S)=f_i^\to(S,\bar S)$, we approximate $f_i(S)$ by $f_i^\to(I,O)$, where $I\subseteq S$ and $O\subseteq\bar S$ consist of the coordinates with large influence on $F(S)$ (made precise below). Leaving the coordinates in $U=V\setminus(I\cup O)$ free loses the residual
\[
    \Delta_i(S,U):=f_i^\to(S,\bar S)-f_i^\to(S\setminus U,\bar S\setminus U)\ \ge0,
\]
    as in~\eqref{eq:bound_Delta_S1_S2}.

    We first introduce two quantities of the Gibbs distributions of $H_u$ and $G_u$: the share of $f_i$, which measures how much $f_i$ contributes to them, and the influence of a coordinate, which measures how much the coordinate affects them.

\paragraph{Shares.}
    In Algorithm~\ref{alg:sampling_n2_code}, $s_i^{(u)}$ and $t_i^{(u)}$ are the \emph{shares} of $f_i$ in the Gibbs distributions of $H_u$ and $G_u$. Since $H_u=\sum_i\widehat{f_i^\to}(u,\cdot)$, linearity gives $\sum_is_i^{(u)}=\mathbb E\,H_u(X_u)$ for $X_u\sim\mu_{H_u}$, and similarly for $G_u$. By the positive homogeneity of the Lov\'asz extension, this total equals the number of coordinates on which $H_u$ depends, which is at most $n-1$. Summing over the $2n$ distributions gives the following bound on the total score; we give the details in Section~\ref{sec:gibbs-shares}.

\begin{lemma}\label{lem:gibbs-total-share}
    For every $u\in V$, $\sum_is_i^{(u)}\le n-1$ and $\sum_it_i^{(u)}\le n-1$. Consequently, $\sum_ip_i\le2n(n-1)$.
\end{lemma}

    The shares are Gibbs expectations, which may be hard to compute exactly. For integer $\beta$, however, they can be overestimated by sampling, at the cost of a factor two in the total; we prove the following in Section~\ref{sec:gibbs-polynomial-time}.

\begin{lemma}\label{lem:gibbs-share-estimation}
    Let $\beta\ge1$ be an integer, $u\in V$, and $\delta'\in(0,1)$. For fixed $\beta$, one can compute, in time polynomial in $n$, $m$, and $\log(1/\delta')$, nonnegative numbers $\tilde s_1^{(u)},\ldots,\tilde s_m^{(u)}$ with $\sum_i\tilde s_i^{(u)}\le2(n-1)$ such that, with probability at least $1-\delta'$, $\tilde s_i^{(u)}\ge s_i^{(u)}$ for every $i\in[m]$. The same holds for $t_1^{(u)},\ldots,t_m^{(u)}$.
\end{lemma}

\paragraph{Influences.}
    For $u\ne v$, the \emph{influences} of $v$ on $H_u$ and on $G_u$ are
\[
    r_{H_u}(v):=\frac{1}{\mathbb E\,X_u(v)}
    \qquad\text{and}\qquad
    r_{G_u}(v):=\frac{1}{\mathbb E\,Y_u(v)},
\]
    where $X_u\sim\mu_{H_u}$ and $Y_u\sim\mu_{G_u}$. If $H_u$ does not depend on $x(v)$, then $\mu_{H_u}$ does not involve this coordinate, and we set $r_{H_u}(v):=0$. The same convention applies to $G_u$. For $\beta=1$, $r_{H_u}(v)$ is of the same order as the density $q_v^{(u)}(0)$ of $X_u(v)$ at zero, as discussed before~\eqref{eq:beta_i_density_submodular}; see Section~\ref{sec:gibbs-tau-one} for the precise relation. For $\beta>1$, this density vanishes, so we use the more general definition above.

    Fix a query $S$ with $F(S)>0$. Summing the subadditivity bounds $f_i^\to(S,\bar S)\le\sum_{u\in S}f_i^\to(u,\bar S)$ and $f_i^\to(S,\bar S)\le\sum_{u\in\bar S}f_i^\to(S,u)$ over $i$ gives
\[
    F(S)\le\sum_{u\in S}H_u(\mathbf 1_{\bar S})
    \qquad\text{and}\qquad
    F(S)\le\sum_{u\in\bar S}G_u(\mathbf 1_S).
\]
    In the first bound, $H_u$ is evaluated at $\mathbf 1_{\bar S}$ for $u\in S$, so only the coordinates in $\bar S$ matter; in the second, $G_u$ is evaluated at $\mathbf 1_S$ for $u\in\bar S$. Accordingly, we measure the influence of $v\in\bar S$ on $F(S)$ by its influences on $H_u$ for $u\in S$, and that of $v\in S$ by its influences on $G_u$ for $u\in\bar S$, taking the largest one:
\begin{equation}\label{eq:gibbs-query-prices}
    \inf_S(v):=
    \begin{cases}
       \displaystyle\max_{u\in\bar S}r_{G_u}(v),&v\in S,\\[5pt]
       \displaystyle\max_{u\in S}r_{H_u}(v),&v\in\bar S.
    \end{cases}
\end{equation}
    For $\theta\ge0$, let $U_\theta:=\{v\in V:\inf_S(v)\le \theta\}$ be the set of coordinates with influence at most $\theta$. The next two lemmas show that leaving free the coordinates of influence at most $\theta$ changes each $f_i$ by at most $\theta\,p_i$, and the influences of all coordinates add up to $O(F(S))$.

\begin{lemma}
\label{lem:gibbs-price-certificate}
    For every $S$ with $F(S)>0$ and every $\theta\ge0$,
\begin{equation}\label{eq:gibbs-price-residual}
    \Delta_i(S,U_\theta)\le\theta p_i
    \qquad(i\in[m]).
\end{equation}
\end{lemma}

\begin{lemma}\label{lem:gibbs-query-prices}
    For every $S$ with $F(S)>0$,
\begin{equation}\label{eq:gibbs-query-price-sides}
    \sum_{v\in\bar S}\inf_S(v)\le\Bigl(1+\frac1\beta\Bigr)F(S)
    \qquad\text{and}\qquad
    \sum_{v\in S}\inf_S(v)\le\Bigl(1+\frac1\beta\Bigr)F(S).
\end{equation}
    Consequently,
\begin{equation}\label{eq:gibbs-query-price-budget}
    \sum_{v\in V}\inf_S(v)\le2\Bigl(1+\frac1\beta\Bigr)F(S)\le4F(S),
\end{equation}
    and $\inf_S(v)\le2F(S)$ for every $v\in V$.
\end{lemma}

    The two lemmas combine as follows. Let the threshold $\theta_j$ halve at each round, and let $I_j\subseteq S$ and $O_j\subseteq\bar S$ consist of the coordinates of influence larger than $\theta_j$. In round $j$, we approximate $f_i(S)$ by $f_i^\to(I_j,O_j)$. By Lemma~\ref{lem:gibbs-price-certificate}, the change $f_i^\to(I_j,O_j)-f_i^\to(I_{j-1},O_{j-1})$ is at most $\theta_{j-1}p_i$. We then apply Bernstein's inequality to these changes, with a union bound over all possible changes. In the proof below, we show that the resulting error, summed over all rounds, is at most $\varepsilon F(S)/2$ plus a term proportional to the total influence $\sum_{v\in V}\inf_S(v)$. Lemma~\ref{lem:gibbs-query-prices} bounds the total influence by $4F(S)$, so this term is also $O(\varepsilon F(S))$.

    Given Lemmas~\ref{lem:gibbs-total-share}--\ref{lem:gibbs-query-prices}, we now prove Theorem~\ref{thm:gibbs-near-quadratic}.

\begin{proofof}{Theorem~\ref{thm:gibbs-near-quadratic}}
    By preprocessing, we may assume $f_i(\varnothing)=f_i(V)=0$ for every $i$, and we discard identically zero components.

    We first bound the scores. By Lemma~\ref{lem:gibbs-total-share}, $\sum_i p_i\le2n(n-1)$. For a query $S$ with $F(S)>0$, Lemma~\ref{lem:gibbs-query-prices} gives $U_{2F(S)}=V$, and Lemma~\ref{lem:gibbs-price-certificate} gives
\begin{equation}\label{eq:gibbs-positive-ideal-scores}
    f_i(S)=\Delta_i(S,V)\le2F(S)p_i.
\end{equation}
    Hence $p_i>0$ for every component. Summing~\eqref{eq:gibbs-positive-ideal-scores} over $i$ also gives $\sum_ip_i\ge1/2$. These properties will be used later.

    We next construct the scores $\tilde p_i$. Let $\beta$ be an integer. For every $u\in V$, apply Lemma~\ref{lem:gibbs-share-estimation} with $\delta'=\delta/(8n)$ to obtain $\tilde s_i^{(u)}$ and $\tilde t_i^{(u)}$, and set
\[
    \tilde p_i:=\frac1m+\sum_{u\in V}\bigl(\tilde s_i^{(u)}+\tilde t_i^{(u)}\bigr)\qquad(i\in[m]).
\]
    Then $\tilde p_i>0$ and $1\le\sum_i\tilde p_i\le1+2n\cdot2(n-1)=4n(n-1)+1$. A union bound over the $2n$ applications of Lemma~\ref{lem:gibbs-share-estimation} shows that, with probability at least $1-\delta/4$, $\tilde s_i^{(u)}\ge s_i^{(u)}$ and $\tilde t_i^{(u)}\ge t_i^{(u)}$ for all $i$ and $u$, and hence $\tilde p_i\ge p_i$ for every $i$. The computation consists of $2n$ applications of Lemma~\ref{lem:gibbs-share-estimation} and $O(nm)$ additions, so it takes randomized polynomial time.

    For the approximation claim, let $q_i$ denote the scores used by the algorithm, so that $q=p$ or $q=\tilde p$, and let $\mathcal E_0$ be the event that $q_i\ge p_i$ for every $i$. If $q=p$, then $\mathcal E_0$ always holds; if $q=\tilde p$, then it holds with probability at least $1-\delta/4$. In both cases, $q_i>0$ for every $i$, $1/2\le\sum_iq_i\le4n(n-1)+1$, and, conditioned on $q$, the draws $\xi_i$ are independent with $\Pr[\xi_i=1]=\min\{1,\Lambda q_i\}$.

    Next, we decompose $f_i(S)$ into increments along the chaining. Fix $S$ with $F(S)>0$; if $F(S)=0$, then every $f_i(S)=0$ and $\tilde F(S)=F(S)$ trivially. As above, for $j\ge0$, let $\theta_j=2^{1-j}F(S)$ and
\[
    I_j=S\setminus U_{\theta_j},\qquad
    O_j=\bar S\setminus U_{\theta_j},\qquad
    \ell_j=|I_j|+|O_j|.
\]
    By Lemma~\ref{lem:gibbs-query-prices}, every influence is at most $2F(S)=\theta_0$, so $I_0=O_0=\varnothing$. Choose $J$ so that $\theta_J$ is smaller than every positive influence $\inf_S(v)$. Then $U_{\theta_J}=U_0$, and Lemma~\ref{lem:gibbs-price-certificate} at $\theta=0$ gives $f_i^\to(I_J,O_J)=f_i(S)-\Delta_i(S,U_0)=f_i(S)$. For $1\le j\le J$, define the increments
\begin{equation}\label{eq:gibbs-chain-increments}
    z_{i,j}
    :=f_i^\to(I_j,O_j)-f_i^\to(I_{j-1},O_{j-1}).
\end{equation}
    They are nonnegative by the monotonicity of $f_i^\to$ (Lemma~\ref{lem:block-minimum-separate-submodularity}), and $\sum_{j=1}^Jz_{i,j}=f_i(S)$. Moreover, $S$ is feasible for the block minimum $f_i^\to(I_j,O_j)$, so $z_{i,j}\le f_i(S)-f_i^\to(I_{j-1},O_{j-1})=\Delta_i(S,U_{\theta_{j-1}})$, which is at most $\theta_{j-1}p_i$ by Lemma~\ref{lem:gibbs-price-certificate}. 

    Since $I_j\cup O_j=\{v\in V:\inf_S(v)>\theta_j\}$ and $\ell_j=|I_j\cup O_j|$, exchanging the order of summation gives
\[
    \sum_{j=1}^J\ell_j\theta_{j-1}
    =\sum_{v\in V}\ \sum_{\substack{1\le j\le J\\ \theta_j<\inf_S(v)}}\theta_{j-1}.
\]
    The inner sum is empty if $\inf_S(v)=0$. Otherwise, let $k$ be the smallest index with $\theta_k<\inf_S(v)$; then $1\le k\le J$ by $\inf_S(v)\le \theta_0$ and the choice of $J$, and $\inf_S(v)\le \theta_{k-1}$. Since the thresholds halve at each step, the inner sum is
\[
    \sum_{j=k}^J\theta_{j-1}
    <2\theta_{k-1}
    =4\theta_k
    <4\inf_S(v).
\]
    Summing over $v\in V$, the bound $z_{i,j}\le \theta_{j-1}p_i$ and Lemma~\ref{lem:gibbs-query-prices} give
\begin{equation}\label{eq:gibbs-chain-budget}
    \sum_{j=1}^J\ell_j\max_{i\in[m]}\frac{z_{i,j}}{p_i}
    \le\sum_{j=1}^J\ell_j\theta_{j-1}
    \le4\sum_{v\in V}\inf_S(v)
    \le16F(S).
\end{equation}

    We now show that, on $\mathcal E_0$, $\tilde F$ approximates $F$ on all queries simultaneously. On $\mathcal E_0$, we have $z_{i,j}/q_i\le z_{i,j}/p_i$, so~\eqref{eq:gibbs-chain-budget} also holds with $q_i$ in place of $p_i$. Condition on $q$, and apply Lemma~\ref{lem:Bernstein} to all possible increment vectors. An increment with $\ell$ fixed coordinates is specified by its old and new inclusion and exclusion blocks. Each fixed coordinate has four possible labels, so there are at most $\binom n\ell4^\ell\le(4n)^\ell$ such vectors, independently of the query. For each of them, apply Lemma~\ref{lem:Bernstein} with the weights $q_i$, the parameter $\Lambda=128\varepsilon^{-2}\log(64n/\delta)$ of Algorithm~\ref{alg:sampling_n2_code}, $\eta=\varepsilon/2$, and $x=\ell\log(64n/\delta)$. A union bound fails with conditional probability at most
\[
    2\sum_{\ell=1}^n(4n)^\ell
       \left(\frac{\delta}{64n}\right)^\ell
    \le2\sum_{\ell=1}^\infty(\delta/16)^\ell
    <\frac\delta4.
\]
    On the intersection of $\mathcal E_0$ and the complementary event, the identity $\sum_jz_{i,j}=f_i(S)$ and~\eqref{eq:gibbs-chain-budget} with $q_i$ in place of $p_i$ give, for every $S$ with $F(S)>0$,
\begin{align*}
    \big|\tilde F(S)-F(S)\big|
    &\le\sum_{j=1}^J
       \left|\sum_i(w_i-1)z_{i,j}\right|\\
    &\le\frac\varepsilon2\sum_{j=1}^J\sum_i z_{i,j}
       +\frac{2\log(64n/\delta)}{\varepsilon\Lambda}
        \sum_{j=1}^J\ell_j\max_i\frac{z_{i,j}}{q_i}\\
    &\le\left(\frac\varepsilon2
       +\frac{32\log(64n/\delta)}{\varepsilon\Lambda}\right)F(S)
    =\frac{3\varepsilon}{4}F(S).
\end{align*}

    Finally, we bound the sparsity $|\supp(w)|$, where $\supp(w)=\{i\in[m]:w_i>0\}$ is the set of sampled components. Conditioned on $q$, independence gives
\[
    \mathbb E\bigl[2^{|\supp(w)|}\bigm|q\bigr]
    =\prod_i\bigl(1+\min\{1,\Lambda q_i\}\bigr)
    \le\exp\Bigl(\Lambda\sum_i q_i\Bigr).
\]
    By Markov's inequality and $\sum_i q_i\ge1/2$,
\[
    \Pr\Bigl[|\supp(w)|>2\Lambda\sum_i q_i\Bigm|q\Bigr]
    \le\exp\Bigl(-(2\log2-1)\Lambda\sum_i q_i\Bigr)
    <\frac\delta4.
\]
    The complement of $\mathcal E_0$, the failure of the union bound, and this event each have probability at most $\delta/4$. Hence, with probability at least $1-3\delta/4$, $\tilde F$ is an $\varepsilon$-approximation of $F$ with sparsity at most $2\Lambda\sum_iq_i\le2\Lambda(4n(n-1)+1)$, which is $O(n^2\varepsilon^{-2}\log(2n/\delta))$.
\end{proofof}

    The rest of this appendix is organized as follows. Sections~\ref{sec:gibbs-shares}, \ref{sec:gibbs-residual}, and~\ref{sec:gibbs-budget} prove Lemmas~\ref{lem:gibbs-total-share}, \ref{lem:gibbs-price-certificate}, and~\ref{lem:gibbs-query-prices}, respectively, and Section~\ref{sec:gibbs-polynomial-time} proves Lemma~\ref{lem:gibbs-share-estimation} by sampling from a convex body. Section~\ref{sec:gibbs-tau-one} discusses the case $\beta=1$, where the influence is comparable to the marginal density at zero.

\subsection{Shares of Gibbs distributions}\label{sec:gibbs-shares}

    In this subsection, we prove Lemma~\ref{lem:gibbs-total-share}. We first set up the shares and influences of a general Gibbs distribution, which are also used in the next two subsections.

    Let $h_1,\ldots,h_m:2^E\to\mathbb R_{\ge0}$ be monotone submodular functions with $h_i(\varnothing)=0$, and put $h=\sum_i h_i$. Recall that $\widehat h$ and $\widehat h_i$ denote their Lov\'asz extensions~\eqref{eq:Lovasz_ext}. A coordinate $v$ is \emph{effective} if $h(\{v\})>0$. Since $h$ is monotone and submodular, the value of $h$ and of every $h_i$ is independent of the ineffective coordinates. We henceforth restrict $h$ and every $h_i$ to the set of effective coordinates of $h$, which we still denote by $E$, and put $d=|E|$. If $d=0$, all shares and influences below are defined to be zero.

    For $d>0$, let
\begin{equation}\label{eq:gibbs-distribution}
    \mu_h(dx)
    =\frac{1}{Z_h}\,e^{-\beta\widehat h(x)}\prod_{v\in E}x(v)^{\beta-1}\,dx,
    \qquad
    Z_h=\int_{(0,\infty)^{E}}e^{-\beta\widehat h(y)}
        \prod_{v\in E}y(v)^{\beta-1}\,dy,
\end{equation}
    for $x\in(0,\infty)^E$. For $\beta=1$, this is the distribution with density proportional to $e^{-\widehat h(x)}$ on the nonnegative orthant. The boundary of the orthant has measure zero. The distribution is well defined: for nonnegative $x$,
\[
    \max_{v\in E}h(\{v\})x(v)
    \le\widehat h(x)
    \le\sum_{v\in E}h(\{v\})x(v).
\]
    The lower bound follows from the monotonicity of $\widehat h$, since $x\ge x(v)\mathbf 1_{\{v\}}$; the upper bound follows from the greedy formula~\eqref{eq:Lovasz_ext_per} and $h(B\cup\{v\})-h(B)\le h(\{v\})$. Put $\alpha=\min_{v\in E}h(\{v\})>0$. Since a maximum is at least the average, the lower bound gives
\[
    e^{-\beta\widehat h(x)}\prod_{v\in E}x(v)^{\beta-1}
    \le\prod_{v\in E}x(v)^{\beta-1}e^{-\beta\alpha x(v)/d}.
\]
    Each one-dimensional factor is integrable near zero because $\beta-1\ge0$, and it decays exponentially at infinity. Hence $0<Z_h<\infty$, and $\mu_h$ has finite polynomial moments. Since the density is positive on $(0,\infty)^E$, we also have $\mathbb E_{\mu_h}X(v)>0$ for every $v\in E$.

    Define
\begin{equation}\label{eq:gibbs-shares-prices}
    s_i(h):=\mathbb E_{\mu_h}\widehat h_i(X)
    \qquad\text{and}\qquad
    \pi_h(v):=\frac{1}{\mathbb E_{\mu_h}X(v)},
\end{equation}
    and set $\pi_h(v)=0$ for ineffective coordinates $v$. We call $s_i(h)$ the \emph{share} of $h_i$ and $\pi_h(v)$ the \emph{influence} of $v$ on $h$. For the family $(f_i^\to(u,\cdot))_i$ on $V\setminus\{u\}$, we have $\mu_h=\mu_{H_u}$, $s_i(h)=s_i^{(u)}$, and $\pi_h(v)=r_{H_u}(v)$; the family $(f_i^\to(\cdot,u))_i$ gives $t_i^{(u)}$ and $r_{G_u}(v)$ in the same way.

    Ineffective coordinates affect neither $h$ nor any $h_i$, and their influence is zero. Hence, for every set $C$ of coordinates, $h(C)=h(C\cap E)$, $h_i(C)=h_i(C\cap E)$, and sums and maxima of influences over $C$ equal those over $C\cap E$. Thus every statement below about a set $C$ follows from the same statement for $C\cap E$, and in the proofs we assume that all coordinates are effective without loss of generality.

    The total share of a Gibbs distribution is its number of effective coordinates.

\begin{lemma}\label{lem:gibbs-share-identity}
    For every such family, $\sum_{i=1}^m s_i(h)=d$.
\end{lemma}

\begin{proof}
    We may assume $d>0$. Otherwise all quantities vanish. All expectations are with respect to $\mu_h$. For $t>0$, put
\[
    Z(t):=\int_{(0,\infty)^E}e^{-t\widehat h(x)}
          \prod_{v\in E}x(v)^{\beta-1}\,dx,
\]
    so that $Z(\beta)=Z_h$. Substituting $x=y/t$, so that $dx=t^{-d}\,dy$, and using the positive homogeneity $t\,\widehat h(y/t)=\widehat h(y)$, we get
\[
    Z(t)
    =\int_{(0,\infty)^E}e^{-\widehat h(y)}
       \prod_{v\in E}\Bigl(\frac{y(v)}{t}\Bigr)^{\beta-1}\,\frac{dy}{t^d}
    =t^{-\beta d}Z(1).
\]
    We compute $Z'(\beta)$ in two ways. First, differentiating under the integral sign gives
\[
    Z'(t)=-\int_{(0,\infty)^E}\widehat h(x)\,e^{-t\widehat h(x)}\prod_{v\in E}x(v)^{\beta-1}\,dx.
\]
    This is allowed because, for $t$ in a compact interval $[a,b]\subset(0,\infty)$, the integrand is at most $\widehat h(x)e^{-a\widehat h(x)}\prod_vx(v)^{\beta-1}$, which is integrable by the bounds following~\eqref{eq:gibbs-distribution}. Since $\mu_h$ has density $e^{-\beta\widehat h(x)}\prod_vx(v)^{\beta-1}/Z(\beta)$, taking $t=\beta$ gives $Z'(\beta)/Z(\beta)=-\mathbb E\,\widehat h(X)$. Second, $Z(t)=t^{-\beta d}Z(1)$ gives $Z'(t)/Z(t)=-\beta d/t$. Comparing the two at $t=\beta$, and using the linearity of the Lov\'asz extension,
\[
    \sum_{i=1}^m s_i(h)
    =\mathbb E\,\widehat h(X)
    =-\frac{Z'(\beta)}{Z(\beta)}
    =d.
\]
\end{proof}

    Lemma~\ref{lem:gibbs-total-share} follows by applying Lemma~\ref{lem:gibbs-share-identity} to the families $(f_i^\to(u,\cdot))_i$ and $(f_i^\to(\cdot,u))_i$, each with at most $n-1$ effective coordinates, and summing over $u\in V$.

\subsection{The residual bound}\label{sec:gibbs-residual}

    In this subsection, we prove Lemma~\ref{lem:gibbs-price-certificate}. We keep the notation of Section~\ref{sec:gibbs-shares}. The key step is the following bound: a component on a set is at most its share times the largest influence in the set.

\begin{lemma}\label{lem:gibbs-low-price}
    For every set $C\subseteq E$ and every $i\in[m]$,
\begin{equation}\label{eq:gibbs-low-price}
    h_i(C)\le s_i(h)\max_{v\in C}\pi_h(v),
\end{equation}
    where the maximum over the empty set is zero.
\end{lemma}

\begin{proofof}{Lemma~\ref{lem:gibbs-low-price}}
    If $C=\varnothing$, both sides vanish. Otherwise put
\[
    \theta=\min_{v\in C}\mathbb E_{\mu_h}X(v)>0,
    \qquad\text{so that}\qquad
    \frac1\theta=\max_{v\in C}\pi_h(v).
\]
    The Lov\'asz extension $\widehat h_i$ is convex, because $h_i$ is submodular. It is coordinatewise nondecreasing, because $h_i$ is monotone. It is also positively homogeneous and integrable under $\mu_h$. The vector $\mathbb E_{\mu_h}X$ satisfies $\mathbb E_{\mu_h}X\ge\theta\mathbf 1_C$ coordinatewise. Jensen's inequality and these properties give
\[
    s_i(h)
    =\mathbb E_{\mu_h}\widehat h_i(X)
    \ge\widehat h_i\bigl(\mathbb E_{\mu_h}X\bigr)
    \ge\widehat h_i(\theta\mathbf 1_C)
    =\theta\,h_i(C).
\]
    Dividing by $\theta$ proves~\eqref{eq:gibbs-low-price}.
\end{proofof}

Now we are ready to prove Lemma~\ref{lem:gibbs-price-certificate}. 

\begin{proofof}{Lemma~\ref{lem:gibbs-price-certificate}}
    Let $U=U_\theta$. By the two-block residual inequality \eqref{eq:bound_Delta_S1_S2}, stated formally in Lemma~\ref{lem:block-residual-two-block},
\[
    \Delta_i(S,U)
    \le f_i^\to(S\cap U,\bar S)
       +f_i^\to(S,U\setminus S).
\]
    To bound the two terms, we use the influences. For disjoint $A,D\subseteq V$, subadditivity of $f_i^\to$ in its first argument gives $f_i^\to(A,D)\le\sum_{u\in A}f_i^\to(u,D)$. For each $u\in A$, the function $D\mapsto f_i^\to(u,D)$ is the $i$th member of the family defining $H_u$, and its share is $s_i^{(u)}$. Applying Lemma~\ref{lem:gibbs-low-price} to it, and arguing in the same way with subadditivity in the second argument and the families defining $G_u$, we obtain
\begin{equation}\label{eq:gibbs-terminal-local-bound}
\begin{aligned}
    f_i^\to(A,D)
    &\le\sum_{u\in A}s_i^{(u)}
         \max_{v\in D}r_{H_u}(v),\\
    f_i^\to(A,D)
    &\le\sum_{u\in D}t_i^{(u)}
         \max_{v\in A}r_{G_u}(v),
\end{aligned}
\end{equation}
    where a maximum over the empty set is zero. The first inequality is the influence version of~\eqref{eq:beta_i_density_submodular}; for $\beta=1$, Section~\ref{sec:gibbs-tau-one} recovers~\eqref{eq:beta_i_density_submodular} itself, with constant two.

    We bound the first term by the second inequality in~\eqref{eq:gibbs-terminal-local-bound}, with $A=S\cap U$ and $D=\bar S$. Every $v\in S\cap U$ satisfies $r_{G_u}(v)\le \inf_S(v)\le \theta$ for all $u\in\bar S$, by~\eqref{eq:gibbs-query-prices}. Hence
\[
    f_i^\to(S\cap U,\bar S)\le \theta\sum_{u\in\bar S}t_i^{(u)}.
\]
    We bound the second term by the first inequality in~\eqref{eq:gibbs-terminal-local-bound}, with $A=S$ and $D=U\setminus S$. Every $v\in U\setminus S$ satisfies $r_{H_u}(v)\le \inf_S(v)\le \theta$ for all $u\in S$. Hence
\[
    f_i^\to(S,U\setminus S)\le \theta\sum_{u\in S}s_i^{(u)}.
\]
    The two right-hand sides are disjoint parts of the sum $p_i=\sum_{u\in V}(s_i^{(u)}+t_i^{(u)})$. Adding them proves $\Delta_i(S,U)\le\theta p_i$.
\end{proofof}

\subsection{The influence budget}\label{sec:gibbs-budget}

    In this subsection, we prove Lemma~\ref{lem:gibbs-query-prices}. For a vector $x$ with distinct coordinates, order $E$ as $x(\sigma_1)>\cdots>x(\sigma_d)$, and let $\mathrm{Pre}_\sigma(j)=\{\sigma_1,\ldots,\sigma_j\}$ be the set of its $j$ largest coordinates, as in~\eqref{eq:Lovasz_ext_per}. Since $\mu_h$ has a density, the coordinates of $X\sim\mu_h$ are distinct almost surely. Put
\[
    b_{\sigma_j}(x)
    :=h(\mathrm{Pre}_{\sigma}(j))
       -h(\mathrm{Pre}_{\sigma}(j-1)).
\]
    These are the greedy marginals of $h$, and
\begin{equation}\label{eq:gibbs-greedy-budget}
    b_v(x)\ge0,\qquad
    \sum_{v\in C}b_v(x)\le h(C),\qquad
    \sum_{v\in E}b_v(x)=h(E).
\end{equation}
    The first inequality is monotonicity. The last identity holds because the marginals telescope along the whole order to $h(\mathrm{Pre}_\sigma(d))-h(\varnothing)=h(E)$. For the middle one, if $v=\sigma_j\in C$, submodularity gives $b_v(x)\le h(\mathrm{Pre}_\sigma(j)\cap C)-h(\mathrm{Pre}_\sigma(j-1)\cap C)$, and these upper bounds telescope over $v\in C$ to $h(C)$. We use the averaged greedy marginals $a_h(v):=\mathbb E_{\mu_h}b_v(X)$, and set $a_h(v)=0$ for ineffective coordinates $v$. Taking expectations in~\eqref{eq:gibbs-greedy-budget} gives
\begin{equation}\label{eq:gibbs-share-price-budgets}
    \sum_{v\in C}a_h(v)\le h(C)\quad(C\subseteq E),
    \qquad
    \sum_{v\in E}a_h(v)=h(E).
\end{equation}
    For $\beta=1$, the averaged greedy marginals coincide with the marginal densities at zero used in the overview; see Section~\ref{sec:gibbs-tau-one}. The influences are bounded by $1+1/\beta$ times the averaged greedy marginals, which gives the following budget.

\begin{lemma}\label{lem:gibbs-influence-budget}
    For every such family and every set $C$ of coordinates,
\begin{equation}\label{eq:gibbs-price-budget}
    \sum_{v\in C}\pi_h(v)\le\Bigl(1+\frac1\beta\Bigr)h(C)\le2h(C).
\end{equation}
\end{lemma}

\begin{proof}
    We may assume $d>0$; otherwise all quantities vanish. All expectations are with respect to $\mu_h$.

    We first establish two integral identities. Fix $v\in E$; we show
\begin{equation}\label{eq:gibbs-integral-identities}
    \mathbb E\bigl[X(v)\,b_v(X)\bigr]=1,
    \qquad
    \mathbb E\bigl[X(v)^2\,b_v(X)\bigr]
    =\Bigl(1+\frac1\beta\Bigr)\mathbb E X(v).
\end{equation}
    Fix $z\in(0,\infty)^{E\setminus\{v\}}$ with pairwise distinct coordinates; such vectors have full measure. Write $(t,z)$ for the vector whose $v$th coordinate is $t>0$ and whose other coordinates are given by $z$, and put $\varphi(t)=\widehat h(t,z)$. Rearranging the greedy formula~\eqref{eq:Lovasz_ext_per} gives
\[
    \widehat h(x)=\sum_{w\in E}x(w)\,b_w(x).
\]
    When $t$ ranges over an open interval containing none of the values $z(w)$, the order of the coordinates of $(t,z)$ does not change, so every $b_w(t,z)$ is constant there, and
\[
    \varphi(t)=t\,b_v(t,z)+\sum_{w\ne v}z(w)\,b_w(t,z)
\]
    is linear in $t$ with slope $b_v(t,z)$. Hence $\varphi$ is continuous and piecewise linear, and $\varphi'(t)=b_v(t,z)$ away from the values $z(w)$. Moreover, $0\le b_v\le h(\{v\})$ and $\varphi(t)\ge h(\{v\})t$.

    Let $s\in\{\beta,\beta+1\}$ and $g(t)=t^se^{-\beta\varphi(t)}$. On every compact subinterval of $(0,\infty)$, the function $g$ is absolutely continuous, and for almost every $t$,
\[
    g'(t)
    =s\,t^{s-1}e^{-\beta\varphi(t)}
     -\beta\,t^s\,b_v(t,z)\,e^{-\beta\varphi(t)}.
\]
    Since $s>0$ and $\varphi\ge0$, we have $g(t)\to0$ as $t\downarrow0$. Since $\varphi(t)\ge h(\{v\})t$ and $h(\{v\})>0$, we have $g(t)\to0$ as $t\to\infty$. The two quantities $s\,t^{s-1}e^{-\beta\varphi(t)}$ and $\beta\,t^sb_v(t,z)e^{-\beta\varphi(t)}$, whose difference is $g'$, are nonnegative and integrable on $(0,\infty)$, because $b_v\le h(\{v\})$ and $e^{-\beta\varphi(t)}\le e^{-\beta h(\{v\})t}$. Integrating $g'$ over $[a,b]$ and letting $a\downarrow0$ and $b\to\infty$, we obtain
\[
    \beta\int_0^\infty t^s\,b_v(t,z)\,e^{-\beta\varphi(t)}\,dt
    =s\int_0^\infty t^{s-1}e^{-\beta\varphi(t)}\,dt.
\]
    Multiply both sides by $\prod_{w\ne v}z(w)^{\beta-1}$, integrate over $z$ by Tonelli's theorem, and divide by $Z_h$. Since $\mu_h$ has density $e^{-\beta\widehat h(x)}\prod_wx(w)^{\beta-1}/Z_h$, writing $t^s=t^{s-\beta+1}\cdot t^{\beta-1}$ turns both sides into expectations:
\[
    \beta\,\mathbb E\bigl[X(v)^{s-\beta+1}b_v(X)\bigr]=s\,\mathbb E\bigl[X(v)^{s-\beta}\bigr].
\]
    Taking $s=\beta$ and $s=\beta+1$ gives~\eqref{eq:gibbs-integral-identities}.

    We now use these identities to compare $\pi_h(v)$ with $a_h(v)$. By the Cauchy--Schwarz inequality, applied to $\sqrt{b_v(X)}$ and $X(v)\sqrt{b_v(X)}$, and by~\eqref{eq:gibbs-integral-identities},
\[
    1=\bigl(\mathbb E[X(v)b_v(X)]\bigr)^2
    \le\mathbb E[b_v(X)]\;\mathbb E\bigl[X(v)^2b_v(X)\bigr]
    =\Bigl(1+\frac1\beta\Bigr)a_h(v)\,\mathbb EX(v).
\]
    Dividing by $\mathbb EX(v)>0$ gives $\pi_h(v)\le(1+1/\beta)a_h(v)$. Summing over $v\in C$ and using $\sum_{v\in C}a_h(v)\le h(C)$ proves~\eqref{eq:gibbs-price-budget}.
\end{proof}

    Lemma~\ref{lem:gibbs-influence-budget} bounds the total influence within a single Gibbs distribution. To prove Lemma~\ref{lem:gibbs-query-prices},  we further need to bound the influences of many Gibbs distributions $\mu_{H_u}$ and $\mu_{G_u}$. Next, we compare different Gibbs distributions. We use positive association in the following form~\cite{Sarkar69,KR80}.

\begin{lemma}
\label{lem:mtp2-association}
    For $x,y\in\mathbb R_{\ge0}^{E}$, write $x\wedge y$ and $x\vee y$ for their coordinatewise minimum and maximum. Suppose a probability density $p$ on $\mathbb R_{\ge0}^{E}$ satisfies
\[
    p(x\wedge y)p(x\vee y)\ge p(x)p(y)
    \qquad(x,y\in\mathbb R_{\ge0}^{E}).
\]
    Then $\mathbb E[G_1G_2]\ge\mathbb E G_1\,\mathbb E G_2$ for all coordinatewise nondecreasing functions $G_1,G_2:\mathbb R_{\ge0}^{E}\to\mathbb R$, whenever the expectations exist.
\end{lemma}

    We apply Lemma~\ref{lem:mtp2-association} to Gibbs distributions. The next lemma says that reweighting $\mu_h$ by a monotone factor moves the expectation of a nondecreasing function in the same direction.

\begin{lemma}\label{lem:gibbs-association}
    Let $h$ be normalized, monotone, and submodular, let $G:\mathbb R_{\ge0}^{E}\to\mathbb R_{\ge0}$ be coordinatewise nondecreasing, and let $\psi:\mathbb R_{\ge0}^{E}\to(0,1]$. If $\psi$ is coordinatewise nondecreasing, then
\begin{equation}\label{eq:gibbs-monotone-tilting}
    \mathbb E_{\mu_h}[G\psi]
    \ge\mathbb E_{\mu_h}G\,\mathbb E_{\mu_h}\psi,
\end{equation}
    and the inequality is reversed if $\psi$ is coordinatewise nonincreasing. All expectations are assumed finite.
\end{lemma}

\begin{proof}
    For $x,y\in(0,\infty)^E$, write $\widehat h(x)=\int_0^\infty h(\{x\ge\lambda\})\,d\lambda$ with $\{x\ge\lambda\}:=\{v:x(v)\ge\lambda\}$. Then $\{x\wedge y\ge\lambda\}=\{x\ge\lambda\}\cap\{y\ge\lambda\}$ and $\{x\vee y\ge\lambda\}=\{x\ge\lambda\}\cup\{y\ge\lambda\}$, so applying the submodularity of $h$ for each $\lambda$ and integrating gives
\begin{equation}\label{eq:gibbs-lattice-submodular}
    \widehat h(x\wedge y)+\widehat h(x\vee y)\le\widehat h(x)+\widehat h(y).
\end{equation}
    Moreover, $\prod_v(x\wedge y)(v)^{\beta-1}(x\vee y)(v)^{\beta-1}=\prod_vx(v)^{\beta-1}y(v)^{\beta-1}$, since $\{(x\wedge y)(v),(x\vee y)(v)\}=\{x(v),y(v)\}$ for every $v$. Hence the density $p$ of $\mu_h$ satisfies $p(x\wedge y)p(x\vee y)\ge p(x)p(y)$; if $x$ or $y$ has a zero coordinate, the right-hand side vanishes, since we take $p$ to be zero outside $(0,\infty)^E$. Lemma~\ref{lem:mtp2-association} applied to $G$ and $\psi$, or to $G$ and $1-\psi$, proves the claim.
\end{proof}

    We now compare the Gibbs distributions of different functions. The next lemma shows that restricting $h$ to fewer coordinates, or increasing it by a nondecreasing function, can only decrease the expectation of a nondecreasing function. For $D\subseteq E$, let $h|_D$ be the restriction of $h$ to subsets of $D$; its Gibbs distribution $\mu_{h|_D}$ on $\mathbb R_{\ge0}^{D}$ is defined by~\eqref{eq:gibbs-distribution} with $D$ in place of $E$, and $\widehat{h|_D}(x)=\widehat h(x,0)$ is obtained by setting the coordinates outside $D$ to zero. As in Section~\ref{sec:gibbs-shares}, each Gibbs distribution lives on the effective coordinates of its function.

\begin{lemma}
\label{lem:gibbs-comparison}
    Let $h$ be normalized, monotone, and submodular, and let $G$ be nonnegative, coordinatewise nondecreasing, and of at most linear growth, depending only on coordinates that are effective for all functions involved.
\begin{enumerate}
    \item For $D\subseteq E$,
\begin{equation}\label{eq:gibbs-restriction}
    \mathbb E_{\mu_{h|_D}}G(X)
    \le\mathbb E_{\mu_h}G(X(D)),
\end{equation}
    where $X(D)$ denotes the subvector of $X$ on $D$.
    \item If $h_0,h_1$ are normalized, monotone, and submodular on the same ground set and $h_1-h_0$ is nonnegative and coordinatewise nondecreasing, then
\begin{equation}\label{eq:gibbs-potential-comparison}
    \mathbb E_{\mu_{h_1}}G(X)
    \le\mathbb E_{\mu_{h_0}}G(X).
\end{equation}
\end{enumerate}
\end{lemma}

\begin{proof}
    First suppose all coordinates are effective. Write $x$ for the coordinates in $D$ and $y$ for the remaining coordinates. For fixed $y$, the conditional law of $x$ under $\mu_h$ is proportional to
\[
    e^{-\beta\widehat{h|_D}(x)}\prod_{v\in D}x(v)^{\beta-1}\,\psi(x),\qquad
    \psi(x):=\exp\bigl(-\beta\bigl(\widehat h(x,y)-\widehat h(x,0)\bigr)\bigr),
\]
    because $\widehat h(x,0)=\widehat{h|_D}(x)$ and the factors $y(v)^{\beta-1}$ do not depend on $x$.
    The factor $\psi$ lies in $(0,1]$, since $\widehat h$ is coordinatewise nondecreasing. It is coordinatewise nondecreasing in $x$: if $x\le x'$, then $(x,y)\wedge(x',0)=(x,0)$ and $(x,y)\vee(x',0)=(x',y)$, so applying the submodularity of $h$ to their threshold sets for each $\lambda$ and integrating gives $\widehat h(x,y)+\widehat h(x',0)\ge\widehat h(x,0)+\widehat h(x',y)$. Lemma~\ref{lem:gibbs-association} applied to $h|_D$ therefore shows that the conditional expectation of $G$ is no smaller than its expectation under $\mu_{h|_D}$. Averaging over $y$ proves \eqref{eq:gibbs-restriction}.

    For the second inequality, linearity of the Lov\'asz extension gives
\[
    \psi(x):=\exp\bigl(-\beta\bigl(\widehat h_1(x)-\widehat h_0(x)\bigr)\bigr)
    =e^{-\beta\widehat{(h_1-h_0)}(x)},
\]
    so that $\mu_{h_1}$ is $\mu_{h_0}$ tilted by $\psi$ and renormalized. The function $h_1-h_0$ is nonnegative and nondecreasing, so $\psi$ lies in $(0,1]$ and is coordinatewise nonincreasing. The reverse inequality in Lemma~\ref{lem:gibbs-association}, applied under $\mu_{h_0}$, proves~\eqref{eq:gibbs-potential-comparison}.

    It remains to remove the assumption that all coordinates are effective. For $\eta>0$ and a set function $g$ on $E$, put $g^\eta(C):=g(C)+\eta|C|$ for $C\subseteq E$. If $g$ is normalized, monotone, and submodular, then so is $g^\eta$; moreover, every coordinate is effective for $g^\eta$, and $\widehat{g^\eta}(x)=\widehat g(x)+\eta\sum_vx(v)$. Since $h_1^\eta-h_0^\eta=h_1-h_0$, the argument above applies to $h^\eta$, $h_0^\eta$, and $h_1^\eta$, and gives~\eqref{eq:gibbs-restriction} and~\eqref{eq:gibbs-potential-comparison} for their Gibbs distributions.

    We now let $\eta\downarrow0$. An ineffective coordinate $v$ of $g$ does not affect $\widehat g$, so it appears in the density of $\mu_{g^\eta}$ only through an independent factor $x(v)^{\beta-1}e^{-\beta\eta x(v)}$, and it can be integrated out. After this, $\mu_{g^\eta}$ is the distribution $\mu_g$ reweighted by $e^{-\beta\eta\sum_vx(v)}$ and renormalized, where the sum is over the effective coordinates of $g$. This weight increases to $1$ as $\eta\downarrow0$, so by monotone convergence, $\mathbb E_{\mu_{g^\eta}}G\to\mathbb E_{\mu_g}G$ for every $G$ as in the lemma. Passing to the limit proves~\eqref{eq:gibbs-restriction} and~\eqref{eq:gibbs-potential-comparison}.
\end{proof}

\begin{proofof}{Lemma~\ref{lem:gibbs-query-prices}}
    Let $h(D)=\sum_i f_i^\to(S,D)$ for $D\subseteq\bar S$. Then $h(\bar S)=F(S)$. Fix $u\in S$. For $A\subseteq B\subseteq\bar S$, submodularity gives
\begin{equation}\label{eq:gibbs-terminal-query-comparison}
    f_i^\to(S,B)-f_i^\to(S,A)
    \ge f_i^\to(u,B)-f_i^\to(u,A).
\end{equation}
    Indeed, take minimizers for $f_i^\to(S,B)$ and $f_i^\to(u,A)$. Their union is feasible for $f_i^\to(S,A)$, and their intersection is feasible for $f_i^\to(u,B)$. Summing over $i$ shows that $h-H_u|_{\bar S}$ is nondecreasing; since both functions vanish at the empty set, it is also nonnegative.

    Let $v\in\bar S$. If $v$ is ineffective for $H_u$, then $r_{H_u}(v)=0$. Otherwise $v$ is effective for $H_u|_{\bar S}$, and also for $h$, because $h(\{v\})\ge H_u(\{v\})>0$. Apply the second part of Lemma~\ref{lem:gibbs-comparison} with $h_0=H_u|_{\bar S}$, $h_1=h$, and $G(x)=x(v)$, and then its first part with $D=\bar S$ to the family defining $H_u$. This gives
\[
    \mathbb E_{\mu_h}X(v)
    \le\mathbb E_{\mu_{H_u|_{\bar S}}}X(v)
    \le\mathbb E_{\mu_{H_u}}X_u(v),
    \qquad\text{that is,}\qquad
    r_{H_u}(v)\le \pi_h(v).
\]
    Taking the maximum over $u\in S$, which gives $\inf_S(v)$ by the definition~\eqref{eq:gibbs-query-prices}, summing over $v\in\bar S$, and applying Lemma~\ref{lem:gibbs-influence-budget} to $h$, we obtain
\[
    \sum_{v\in\bar S}\inf_S(v)
    \le\sum_{v\in\bar S}\pi_h(v)
    \le\Bigl(1+\frac1\beta\Bigr)h(\bar S)
    =\Bigl(1+\frac1\beta\Bigr)F(S).
\]
    For the other side, apply the same argument to $h(A)=\sum_i f_i^\to(A,\bar S)$ on $S$, which satisfies $h(S)=F(S)$, and compare it with the families defining $G_u$ for $u\in\bar S$. The analogue of~\eqref{eq:gibbs-terminal-query-comparison}, $f_i^\to(B,\bar S)-f_i^\to(A,\bar S)\ge f_i^\to(B,u)-f_i^\to(A,u)$ for $A\subseteq B\subseteq S$, is proved by the same exchange of minimizers. This gives $\sum_{v\in S}\inf_S(v)\le(1+1/\beta)F(S)$. Each influence $\inf_S(v)$ is a summand of one of these two sums, so $\inf_S(v)\le(1+1/\beta)F(S)\le2F(S)$. Adding the two bounds gives~\eqref{eq:gibbs-query-price-budget}.
\end{proofof}

\subsection{Approximating the scores}
\label{sec:gibbs-polynomial-time}

    In this subsection, we prove Lemma~\ref{lem:gibbs-share-estimation}. Throughout, $\beta\ge1$ is a fixed integer. For a general family $h_1,\ldots,h_m$ as in Section~\ref{sec:gibbs-shares}, we construct a convex body whose uniform distribution determines the shares (Claim~\ref{clm:gibbs-body}), and show that averages over polynomially many approximate samples from it estimate all shares within a factor two (Claim~\ref{clm:gibbs-estimates}). We use the following result from \cite{JLLS23} to sample from a convex body.

\begin{theorem}[{\cite[Theorem~3.22]{JLLS23}}]\label{thm:jlls-sampler}
    Let $A\subseteq\mathbb R^N$ be a convex body, and suppose that known radii $0<r\le R$ satisfy $rB_2^N\subseteq A\subseteq RB_2^N$, where $B_2^N=\{x\in\mathbb R^N:\|x\|_2\le1\}$. Given $\eta\in(0,1)$ and an oracle that decides whether a queried point belongs to $A$, a randomized algorithm can generate a point whose distribution has total variation distance at most $\eta$ from the uniform distribution on $A$. Its cost consists of
\[
    O\!\left(N^3\left(\log\frac{NR}{\eta r}\right)^{O(1)}\right)
\]
    membership queries and $\bigl(N\log(NR/(\eta r))\bigr)^{O(1)}$ additional time in the membership-oracle model.
\end{theorem}

    Let $h_1,\ldots,h_m$ be a family as in Section~\ref{sec:gibbs-shares}, and put $h=\sum_i h_i$. As there, we restrict to the effective coordinates, so that $h(\{v\})>0$ for every $v\in E$, and put $d=|E|$. Assume $d\ge1$, and put $c_v=h(\{v\})>0$ for $v\in E$. For a vector $z\in\mathbb R_{\ge0}^E$, let $z/c\in\mathbb R_{\ge0}^E$ denote the vector with entries $z(v)/c_v$, and define
\[
    H(z)=\widehat h(z/c),\qquad H_i(z)=\widehat h_i(z/c)\qquad(z\in\mathbb R_{\ge0}^E).
\]
    The Lov\'asz extension $\widehat h$ is coordinatewise nondecreasing and positively homogeneous, with $\widehat h(\mathbf 1_{\{v\}})=h(\{v\})$. Being convex and positively homogeneous, it is also subadditive. Hence $\max_vx(v)h(\{v\})\le\widehat h(x)\le\sum_vx(v)h(\{v\})$ for $x\in\mathbb R_{\ge0}^E$, and taking $x=z/c$ gives
\begin{equation}\label{eq:gibbs-rounding}
    \max_v z(v)\le H(z)\le\sum_v z(v)\qquad(z\in\mathbb R_{\ge0}^E).
\end{equation}
    To produce the factor $\prod_vz(v)^{\beta-1}$ of the Gibbs density from a uniform distribution, we use auxiliary variables: for fixed $z$, they range over a box of volume $\prod_vz(v)^{\beta-1}$. When $\beta>1$, we introduce $\beta-1$ auxiliary variables $y(v,1),\ldots,y(v,\beta-1)$ for each $v\in E$, and write $y\in\mathbb R^{E\times[\beta-1]}$ for the vector of all of them. We consider points $(z,y)$ with $z\in\mathbb R^E$ and $y\in\mathbb R^{E\times[\beta-1]}$, viewed as vectors in $\mathbb R^D$ with $D=d\beta$, and call $z$ and $y$ the \emph{$z$-part} and the \emph{$y$-part} of $(z,y)$. Consider the body
\begin{equation}\label{eq:gibbs-lifted-body}
    \mathcal K=\bigl\{(z,y)\in\mathbb R_{\ge0}^{D}:\ H(z)\le1\ \text{and}\ y(v,j)\le z(v)\ \text{for all}\ v\in E\ \text{and}\ j\in[\beta-1]\bigr\}.
\end{equation}
    For $\beta=1$, there are no auxiliary variables, and $\mathcal K=\{z\in\mathbb R_{\ge0}^E:H(z)\le1\}$. Let $b=(b_z,b_y)\in\mathbb R^D$ be the point with $b_z(v)=1/(2d)$ for every $v\in E$ and $b_y(v,j)=1/(4d)$ for every auxiliary variable $y(v,j)$. Let $U\in\mathbb R_{\ge0}^E$ be the $z$-part of an exactly uniform point of $\mathcal K$.

\begin{claim}\label{clm:gibbs-body}
    The set $\mathcal K$ is convex and satisfies $\frac1{8d}B_2^D\subseteq\mathcal K-b\subseteq\sqrt D\,B_2^D$. Moreover, $s_i(h)=(d+1/\beta)\,\mathbb E\,H_i(U)$ for every $i\in[m]$.
\end{claim}

\begin{proof}
    The body is convex, since $H$ is convex and the other constraints are linear. By~\eqref{eq:gibbs-rounding}, $\mathcal K\subseteq[0,1]^D$. If $\|(z,y)-b\|_2\le1/(8d)$, then every entry of $(z,y)$ differs from the corresponding entry of $b$ by at most $1/(8d)$. Hence all entries are nonnegative, $z(v)\ge3/(8d)$, and $y(v,j)\le3/(8d)$. Moreover, by the Cauchy--Schwarz inequality,
\[
    \sum_{v\in E}z(v)\le\sum_{v\in E}b_z(v)+\sqrt d\,\|z-b_z\|_2\le\frac12+\frac1{8\sqrt d}<1.
\]
    In particular, $y(v,j)\le z(v)$ and $H(z)<1$ by~\eqref{eq:gibbs-rounding}. Thus the ball $\{(z,y)\in\mathbb R^D:\|(z,y)-b\|_2\le1/(8d)\}$ lies in $\mathcal K$. Since both $b$ and $\mathcal K$ lie in $[0,1]^D$, every entry of a point of $\mathcal K-b$ lies in $[-1,1]$, so its Euclidean norm is at most $\sqrt D$. Therefore
\[
    \frac1{8d}B_2^D\subseteq\mathcal K-b\subseteq\sqrt D B_2^D,
\]
    where $B_2^D=\{x\in\mathbb R^D:\|x\|_2\le1\}$.

    Integrating out $y$ shows that $U$ has density proportional to $\mathbf 1_{\{z\in\mathbb R_{\ge0}^E:H(z)\le1\}}\,\omega(z)$, where $\omega(z):=\prod_vz(v)^{\beta-1}$.

    Under the change of variables $z(v)=c_vx(v)$, the Gibbs distribution $\mu_h$ becomes the distribution with density proportional to $e^{-\beta H(z)}\omega(z)$ on $\mathbb R_{\ge0}^E$, and $s_i(h)=\mathbb E\,H_i(Z)$ for $Z$ with this density. Put $q=\beta d$. For $t>0$, the substitution $z=tz'$ and the positive homogeneity of $H$ and $H_i$ give
\[
    \int_{H(z)\le t}\omega(z)\,dz=t^q\int_{H(z)\le1}\omega(z)\,dz,\qquad
    \int_{H(z)\le t}H_i(z)\omega(z)\,dz=t^{q+1}\int_{H(z)\le1}H_i(z)\omega(z)\,dz.
\]
    Therefore,
\begin{equation}\label{eq:gibbs-body-share}
\begin{aligned}
    s_i(h)
    &=\frac{\int H_i(z)\,e^{-\beta H(z)}\omega(z)\,dz}{\int e^{-\beta H(z)}\omega(z)\,dz}
     =\frac{\int_0^\infty\beta e^{-\beta t}\int_{H(z)\le t}H_i(z)\omega(z)\,dz\,dt}
           {\int_0^\infty\beta e^{-\beta t}\int_{H(z)\le t}\omega(z)\,dz\,dt}\\
    &=\frac{\int_0^\infty e^{-\beta t}t^{q+1}\,dt}{\int_0^\infty e^{-\beta t}t^{q}\,dt}\cdot
      \frac{\int_{H(z)\le1}H_i(z)\omega(z)\,dz}{\int_{H(z)\le1}\omega(z)\,dz}\\
    &=\frac{q+1}{\beta}\,\mathbb E\,H_i(U),
\end{aligned}
\end{equation}
    where all integrals in $z$ are over $\mathbb R_{\ge0}^E$. The second equality follows by writing $e^{-\beta H(z)}=\int_{H(z)}^\infty\beta e^{-\beta t}\,dt$ and exchanging the order of integration by Tonelli's theorem, and the third follows from the two identities above. The fourth uses $\int_0^\infty e^{-\beta t}t^s\,dt=\Gamma(s+1)/\beta^{s+1}$ and the density of $U$. Since $q=\beta d$, we have $(q+1)/\beta=d+1/\beta$.
\end{proof}

    By Claim~\ref{clm:gibbs-body}, to estimate $s_i(h)$, it suffices to estimate $\mathbb E\,H_i(U)$, for which samples of $U$ suffice. Theorem~\ref{thm:jlls-sampler} provides such samples up to a small total variation distance. Next, we show that averaging $H_i$ over polynomially many samples estimates $\mathbb E\,H_i(U)$ within a constant factor, simultaneously for all $i$.

\begin{claim}\label{clm:gibbs-estimates}
    Let $\delta'\in(0,1)$, $K=\lceil27d(D+1)\log(4m/\delta')\rceil$, and $\eta=\delta'/(2K)$. Let $U^{(1)},\ldots,U^{(K)}$ be the $z$-parts of independent random points of $\mathcal K$, each with distribution within total variation distance $\eta$ of the uniform distribution on $\mathcal K$. Then, with probability at least $1-\delta'$, the estimates
\[
    \bar s_i(h):=\frac{3(d+1/\beta)}{2K}\sum_{r=1}^K H_i(U^{(r)})
\]
    satisfy
\begin{equation}\label{eq:gibbs-empirical-shares}
    s_i(h)\le\bar s_i(h)\le2s_i(h)\qquad(i\in[m]).
\end{equation}
\end{claim}

\begin{proof}
    For each $v$, the point $e_v\in\mathbb R^D$ whose $z$-part is $\mathbf 1_{\{v\}}$ and whose $y$-part is zero belongs to $\mathcal K$, since $H(\mathbf 1_{\{v\}})=1$. Convexity implies $(1-t)\mathcal K+te_v\subseteq\mathcal K\cap\{z(v)\ge t\}$ for $0\le t\le1$. This affine image has volume $(1-t)^D\operatorname{vol}(\mathcal K)$, so
\[
    \Pr[U(v)\ge t]\ge(1-t)^D,
    \qquad \mathbb E U(v)\ge\int_0^1(1-t)^D\,dt=\frac1{D+1}.
\]
    Recall that $c_v=h(\{v\})$. Write $a_{iv}=h_i(\{v\})/c_v$ and $\kappa=d+1/\beta$. Monotonicity and the argument for~\eqref{eq:gibbs-rounding}, applied to $\widehat h_i$, give $a_{iv}U(v)\le H_i(U)\le\sum_v a_{iv}U(v)$. By Claim~\ref{clm:gibbs-body}, the random variable $Y_i=\kappa H_i(U)$ has mean $s_i(h)$. Since $U(v)\le H(U)\le1$ by~\eqref{eq:gibbs-rounding}, the upper bound gives $0\le Y_i\le B_i:=\kappa\sum_va_{iv}$. The lower bound and $\mathbb E\,U(v)\ge1/(D+1)$ give $s_i(h)=\kappa\,\mathbb E\,H_i(U)\ge\kappa a_{iv}/(D+1)$ for every $v$. Since $B_i\le\kappa d\max_va_{iv}$, it follows that
\[
    B_i\le d(D+1)\,s_i(h).
\]
    Let $U^{(1)},\ldots,U^{(K)}$ be independent exact samples of $U$, and put $Y_i^{(r)}=\kappa H_i(U^{(r)})$. If $B_i=0$, then $h_i$ is identically zero and needs no estimation. Otherwise apply the multiplicative Chernoff bound to the independent random variables $Y_i^{(r)}/B_i\in[0,1]$. Explicitly, for a random variable $W\in[0,1]$, convexity of the exponential gives $\mathbb E e^{\lambda W}\le\exp((e^\lambda-1)\mathbb EW)$ for every real $\lambda$; applying Markov's inequality to the product of these bounds with positive and negative $\lambda$ gives the usual upper and lower tails. At relative error $1/3$, they imply
\[
    \Pr\left[\left|\frac1K\sum_{r=1}^K Y_i^{(r)}-s_i(h)\right|>
                       \frac{s_i(h)}3\right]
    \le2\exp\left(-\frac{K s_i(h)}{27B_i}\right)
    \le2\exp\left(-\frac K{27d(D+1)}\right).
\]
    For $K$ as in the claim, the union bound makes the probability of any relative-estimation failure at most $\delta'/2$ for exact samples. Now let $U^{(1)},\ldots,U^{(K)}$ be the approximate samples of the claim. Their joint law differs from that of independent exact samples by at most $K\eta=\delta'/2$ (couple each pair of samples and take a union bound). Since $\bar s_i(h)=\frac{3}{2K}\sum_{r=1}^K\kappa H_i(U^{(r)})$, this proves~\eqref{eq:gibbs-empirical-shares} with probability at least $1-\delta'$.
\end{proof}

Now we can prove Lemma~\ref{lem:gibbs-share-estimation}.

\begin{proofof}{Lemma~\ref{lem:gibbs-share-estimation}}
    Let $h_1,\ldots,h_m$ be a family as above. If $d=0$, take all estimates to be zero. Otherwise, let $K$ and $\eta$ be as in Claim~\ref{clm:gibbs-estimates}, and apply Theorem~\ref{thm:jlls-sampler} to $\mathcal K-b$ with $N=D$, $r=1/(8d)$, and $R=\sqrt D$, which is allowed by Claim~\ref{clm:gibbs-body}, and translate the sampled point by $b$. This produces a sample within total variation distance $\eta$ of the uniform distribution on $\mathcal K$, using a number of membership queries and additional time polynomial in $D$ and $\log(dD/\eta)$. Run the sampler with fresh randomness for each sample. Any output outside $\mathcal K$ is replaced by $b$; this cannot increase its distance from the target law, which is supported on $\mathcal K$. With $K$ such samples, Claim~\ref{clm:gibbs-estimates} shows that the estimates $\bar s_i(h)$ satisfy~\eqref{eq:gibbs-empirical-shares} with probability at least $1-\delta'$.

    To bound the total score for every outcome, set
\[
    \widetilde s_i(h):=
    \begin{cases}
        \bar s_i(h),&\sum_j\bar s_j(h)\le2d,\\
        0,&\text{otherwise}.
    \end{cases}
\]
    The estimates are nonnegative. By Lemma~\ref{lem:gibbs-share-identity}, $\sum_i s_i(h)=d$, so the first case holds on~\eqref{eq:gibbs-empirical-shares}. Consequently,
\begin{equation}\label{eq:gibbs-estimated-shares}
    s_i(h)\le\widetilde s_i(h)\le2s_i(h)\quad(i\in[m])
\end{equation}
    on that event, while $\sum_i\widetilde s_i(h)\le2d$ holds unconditionally.

    Apply the above to the family $(f_i^\to(u,\cdot))_i$ defining $H_u$, for which $s_i(h)=s_i^{(u)}$ and $d\le n-1$, and let $\tilde s_i^{(u)}:=\widetilde s_i(h)$. By~\eqref{eq:gibbs-estimated-shares}, with probability at least $1-\delta'$, $\tilde s_i^{(u)}\ge s_i^{(u)}$ for every $i$, and $\sum_i\tilde s_i^{(u)}\le2d\le2(n-1)$ always. The family $(f_i^\to(\cdot,u))_i$ defining $G_u$ gives $\tilde t_i^{(u)}$ in the same way.

    It remains to bound the running time. Each value $h_i(C)$ is a directed block minimum, computed by one submodular minimization. The effective coordinates and the numbers $c_v$ are obtained from the singleton values, and by the greedy formula~\eqref{eq:Lovasz_ext_per}, each evaluation of $H$ or of some $H_i$ uses $O(dm)$ such values. Hence the membership oracle of $\mathcal K$, which checks nonnegativity, $y(v,j)\le z(v)$, and $H(z)\le1$, takes polynomial time. Since $D=d\beta=O(n)$ for fixed $\beta$, Claim~\ref{clm:gibbs-body} and Theorem~\ref{thm:jlls-sampler} show that each sample takes time polynomial in $n$ and $\log(1/\eta)$. Moreover, $K=O(n^2\log(m/\delta'))$ and $\log(1/\eta)=\log(2K/\delta')=O(\log(nm/\delta'))$, and computing the estimates from the samples takes $K$ evaluations of each $H_i$. Thus the total time is polynomial in $n$, $m$, and $\log(1/\delta')$.
\end{proofof}

\subsection{The case \texorpdfstring{$\beta=1$}{beta=1}}\label{sec:gibbs-tau-one}

    Sections~\ref{sec:gibbs-shares}--\ref{sec:gibbs-budget} apply to every inverse temperature $\beta\ge1$. In this subsection, we take a closer look at the case $\beta=1$, where $\mu_h$ has density proportional to $e^{-\widehat h(x)}$ on the nonnegative orthant, as in~\cite{JLLS23} and Section~\ref{sec:overview}. Without the factor $\prod_vx(v)^{\beta-1}$, every marginal of $\mu_h$ has a positive density at zero,
\[
    \rho_v(0)=\frac1{Z_h}\int_{(0,\infty)^{E\setminus\{v\}}}e^{-\widehat h(0,z)}\,dz
    \qquad(v\in E),
\]
    where $(0,z)$ is the vector whose $v$th coordinate is zero and whose other coordinates are given by $z$. This density gives the influence a concrete meaning: it measures how strongly $h$ pushes the coordinate $x(v)$ towards zero, and for a linear function $h(C)=\sum_{v\in C}\alpha_v$, both $\rho_v(0)$ and $\pi_h(v)$ equal $\alpha_v$.

    We show that, for $\beta=1$, the density at zero equals the averaged greedy marginal $a_h(v)$ of Section~\ref{sec:gibbs-budget} and is within a factor two of the influence $\pi_h(v)$. For general $\beta$, Section~\ref{sec:gibbs-budget} only bounds the influence from above by $1+1/\beta$ times the averaged greedy marginal; for $\beta=1$, the three quantities agree up to a factor two. In particular, the total influence is at least $h(E)$, so the budget of Lemma~\ref{lem:gibbs-influence-budget} is tight up to a factor two. 

    Throughout this subsection, $\beta=1$. We keep the notation of Sections~\ref{sec:gibbs-shares} and~\ref{sec:gibbs-budget} and, as there, assume that all coordinates are effective.

\paragraph{Densities at zero.}
    For $v\in E$, write $(t,z)$ for the vector whose $v$th coordinate is $t\ge0$ and whose other coordinates are given by $z\in(0,\infty)^{E\setminus\{v\}}$. The $v$th marginal of $\mu_h$ has density
\[
    \rho_v(t)=\frac1{Z_h}\int_{(0,\infty)^{E\setminus\{v\}}}
    e^{-\widehat h(t,z)}\,dz
    \qquad(t\ge0),
\]
    which is continuous by dominated convergence, and nonincreasing because $\widehat h$ is coordinatewise nondecreasing. Here $\rho_v(0)=\lim_{t\downarrow0}\rho_v(t)$ is the value of this continuous version; changing the joint density on the boundary, as in Section~\ref{sec:gibbs-budget}, does not affect it. Recall from Section~\ref{sec:gibbs-budget} the greedy marginals $b_v(x)$, that is, the marginal value of $v$ when the coordinates are added in decreasing order of $x$, and the averaged greedy marginals $a_h(v)=\mathbb E_{\mu_h}b_v(X)$. The first result identifies the densities at zero with the averaged greedy marginals.

\begin{lemma}\label{lem:gibbs-density-greedy}
    Let $\beta=1$. For every $v\in E$, $\rho_v(0)=a_h(v)$. Consequently, $\sum_{v\in C}\rho_v(0)\le h(C)$ for every $C\subseteq E$, and $\sum_{v\in E}\rho_v(0)=h(E)$.
\end{lemma}

\begin{proof}
    Fix $z\in(0,\infty)^{E\setminus\{v\}}$ with pairwise distinct coordinates; such vectors have full measure. Put $\varphi(t)=\widehat h(t,z)$ for $t\ge0$. The function $\varphi$ is continuous. By the greedy formula~\eqref{eq:Lovasz_ext_per}, $\widehat h(x)=\sum_{w\in E}x(w)\,b_w(x)$. On an open interval of $t$ containing none of the values $z(w)$, the order of the coordinates of $(t,z)$ does not change, so every $b_w(t,z)$ is constant there, and $\varphi$ is linear with slope $b_v(t,z)$. Hence $e^{-\varphi}$ is absolutely continuous on compact intervals, with derivative $-b_v(t,z)e^{-\varphi(t)}$ almost everywhere. Moreover, $\varphi(t)\ge h(\{v\})t$ by the singleton bounds following~\eqref{eq:gibbs-distribution}, so $e^{-\varphi(t)}\to0$ as $t\to\infty$. Therefore
\[
    \int_0^\infty b_v(t,z)\,e^{-\widehat h(t,z)}\,dt=e^{-\widehat h(0,z)}.
\]
    Integrate over $z$ by Tonelli's theorem and divide by $Z_h$. The left-hand side becomes $\mathbb E_{\mu_h}b_v(X)=a_h(v)$, and the right-hand side becomes $\rho_v(0)$. The remaining claims follow from~\eqref{eq:gibbs-share-price-budgets}.
\end{proof}

    In other words, the vector $(\rho_v(0))_{v\in E}$ lies in the base polytope $\mathcal B(h)$ defined in Section~\ref{sec:preliminaries}: the densities at zero split the value $h(E)$ among the coordinates, and no set $C$ receives more than $h(C)$. For $u\ne v$, let $q_v^{(u)}(0)$ and $\bar q_v^{(u)}(0)$ be the densities at zero of the $v$th marginals of $\mu_{H_u}$ and $\mu_{G_u}$, respectively, set to zero for ineffective coordinates. Applied to the family defining $H_u$, Lemma~\ref{lem:gibbs-density-greedy} gives $\sum_{v}q_v^{(u)}(0)=H_u(V\setminus\{u\})$.

\paragraph{Densities at zero and influences.}
    We next compare the density at zero with the influence $\pi_h(v)=1/\mathbb E_{\mu_h}X(v)$. Since $\rho_v$ is nonincreasing, $X(v)$ is unlikely to be much smaller than $1/\rho_v(0)$; log-concavity prevents it from being much larger on average.

\begin{lemma}\label{lem:gibbs-density-influence}
    Let $\beta=1$. For every $v\in E$,
\begin{equation}\label{eq:gibbs-density-influence}
    \rho_v(0)\le\pi_h(v)\le2\rho_v(0).
\end{equation}
\end{lemma}

\begin{proof}
    Note that $\rho_v(0)>0$, since otherwise the nonincreasing density $\rho_v$ would vanish identically.

    For the right inequality, the monotonicity of $\rho_v$ gives $\Pr_{\mu_h}[X(v)\le t]\le t\rho_v(0)$ for $t\ge0$. Hence
\[
    \mathbb E_{\mu_h}X(v)
    =\int_0^\infty\Pr_{\mu_h}[X(v)>t]\,dt
    \ge\int_0^{1/\rho_v(0)}\bigl(1-t\rho_v(0)\bigr)\,dt
    =\frac1{2\rho_v(0)}.
\]

    For the left inequality, the density of $\mu_h$, extended by zero outside the nonnegative orthant, is log-concave because $\widehat h$ is convex. By Pr\'ekopa's theorem~\cite[Theorem~8]{Prekopa73}, the marginal $\rho_v$ is log-concave. Its tail $T(t)=\Pr_{\mu_h}[X(v)>t]=\int_0^\infty\rho_v(t+s)\,ds$ is also log-concave by~\cite[Theorem~6]{Prekopa73}, applied to the log-concave function $(t,s)\mapsto\rho_v(t+s)$, with $\rho_v$ extended by zero to negative arguments, and integration over the convex set $[0,\infty)$. Since $T(0)=1$ and the right derivative of $\log T$ at zero is $-\rho_v(0)$, concavity gives $T(t)\le e^{-\rho_v(0)t}$. Integrating, $\mathbb E_{\mu_h}X(v)\le1/\rho_v(0)$.
\end{proof}

    Lemma~\ref{lem:gibbs-density-influence} has two consequences for a single Gibbs distribution. First, combined with Lemmas~\ref{lem:gibbs-density-greedy} and~\ref{lem:gibbs-influence-budget}, it gives
\begin{equation}\label{eq:gibbs-total-influence}
    h(E)\le\sum_{v\in E}\pi_h(v)\le2h(E).
\end{equation}
    Lemma~\ref{lem:gibbs-influence-budget} gives the upper bound for every $\beta$; the lower bound shows that, for $\beta=1$, it is tight up to a factor two. As a sanity check, for a linear function $h(C)=\sum_{v\in C}\alpha_v$, the distribution $\mu_h$ is the product of exponential distributions with rates $\alpha_v$, so $\rho_v(0)=\pi_h(v)=\alpha_v$; the left inequality of~\eqref{eq:gibbs-density-influence} and the lower bound in~\eqref{eq:gibbs-total-influence} hold with equality. Second, the right inequality of~\eqref{eq:gibbs-density-influence} turns Lemma~\ref{lem:gibbs-low-price} into a bound by densities at zero: $h_i(C)\le2s_i(h)\max_{v\in C}\rho_v(0)$ for every $C\subseteq E$. For the families defining $H_u$, combining the right inequality with the first inequality of~\eqref{eq:gibbs-terminal-local-bound} gives
\[
    f_i^\to(A,D)
    \le2\sum_{u\in A}s_i^{(u)}\max_{v\in D}q^{(u)}_v(0)
    \qquad(A\cap D=\varnothing).
\]

    Finally, we return to a query $S$ with $F(S)>0$. Replacing the influences in~\eqref{eq:gibbs-query-prices} by densities at zero, put
\[
    q_S(v):=
    \begin{cases}
       \displaystyle\max_{u\in\bar S}\bar q_v^{(u)}(0),&v\in S,\\[5pt]
       \displaystyle\max_{u\in S}q_v^{(u)}(0),&v\in\bar S.
    \end{cases}
\]
    By Lemma~\ref{lem:gibbs-density-influence}, $q_S(v)\le\inf_S(v)\le2q_S(v)$ for every $v\in V$, so Lemmas~\ref{lem:gibbs-price-certificate} and~\ref{lem:gibbs-query-prices}, on which the chaining relies, transfer to $q_S$. By the right inequality, $\{v:q_S(v)\le\theta\}\subseteq U_{2\theta}$. Since $\Delta_i(S,U)$ is nondecreasing in $U$, Lemma~\ref{lem:gibbs-price-certificate} gives
\[
    \Delta_i\bigl(S,\{v:q_S(v)\le\theta\}\bigr)\le2\theta p_i
    \qquad(\theta\ge0,\ i\in[m]).
\]
    By the left inequality and the one-sided bounds~\eqref{eq:gibbs-query-price-sides} of Lemma~\ref{lem:gibbs-query-prices},
\begin{equation}\label{eq:gibbs-overview-split-budget}
    \sum_{v\in\bar S}\max_{u\in S}q_v^{(u)}(0)\le2F(S),
    \qquad
    \sum_{v\in S}\max_{u\in\bar S}\bar q_v^{(u)}(0)\le2F(S).
\end{equation}
    Hence the proof of Theorem~\ref{thm:gibbs-near-quadratic} goes through with $q_S$ in place of $\inf_S$; the factor two only changes the constant $16$ in~\eqref{eq:gibbs-chain-budget} to $32$, which the choice of $\Lambda$ still accommodates.

\section{Preprocess}\label{sec:preprocess}

The following reduction allows us to assume $f_i(\varnothing)=f_i(V)=0$.

\begin{lemma}\label{lem:endpoint-normalization}
    Let $f_1,\ldots,f_m:2^V\to\mathbb R_{\ge0}$ be submodular, with $n=|V|\ge1$. On $V^+=V\cup\{a,b\}$, define $g_1,\ldots,g_m:2^{V^+}\to\mathbb R_{\ge0}$ by
\begin{equation}\label{eq:preprocess-extension}
    g_i(T)=\begin{cases}
        0, & T=\varnothing\text{ or }T=V^+,\\
        f_i(T\cap V), & \text{otherwise}.
    \end{cases}
\end{equation}
    Then each $g_i$ is submodular and can be evaluated with at most one query to $f_i$. An $\varepsilon$-sparsifier of $G=\sum_i g_i$ is also a sparsifier of $F=\sum_i f_i$ with the same weights and support. For queries $|S|\le k$, it suffices to preserve $G$ on sets of size at most $k+1$.
\end{lemma}

\begin{proofof}{Lemma~\ref{lem:endpoint-normalization}}
    The function $T\mapsto f_i(T\cap V)$ is submodular. Setting its values at $\varnothing$ and $V^+$ to zero preserves submodularity: for incomparable sets, only the right-hand side of the submodular inequality ($g_i(A)+g_i(B)\ge g_i(A\cap B)+g_i(A\cup B)$) can decrease; for comparable sets, the inequality is an equality. Nonnegativity is immediate. For $S\subseteq V$, since $S\cup\{a\}$ is neither endpoint, we have $g_i(S\cup\{a\})=f_i(S)$ and, in the block-minimum notation,
\begin{equation}\label{eq:preprocess-block-minimum}
    g_i^{\to}(S\cup\{a\},(V\setminus S)\cup\{b\})
    =g_i(S\cup\{a\})=f_i(S).
\end{equation}
    Thus the same reweighting preserves all corresponding queries. Moreover, $|S\cup\{a\}|=|S|+1$, and $(n+2)(k+1)=O(nk)$ for $1\le k\le n$, so the asymptotic support bounds are unchanged.
\end{proofof}

\section{Omitted Proofs}\label{appendix:proofs}
We finish the three proofs omitted in Section~\ref{sec:preliminaries}. The first proof is about the monotonicity and submodularity of $f^{\to}(A,B)$.

\begin{proofof}{Lemma~\ref{lem:block-minimum-separate-submodularity}}
    Adding constraints can only increase the minimum. Thus $f^\to$ is monotone non-decreasing in either block.

    Fix $A \subseteq [n]$ and let $B_1,B_2 \subseteq [n]\setminus A$. Let $T_1$ and $T_2$ attain the minima in $f^\to(A,B_1)$ and $f^\to(A,B_2)$, respectively. Since $A\subseteq T_1\cap T_2$ and $(T_1\cap T_2)\cap(B_1\cup B_2)=\varnothing$, $T_1\cap T_2$ is feasible for $f^\to(A,B_1\cup B_2)$. Similarly, $T_1\cup T_2$ is feasible for $f^\to(A,B_1\cap B_2)$. By the submodularity of $f$,
    \begin{align*}
        f^\to(A,B_1)+f^\to(A,B_2)
        &= f(T_1)+f(T_2)\\
        &\ge f(T_1\cap T_2)+f(T_1\cup T_2)\\
        &\ge f^\to(A,B_1\cup B_2)+f^\to(A,B_1\cap B_2).
    \end{align*}
    Thus $f^\to(A,\cdot)$ is submodular.

    Similarly, fix $B \subseteq [n]$ and let $A_1,A_2 \subseteq [n]\setminus B$. Let $T_1$ and $T_2$ attain the minima in $f^\to(A_1,B)$ and $f^\to(A_2,B)$, respectively. Since $A_1\cup A_2\subseteq T_1\cup T_2$ and $(T_1\cup T_2)\cap B=\varnothing$, $T_1\cup T_2$ is feasible for $f^\to(A_1\cup A_2,B)$. Similarly, $T_1\cap T_2$ is feasible for $f^\to(A_1\cap A_2,B)$. Applying submodularity again gives
    \[
        f^\to(A_1,B)+f^\to(A_2,B)
        \ge f^\to(A_1\cup A_2,B)+f^\to(A_1\cap A_2,B).
    \]
    Hence $f^\to(\cdot,B)$ is also submodular.
\end{proofof}

The second proof bounds $\Delta(S,U):=f(S)-f^{\to}(S\setminus U,\bar S\setminus U)$.

\begin{proofof}{Lemma~\ref{lem:block-residual-two-block}}
    Let $T_1$ and $T_2$ attain the minima in $f^\to(S\setminus T,\bar S)$ and $f^\to(S\cap T,T\setminus S)$, respectively. By definition, $S\setminus T\subseteq T_1\subseteq S$ and $T\cap T_2=S\cap T$. Hence $T_1\cup(T\cap T_2)=S$. By submodularity and nonnegativity,
    \begin{align*}
        f(S)
        &=f(T_1\cup(T\cap T_2))\\
        &\le f(T_1)+f(T\cap T_2)\\
        &\le f(T_1)+f(T)+f(T_2).
    \end{align*}
    Rearranging proves the first inequality.

    We now consider the residual $\Delta(S,U)$. Let $T$ attain the minimum in $f^\to(S\setminus U,\bar S\setminus U)$. Since $S\setminus T\subseteq S\cap U$ and $T\setminus S\subseteq U\setminus S$, the first inequality and the monotonicity of $f^\to$ give
    \begin{align*}
        \Delta(S,U)
        &=f(S)-f(T)\\
        &\le f^\to(S\setminus T,\bar S)+f^\to(S\cap T,T\setminus S)\\
        &\le f^\to(S\cap U,\bar S)+f^\to(S,U\setminus S).
    \end{align*}
    Moreover, $S$ is feasible for $f^\to(S\setminus U,\bar S\setminus U)$, so $\Delta(S,U)\ge 0$.

    Enlarging $U$ relaxes the constraints in the definition of $f^\to(S\setminus U,\bar S\setminus U)$. Thus $\Delta(S,U)$ is monotone non-decreasing in $U$.

    Finally, $\Delta(S,\varnothing)=f(S)-f^\to(S,\bar S)=0$. Since $f(\varnothing)=0$ and $f$ is nonnegative, we have $f^\to(\varnothing,\varnothing)=0$, which gives $\Delta(S,V)=f(S)$.
\end{proofof}

\subsection{Proof of Lemma~\ref{lem:Bernstein}}\label{sec:proof_Bernstein}

\begin{proof}
    If $R=0$, then all $z_i=0$ and the claim is immediate. We henceforth assume $R>0$.

    The random variables $Y_i$ are independent and have mean zero. If $p_i\Lambda\ge1$, then $Y_i=0$. Otherwise, $f_i$ is sampled with probability $p_i\Lambda$ such that
    \[
        |Y_i|\le\frac{z_i}{p_i\Lambda}\le\frac{R}{\Lambda} \qquad \text{ and }
        \qquad
        \E[Y_i^2]=\frac{1-p_i\Lambda}{p_i\Lambda}z_i^2\le\frac{R}{\Lambda}z_i.
    \]

    $0<\eta\le1$ implies $|\eta\Lambda Y_i/R|\le1$. Then we apply $e^u\le1+u+u^2$ for $|u|\le1$ to
    \begin{align*}
        \E\left[\exp\left(\pm\frac{\eta\Lambda}{R}\sum_iY_i\right)\right]
        &=\prod_i\E\left[\exp\left(\pm\frac{\eta\Lambda}{R}Y_i\right)\right]\\
        &\le\prod_i\left(1+\frac{\eta^2\Lambda^2}{R^2}\E[Y_i^2]\right)   \le\exp\left(\frac{\eta^2\Lambda}{R}\sum_i z_i\right).
    \end{align*}

    Applying Markov's inequality to either sign gives
    \begin{align*}
        &\Pr\left[\pm\sum_iY_i>\eta\sum_i z_i+\frac{xR}{\eta\Lambda}\right]\\
        &\qquad\le\exp\left(-\frac{\eta\Lambda}{R}\left(\eta\sum_i z_i+\frac{xR}{\eta\Lambda}\right)+\frac{\eta^2\Lambda}{R}\sum_i z_i\right) =e^{-x}.
    \end{align*}
    Taking a union bound completes the proof.
\end{proof}

\bibliographystyle{alpha}
\bibliography{plan}


\section{AI usage statement}
The authors used Gemini 3.1 in the early stage of this work to explore proof strategies, search for related literature, and assist with verification. In the first version of this work, they found the key definition of block min-cut functions $f^{\to}(A,B)$ and used it to approximate the sensitivity under cardinality constraints. After finding the counter-examples in Claim~\ref{clm:counter_examples}, the authors realized the idea of using $2n$ Gibbs distributions to generate weights $w_i$. Then the authors used ChatGPT 5.6 sol and 6-pro to assist with proving Theorem~\ref{inform:sparsification_all_subsets} and Theorem~\ref{inform:sparsification_cardinality_constraint}. 

The authors take full
responsibility for all content and arguments in this paper and for their
correctness.

\end{document}